\documentclass[a4paper,10pt]{article}
\usepackage{geometry}
\usepackage{fancyhdr}
\usepackage{graphicx}
\usepackage{amssymb} 
\usepackage{epstopdf}
\usepackage[ruled,norelsize]{algorithm2e}
\makeatletter
\newcommand{\removelatexerror}{\let\@latex@error\@gobble}
\makeatother
\usepackage[caption=false,font=footnotesize]{subfig}
\usepackage{pseudocode}
\usepackage{fancybox}
\usepackage{multicol}
\usepackage{verbatim}
\usepackage{amsthm}
\usepackage{amsmath}
\usepackage{color}
\usepackage{bm}
\usepackage{cite}
\usepackage{marginnote}
\usepackage{caption,lipsum}

\usepackage{multirow}
\usepackage{tikz}
\usepackage{hyperref}
\usepackage{booktabs}
\usepackage{marginnote}
\theoremstyle{plain}
\newtheorem{thm}{Theorem}

\newtheorem{cor}[thm]{Corollary}
\newtheorem{lem}[thm]{Lemma}
\newtheorem{prop}[thm]{Proposition}
\newtheorem{defn}[thm]{Definition}

\DeclareMathOperator{\rank}{rank}

\makeatletter

\newcommand{\Rmnum}[1]{\expandafter\@slowromancap\romannumeral #1@}

\begin{document}
\author{Qing Zou\footnote{Applied Mathematics and Computational Sciences, University of Iowa, Iowa City, IA 52242. (zou-qing@uiowa.edu)}\quad and\quad Mathews Jacob\footnote{Department of Electrical and Computer Engineering, University of Iowa, Iowa City, IA 52242. (mathews-jacob@uiowa.edu)}}

\title{\LARGE{Recovery of surfaces and functions in high dimensions: sampling theory and links to neural networks\footnote{ Parts of this work were presented at SIAM conference on applied algebraic geometry (SIAM AG 2019), at the  IEEE International Symposium on Biomedical Imaging (ISBI 2020) and at the IEEE International Conference on Acoustics, Speech, and Signal Processing (ICASSP 2020).}}}
\date{}
\maketitle

{\bf{Abstract.}} Several imaging algorithms including patch-based image denoising, image time series recovery, and convolutional neural networks can be thought of as methods that exploit the manifold structure of signals. While the empirical performance of these algorithms is impressive, the understanding of recovery of the signals and functions that live on manifold is less understood. In this paper, we focus on the recovery of signals that live on a union of surfaces. In particular, we consider signals living on a union of smooth band-limited surfaces in high dimensions. We show that an exponential mapping transforms the data to a union of low-dimensional subspaces. Using this relation, we introduce a sampling theoretical framework for the recovery of smooth surfaces from few samples and the learning of functions living on smooth surfaces. The low-rank property of the features is used to determine the number of measurements needed to recover the surface. Moreover, the low-rank property of the features also provides an efficient approach, which resembles a neural network, for the local representation of multidimensional functions on the surface. The direct representation of such a  function in high dimensions often suffers from the curse of dimensionality; the large number of parameters would translate to the need for extensive training data. The low-rank property of the features can significantly reduce the number of parameters, which makes the computational structure attractive for learning and inference from limited labeled training data.

\vskip 10pt


\vskip 30pt

\section{Introduction}

Several imaging algorithms were introduced to exploit the extensive redundancy with images to recover them from noisy and possibly undersampled measurements. For instance, several patch-based image denoising methods were introduced in the recent past. Algorithms such as non-local means perform averaging of similar patches within the image to achieve denoising \cite{buades2011non}. Similar patch-based regularization strategies are used for image recovery from undersampled data \cite{mohsin2017accelerated,yang2012nonlocal}.  Similar approaches are also used for the recovery of images in a time series by exploiting their non-local similarity \cite{poddar2018recovery,poddar2015dynamic}. The success of these methods could be attributed to the manifold assumption \cite{fefferman2016testing,shao2018riemannian}, which states that signals in real-world datasets (e.g. patches in images) are restricted to smooth manifolds in high dimensional spaces. In particular, the regularization penalty used in non-local methods can be viewed as the energy of the signal gradients on the patch manifold rather than in the original domain, facilitating the collective recovery of the patch manifold from noisy measurements \cite{belkin2006manifold}. In particular, non-local methods estimate the interpatch weights, which are used for denoising; the interpatch weights are equivalent to the manifold Laplacian, which captures the structure of the manifold. Similarly, image denoising approaches such as BM3D \cite{bm3d} that cluster patches, followed by PCA approximations of the cluster, can also be viewed as modeling the tangent subspaces of the patch manifold in each neighborhood. Patch dictionary based schemes, which allow the coefficients to be adapted to the specific patch, could also be viewed as tangent subspace approximation methods.

Convolutional neural networks are now emerging as very powerful alternatives for image denoising \cite{zhang2017beyond,tian2019enhanced} and image recovery \cite{aggarwal2018modl,jackson2017large}. Rather than averaging similar patches, neural networks learn how to denoise the image neighborhoods from example pairs of noisy and noise-free patches. These frameworks can be viewed as learning a multidimensional function in high dimensional patch spaces. In particular, the inputs to the network are noisy patches and the corresponding outputs are the denoised patches/pixels. We note that the learning of such functions using conventional methods will suffer from the curse of dimensionality. Specifically, large amounts of training data may be needed to learn the parameters of such a high-dimensional function, if represented using conventional methods. While the empirical performance of neural networks is impressive, the mathematical understanding of why and how they can learn complex multidimensional functions in high-dimensional spaces from relatively limited training data is still emerging. We note that the manifold assumption is also used in the CNN literature to explain the good performance of neural networks.

With the goal of understanding the above algorithms from a geometrical perspective, we consider the following conceptual problems \textbf{(a)} when can we learn and recover a manifold or surface in high dimensional space from few samples or training data, \textbf{(b)} when can we exactly learn and recover a function that lives on a surface, from few input-output examples, \textbf{(c)} can these results explain the good performance of imaging algorithms that use manifold structure. We note that many different surface models including parametric shape models \cite{jacob2004efficient,li2011robust}, local and multi-resolution representations \cite{peyre2005surface,sajda2002multi}, and implicit level-set \cite{osher2001level,rousson2002shape,li2010distance} shape representations have been used in low-dimensional settings (e.g. 2D/3D). Our main focus in this paper is on surface recovery in high-dimensional spaces with application to machine learning and learning surfaces of patches and images. The utility of the above algorithms in such high dimensional applications have not been well-studied, to the best of our knowledge; the direct extension of the low-dimensional algorithms is expected to be associated with high computational complexity. We consider implicit level set representation of the surface to deal with shapes of arbitrary topology. Specifically, we model the surface or union of surfaces as the zero level set of a multidimensional function $\psi$. To restrict the degrees of freedom of the surface, we consider the level set function to be bandlimited. The bandwidth of $\psi$ can be viewed as a measure of the complexity of the surface; a more band-limited function will translate to a smoother surface. We refer the readers to our earlier works \cite{ongie,poddar2018recovery,poddaricassp,zou2019sampling,zou2020sampling} for examples of 2D/3D recovery of shapes, where the above representation is used to represent and recover shapes with sharp corners and edges.

We show that under the above assumptions on the surface, a non-linear mapping of the points on the surface will live on a low-dimensional feature subspace, whose dimension depends on the complexity of the surface. Specifically, one can transform each data point to a feature vector, whose size is equal to the number of basis functions used for the surface representation. Since we use a linear combination of complex exponentials to represent the surface, the lifting in our setting is an exponential mapping. We use the low-rank property of the feature matrix to estimate the surface from few of its samples. Our sampling results show that an irreducible surface can be perfectly recovered from very few samples, whose number is dependent on the bandwidth. Our experiments in these settings show the good recovery of the surfaces from few noisy points. Our results also show that the union of irreducible surfaces can also be recovered from few samples, provided each of the irreducible components are adequately sampled. We also use a kernel low-rank algorithm to recover a surface from its noisy samples, which bears close similarity with non-local means algorithms widely used in image processing. 

We also show that the low-rank property can be used to efficiently represent multidimensional functions of points living on the surface. In particular, we are only interested in the good representation of the function when the input is on or in the vicinity of the surface. We assume the functions are linear combination of the same basis functions (exponentials in our case). Since such representations are linear in the feature space, the low-rank nature of the exponential features provides an elegant approach to represent the function using considerably fewer parameters. In particular, we show that the feature vectors of a few anchor points on the surface span the space, which allows us to efficiently represent the function as the interpolation of the function values at the anchor points using a Dirichlet kernel. The significant reduction in the number of free parameters offered by this local representation makes the learning of the function from finite samples tractable. We note that the computational structure of the representation is essentially a one-layer kernel network. Note that the approximation is highly local; the true function and the local representation match only on the surface, while they may deviate significantly on points which are not on the surface.  We demonstrate the preliminary utility of this network in denoising, which shows improved performance compared to some state-of-the-art methods. Here, we model the denoiser as a function $f: \mathbb R^{p^2}\rightarrow \mathbb R$ that provides a \emph{noise-free} center pixel of a $p\times p$ noisy patch. The noisy patch is assumed to a point in $p^2$ dimensional space, close to the low-dimensional patch surface or union of surfaces. We also show that this framework can be used to learn a manifold, which can be viewed as the signal subspace version of the null-space based kernel low-rank algorithm considered above. In this case, the network structure is an auto-encoder. 

This work is related to kernel methods, which are widely used for the approximation of functions \cite{belkin2006manifold,cho2009kernel,mairal2014convolutional}. It is well-known that an arbitrary function can be approximated using kernel methods, and the computational structure resembles a single hidden layer neural network. Our work has two key distinctions with the above approaches: (a) unlike most kernel methods that choose infinite bandwidth kernels (e.g. Gaussians), we restrict our attention to a band-limited kernel. (b) We focus on a restrictive data model, where the data samples are localized or close to the zero set of a band-limited function. These two restrictions allow us to come up with theoretical results on when such a surface can be perfectly recovered from few samples. The results also provide clues on how many training data pairs are needed to learn functions on such surfaces. We note that such sampling theoretical results are not available in kernel literature, to the best of our knowledge. This work is inspired by the recent work on algebraic varieties \cite{ongie2017algebraic}, which also considers surfaces with finite degrees of freedom. The main distinction of this paper is the novel theoretical guarantees on recovery of the surface and functions living on the surface, which go beyond the empirical results in \cite{ongie2017algebraic}. We focus on bandlimited surfaces in this work to borrow the theoretical tools from \cite{poddar2018recovery,poddaricassp,zou2019sampling}. This work extends the results in \cite{poddar2018recovery,poddaricassp,zou2019sampling} in three important ways \textbf{(i).} The planar results are generalized to the high dimensional setting. \textbf{(ii).} The worst-case sampling conditions are replaced by high-probability results, which are far less conservative, and are in good agreement with experimental results.  \textbf{(iii).} The sampling results are extended to the local representation of functions. While we focus on bandlimited functions to come up with theoretical bounds, the results could be generalized to arbitrary surface representations including most basis functions, such as polynomial basis functions considered in \cite{ongie2017algebraic,tsakiris2017algebraic} and shift invariant representations \cite{wang2006radial, bernard2009variational}. We note that this work uses parametric level-set representations unlike non-parametric level-set models (e.g. \cite{rousson2002shape,li2010distance} for image segmentation). The narrow-band evolution used by these approaches to manage computational complexity makes these algorithms highly vulnerable to initial guess, unlike our algorithms as illustrated in \cite{zou2019sampling}. While we illustrate our algorithms in 2D/3D applications for visualization purposes,  we stress that our main focus is on high-dimensional ($\gg 3$) extensions of the level set approach and generalization to shape recovery. Non-parametric and even parametric level-set methods \cite{wang2006radial, bernard2009variational} will be associated with very high computational complexity in this setting without the proposed computational approaches, and has not been reported to the best of our knowledge.

\subsection{Terminology and Notation}

We introduce some commonly used terminologies and notations throughout the paper. We term the zero level set of a trigonometric polynomial as a surface. Usually, the lower-case Greek letters $\psi, \eta,$ etc. are used to represent the trigonometric polynomials. The calligraphic letter $\mathcal{S}$ or $\mathcal{S}[\psi]$ is used to represent the zero level sets of the trigonometric polynomials and hence the surfaces. The bold lower-case letters $\mathbf{x}$ denotes the real variable in $[0,1)^n$ and sometimes the points on the surface. The indexed  bold lower-case letters $\mathbf{x}_i$ represent the samples on the surface. The upper-case Greek letters $\Lambda, \Gamma \subset \mathbb{Z}^n$ are used to denote the bandwidth of the trigonometric polynomials. In other words, the upper-case Greek letters are the coefficients index set. The coefficients set is shown as $\{\mathbf{c_k} : \mathbf{k}\in\Lambda\}$. The cardinality of bandwidth $\Lambda$ is given by $|\Lambda|$, which will serve as a measure of the complexity of the surface. The notation $\Gamma\ominus\Lambda$ indicates the set of all the possible uniform shifts of the set $\Lambda$ within the set $\Gamma$. The specific Greek letter $\Phi$ (sometimes subscripts are used to identify the corresponding bandwidth) is used to represent the lifting map (feature map) of the point on the surface. The notation $\Phi(\mathbf{X})$ denotes the feature matrix of the sampling set $\mathbf{X}$.

	\subsection{Background on non-local means; reinterpretation as manifold regularization}
Non-local means (NLM) methods average patches in an image based on their similarity to obtain a denoised image. In particular, they compute a weight matrix, whose entries are $
\mathbf W_{i,j} = \exp\left(-\frac{\|\mathbf P_{\mathbf r_i}(f)-\mathbf P_{\mathbf r_j}(f)\|^2}{\sigma^2}\right)$, 
where $\mathbf{P}_{\mathbf{r}}(f)$ denotes a patch in the image $f$, centered at $\mathbf{r}$. The smoothing approach in NLM can be viewed as the minimization problem
\begin{equation}\label{manifold}
\{\mathbf f^*\} = \arg \min_f \|\mathbf f-\mathbf g\|^2 + \eta \sum_{i=1}^N \sum_{i=1}^{N} \mathbf W_{i,j} ~\|\mathbf P_{\mathbf r_i}(f)-\mathbf P_{\mathbf r_j}(f)\|^2.
\end{equation}
This optimization problem can be viewed as the discretization of the manifold smoothness regularization strategy used in machine learning \cite{belkin2006manifold}, which considers the recovery of a multidimensional function $\mathbf f(\mathbf s)$ on a manifold $\mathcal M$ from its noisy samples $\mathbf f(\mathbf s_k)=\mathbf y_k$:
\begin{equation}\label{key}
\{\mathbf F^{*}\} = \arg \min_f \|\mathbf f(s_k)-\mathbf y_k\|^2 + \eta \int_{\mathcal M}\|\nabla_{\mathcal M} \mathbf f\|^2 dx.
\end{equation}
Here, $\mathcal M$ is a smooth surface/manifold and $\nabla_{\mathcal M}$ denotes the  gradient of the function on the manifold. 
The weight matrix $\mathbf W$ captures the geometry of the patch manifold in \eqref{manifold}. Specifically, closer point pairs on $\mathcal M$ will have higher weights, while distant point pairs will have smaller weights. The equivalence with NLM can be seen by viewing the noisy patches as noisy samples $\mathbf y_k$ on the patch manifold. We note that the weighted sum is often expressed in a compact form as $$\sum_{i=1}^N \sum_{i=1}^{N} \mathbf W_{i,j} \|\mathbf f(x_i)-\mathbf f(x_j)\|^2 = {\rm trace}\left(\mathbf F \mathbf L \mathbf F^T\right).$$ Here, $\mathbf F= \begin{bmatrix}
\mathbf f_1&\ldots & \mathbf f_N
\end{bmatrix}$ and $\mathbf L$ is the Laplacian matrix $\mathbf L = \mathbf D-\mathbf W$, which captures the structure of the manifold and $\mathbf D$ is a diagonal matrix $\mathbf D = {\rm diag}(\sum_j \mathbf W_{i,j})$. $\mathbf L$ can be viewed as the discrete approximation of the Laplace-Beltrami operator on the continuous surface/manifold \cite{belkin2006manifold}.

\section{Parametric surface representation}\label{background}

In this work, we use the level set representation to describe a (hyper-)surface. We model a (hyper-)surface $\mathcal S$ in $[0,1)^n; n\ge 2$ as the zero level set of a function $\psi$:
\begin{equation}\label{zerolevel}
\mathcal{S}[\psi] = \{\mathbf{x}\in\mathbb{R}^n | \psi(\mathbf{x})=0\}.
\end{equation}
 For example, when $n=2$, $\mathcal{S}$ is a (hyper-)surface of dimension 1, which is typically a curve. We note that the level set representation is widely used in image segmentation \cite{li2010distance}. The normal practice in image segmentation is the non-parametric level set representation of a time-dependent evolution function $\psi$, which results in the PDE-driven models. Note that the initialization of these models affects the stability and the rate of convergence of the methods. So good initialization of level set functions is usually a requirement for good segmentation.

\label{parametric}
Several authors have recently proposed to represent the level set function $\psi$ as a linear combination of basis functions $\varphi_{\mathbf {k}}(\mathbf {x})$ \cite{wang2006radial, bernard2009variational}. These schemes argue that the reduced number of parameters translate to fast and efficient algorithms. Besides, we do not require the good initialization in this setting. Motivated by these schemes, we represent $\psi(\mathbf {x})$ as
 
 \begin{equation}\label{key}
 \psi(\mathbf {x})=\sum_{\mathbf {k}\in\Lambda}\mathbf {c_k}~\varphi_{\mathbf {k}}(\mathbf {x}).
 \end{equation}
Since the level set function is the linear combination of some basis functions, we term the  corresponding zero level set as parametric zero level set. We note that the surface properties would depend on the specific basis functions and will indeed decide the type of the kernel used in the algorithms in Section \ref{noisy}. We now provide some examples of parametric representations, depending on the choices of the basis functions. 


\subsection{Shift invariant surface representation}
A popular choice for the basis functions is the shift invariant representation, where compactly supported basis functions such as B-splines are used. Specifically, the basis functions are shifted copies of a template $\varphi$, denoted by:
\begin{equation}
\varphi_{\mathbf k}(\mathbf x) = \varphi\left(\frac{\mathbf x}{T}-\mathbf k\right).
\end{equation}
Here $T$ is the grid spacing, which controls the degrees of freedom of the representation. The number of B-splines in the above representation is $1/(T - 1)^n$. One may also choose a multi-resolution or sparse wavelet surface representation, when the basis functions are shifted and dilated copies of a template. This approach allows the surface to have different smoothness properties at different spatial regions.

\subsection{Polynomial surface representation}
	\label{poly}
The surface can also be represented as a linear combination of polynomials \cite{ongie2017algebraic}. The polynomial degree will control the degrees of freedom in this setting. This work is inspired by \cite{ongie2017algebraic}. However, we note that the recovery of the surface and functions from few points are not considered in the polynomial setting. In addition, the stability of polynomial representations is not fully clear, which may be needed to represent complex surfaces; we note that low degree polynomials were considered in the examples in \cite{ongie2017algebraic}.

\subsection{Band-limited surface representation}
We assume that the surface is within $[0,1)^n$. A well-studied representation for support limited functions is the Fourier exponential basis, which is widely used in digital image processing \cite{strohmer1997computationally, pan2013sampling, zou2019sampling}, biomedical image processing \cite{strohmer1996recover, potts2002fourier, ongie}, and geophysics \cite{rauth1998smooth}. The level set function can be assumed to be band-limited \cite{ongie}, when $\psi$ is expressed as a Fourier series:
\begin{equation}\label{tri}
\psi(\mathbf {x}) = \sum_{\mathbf {k}\in\Lambda}\mathbf {c_k}\exp(j2\pi \mathbf {k}^T\mathbf {x}),\quad \mathbf {x}\in[0,1)^n.
\end{equation}
In the above representation, the set $\Lambda$ denotes the bandwidth of the Fourier coefficients $\mathbf{c} = \{\mathbf{c_k}: \mathbf{k}\in\Lambda\}$; its cardinality $|\Lambda|$ is the number of free parameters in the surface representation. We refer to $\Lambda$ as the Fourier support of $\psi$ and we note that we always choose the support to be symmetric with respect to the origin. This choice is governed by the relation of this representation with polynomials, described in the next subsection. The extension of $\Lambda$ governs the degree of the polynomial.
%

In this work, we focus on the Fourier series representation due to its key benefits including well-developed theoretical tools, fast algorithms such as fast Fourier transform, orthogonality, and the property that $|\exp(j2\pi \mathbf {k}^T\mathbf {x})|=1$, which results in stable representations and also facilitate the theory. In this work, we mainly focus on this representation because it facilitates us to borrow the theoretical tools from our past work \cite{poddar2018recovery,poddaricassp,zou2019sampling}. We note that these results may be extended to other basis sets but is beyond the scope of this work. We will now review some of the properties of this representation, which we will use in the following sections. 
 
\subsubsection{Relation of bandlimited representation with polynomials}\label{pp1section}

 We also note that bandlimited representations \eqref{tri} have an intimate relation with polynomials \cite{ongie}. In particular, we note that one can transform the polynomial basis to an exponential one by the one-to-one mapping $\nu_i :  [0,1)\to \{z\in\mathbb{C} : |z|=1\}$:
\begin{equation}\label{map}
\nu_i(x_i)= \exp(j2\pi x_i)=:z_i.
\end{equation}
We will make use of this correspondence to study the properties of the zero sets of \eqref{tri}. With this transformation, the representation \eqref{tri} simplifies to the complex polynomial denoted as $\mathcal{P}[\psi]$, which is of the form
\begin{equation}\label{comppoly}
\mathcal{P}[\varphi](\mathbf z) = \sum_{\mathbf{k}\in\Lambda}c_{\mathbf{k}}\prod_{i=1}^n z_i^{k_i}.
\end{equation}
Since the mapping involves powers of $z_i$, where $z_i$ are specified by the trigonometric mapping \eqref{map}, we term the expansion in \eqref{tri} as a trigonometric polynomial.

We note that the mapping $\nu=(\nu_1,\cdots,\nu_n)$ defined by \eqref{map} is a bijection from $[0,1)^n$ onto the complex unit torus $\mathbb{T}^n = \{(z_1,\cdots,z_n) : |z_i|=1, i=1,\cdots,n\}$. Hence, 
\begin{equation}\label{key}
\psi(\mathbf{x})=0 ~~\Leftrightarrow~~ \mathcal{P}[\psi][\mathbf z] = 0 ~\text{on}~\mathbb{T}^n, ~~\mbox{where}~~z_i = \nu_i(x_i), ~~i=1,\cdots,n,
\end{equation}
which implies that there is a one-to-one correspondence between the zero sets of $\psi$ and the zeros of $\mathcal{P}[\psi]$ on the unit torus. Accordingly, we can study the algebraic properties of trigonometric polynomials and their zero sets by studying their corresponding complex polynomials under the mapping $\nu$.

\subsubsection{Non-uniqueness of level-set representation}\label{pp1section}

We first show that the level set representation of a surface in \eqref{tri} may not be unique, when the bandwidth of the representation is larger than the minimal one required to represent the surface. We first note that the function $\psi(\mathbf x)$ with bandwidth $\Lambda$ in \eqref{tri} can be expressed with a larger bandwidth $\Gamma \supset \Lambda$ by zero filling the additional Fourier coefficients:
	\begin{equation}\label{tri1}
	\psi(\mathbf {x})=\sum_{\mathbf {k}\in\Gamma}\tilde{\mathbf{c}}_{\mathbf{k}}\exp(j2\pi \mathbf {k}^T\mathbf {x}),\quad \mathbf {x}\in[0,1)^n,
	\end{equation}
	where the coefficients set $\tilde{\mathbf c}$ is the zero-filled version of the vector $\mathbf c$, denoted by $\tilde{\mathbf c} \in \mathbb C^{|\Gamma|}$:
	\begin{equation}\label{zerofilled}
	\tilde{\mathbf{c}}_{\mathbf{k}} =\begin{cases}
	\mathbf{c_k} & \text{if }  \mathbf{k} \in \Lambda \\
	0 & \text{else}
	\end{cases}.
	\end{equation}
We note that the representation of the surface by functions with the larger bandwidth $\Gamma$ is not unique. In particular, any uniform shift of the coefficients in the Fourier domain corresponds to a phase multiplication in the space domain:
	\begin{equation}\label{shift}
	\varphi' = \varphi\cdot\exp(j2\pi \mathbf {k}_0^T\mathbf {x});\quad\mathbf{k}_0\in \Gamma\ominus \Lambda.
	\end{equation}
	Since $|\exp(j2\pi \mathbf {k}_0^T\mathbf {x})| = 1, \forall \mathbf{x}$, we can see that the zero sets of $\varphi'$ are identical to that of $\varphi$.

	Because the exponentials $\exp(j2\pi \mathbf {k}_0^T\mathbf {x})$ are orthogonal to each other, the functions $\varphi'$ that has the same zero set as $\varphi$ lives in a subspace of dimension $\Gamma\ominus \Lambda$. Here, $\Gamma\ominus \Lambda$ denote the set of all valid uniform shifts $\mathbf k_0$ of $\Lambda$, denoted by $\Lambda + \mathbf k_{0}$, that are contained in $\Gamma$. We will introduce the set $\Gamma\ominus \Lambda$ with more details in \S\ref{nonminlifting}.

\subsubsection{Minimal bandwidth representation of a surface}\label{pp1section}
	
	We note from the previous section that the multiplication with the phase term in \eqref{shift} corresponds to multiplying the trigonometric polynomial in \eqref{comppoly} by $\mathbf z^{\mathbf k_0}$; the degree of the resulting trigonometric polynomial $\varphi'$ will be greater than that of $\varphi$. In this section, we show that out of all these polynomials, the one with smallest degree is unique. More importantly, the bandwidth of the above minimal polynomial can be used as a measure of the complexity of the surface. Specifically, a more complex surface would correspond to a polynomial with a larger bandwidth.
	
The following result shows that for any given surface $\mathcal S$, there exists a unique level set function $\psi$, whose coefficient set $\{\mathbf{c_k} : \mathbf{k}\in\Lambda\}$ has the smallest bandwidth. 

\begin{prop}\label{uniqueminimal}
For every (hyper-)surface $\mathcal{S}$ given by the zero level set of \eqref{tri1}, there is a unique (up to scaling) minimal trigonometric polynomial $\psi$, which satisfies $\psi(\mathbf x)=0; \forall \mathbf x \in \mathcal S$. Any other trigonometric polynomial $\psi_1$ that also satisfies $\psi_1(\mathbf x)=0; \forall \mathbf x \in \mathcal S$ will have $BW(\psi_1)\supseteq BW(\psi)$. Here, $BW(\psi)$ denotes the bandwidth of the function $\psi$.
\end{prop}

As seen from \eqref{shift}, the coefficients of $\psi_1$ can be the shifted version of the coefficients of $\psi$. Thus, the Fourier support of $\psi_1$ is larger than (contains) the Fourier support of $\psi$; the degree of the trigonometric polynomial $\psi_1$ is larger than the degree of the minimal polynomial $\psi$, which has the smallest degree or equivalently bandwidth. In this sense, the minimal polynomial $\psi$ is unique, up to scaling. The proof of this result  is given in Appendix \ref{minimalsection}. We refer to the $\psi$ of the form \eqref{tri} with the minimal bandwidth $\Lambda$ that satisfy
\begin{equation}
\psi(\mathbf x) = 0; ~~\forall \mathbf x \in \mathcal S
\end{equation}
 as the \emph{minimal trigonometric polynomial} of the surface $\mathcal S$. 
 
 In other words, when $\psi$ is the minimal trigonometric polynomial of a surface $\mathcal S$, it does not have a factor with no zeros (i.e., never vanishes or vanishes only at isolated points on $[0,1)^n$). In particular, if a polynomial has a factor with no zeros in $[0,1)^n$, one can remove this factor and obtain a polynomial with a smaller bandwidth and with the same support set. Note from \eqref{comppoly} that the minimal trigonometric polynomial will correspond to $\mathcal{P}[\psi]$ being a polynomial with the minimal degree.


As mentioned at the beginning of this section, the bandwidth $\Lambda$ of the minimal polynomial of the surface $\mathcal S$ grows with the complexity of $\mathcal S$; a more oscillatory surface with a lot of details corresponds to a high bandwidth minimal polynomial, while a simple and highly smooth surface corresponds to a low bandwidth minimal polynomial. We hence consider $|\Lambda|$ as a \emph{complexity measure} of the surface. Furthermore, we note that the surface model can approximate an arbitrary closed surface with any degree of accuracy, as long as the bandwidth is large enough \cite{ongie}. One can refer to Fig.2 in \cite{ongie} for illustration in 2D and see Fig. \ref{cxcurve} for illustration in 3D. Here we illustrate this idea in 2D/3D for simplicity, but the approach is general for any dimensions.

\begin{figure}[!h]
\centering
\subfloat[$17\times17\times17$ coefficients]{\includegraphics[angle=90, width=0.27\textwidth]{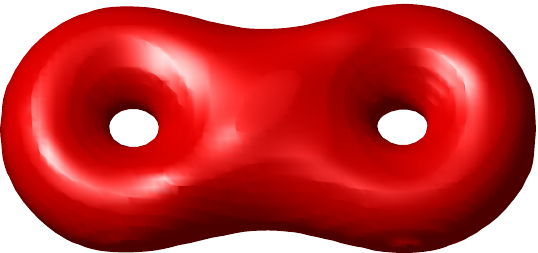}}\hspace{3em}
\subfloat[$25\times25\times25$ coefficients]{\includegraphics[angle=90, width=0.2\textwidth]{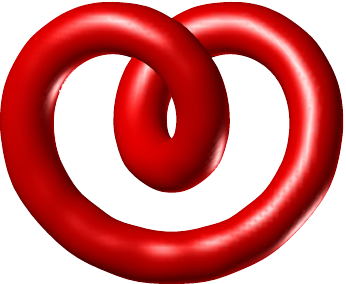}}\hspace{3em}
\subfloat[$33\times33\times33$ coefficients]{\includegraphics[width=0.25\textwidth]{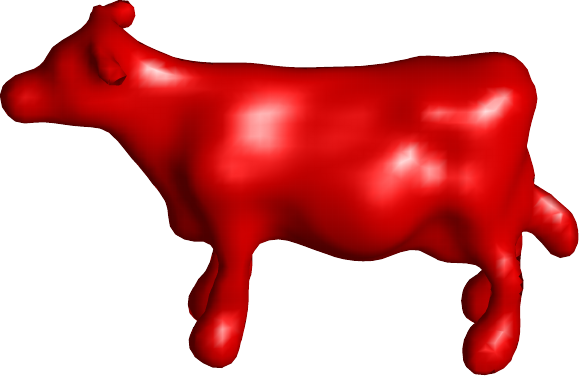}}
	\caption{Illustration the fertility of our level set representation model in 3D. The three examples show that our model is capable to capture the geometry of the shape even though the shape has complicated topologies, which demonstrated that the representation is not restrictive.}
	\label{cxcurve}
\end{figure}

\subsubsection{Irreducible bandlimited surfaces}

We now introduce the concept of irreducible polynomials, which is important for our results. We term a surface to be irreducible if its minimal trigonometric polynomial is irreducible. A polynomial is irreducible if it cannot be factorized into smaller factors, whose zero sets are within $[0,1)^n$. Most of the irreducible surfaces are simply connected (i.e., consist of a single connected component \footnote{One can come up with counter examples of irreducible polynomials with multiple components. In this work, one can ignore these pathological counter examples and assume that an irreducible bandlimited surface will consist of only one connected component.}). Intuitively, a general surface may be composed of several connected components, where each connected component is irreducible. In this case, we term the above surface as the union or irreducible surfaces. The minimal polynomial of the union of irreducible surfaces will be the product of the irreducible minimal polynomials of the individual connected components. The following definitions puts the above explanations into more concrete terms:
\begin{defn}
	A surface is termed as irreducible, if it is the zero set of an irreducible trigonometric polynomial. 
\end{defn}

\begin{defn}
A trigonometric polynomial $\psi(\mathbf {x})$ is said to be irreducible, if the corresponding polynomial $\mathcal{P}[\psi]$ is irreducible in $\mathbb{C}[z_1,\cdots,z_n]$. A polynomial $p$ is irreducible over a field of complex numbers, if it cannot be expressed as the product of two or more non-constant polynomials with complex coefficients.
\end{defn}

When $\psi$ can be written as the product of several irreducible components $\psi = \prod_{i=1}^{m}\psi_i$, then $\mathcal S[\psi]$ is essentially the union of irreducible surfaces:
\begin{equation}\label{uim}
\mathcal{S}[\psi]=\bigcup_{i=1}^{m}\mathcal{S}[\psi_i].
\end{equation}


\section{Lifting mapping and low-dimensional feature spaces}\label{liftingsection}
In this section, we show that there exists a non-linear transformation, which maps the points on an irreducible surface to a low-dimensional subspace. The transformation is intimately tied in with the specific choice of basis functions used to represent the surface. Our results show that the dimension of the subspace depends on the complexity of the surface, or equivalently the bandwidth of the minimal polynomial. We can use the rank of the feature matrix as a surrogate of the complexity of the surface to recover it, much like sparsity is used to recover signals in compressed sensing.

Consider the non-linear lifting mapping $\Phi_{\Gamma}: [0,1]^n \rightarrow \mathbb C^{|\Gamma|}$, obtained by evaluating the basis functions at $\mathbf x$:
\begin{equation}\label{glifting}
\Phi_{\Gamma}(\mathbf{x})=\begin{bmatrix}\varphi_{\mathbf k_1}(\mathbf x) \\ \vdots\\\quad \varphi_{\mathbf k_{|\Gamma|}} (\mathbf x))\end{bmatrix}.
\end{equation}
We can view $\Phi_{\Gamma}(\mathbf x)$ as  the feature vector of the point $\mathbf  x$, analogous to the ones used in kernel methods \cite{scholkopf}. Here, $|\Gamma|$ denotes the cardinality of the set $\Gamma$. We denote the set
\begin{equation}\label{fs}
\mathcal{V}_{\Gamma}(\mathcal{S}) = \{\Phi_{\Gamma}(\mathbf x)| \mathbf x \in \mathcal S\}
\end{equation}
as the feature space of the surface $\mathcal S$. Since any point on a surface $\mathcal{S}$ satisfies \eqref{zerolevel}, the feature vectors of points  from $\mathcal S$ satisfy
\begin{equation}\label{ann}
{\mathbf c}^T \Phi_{\Gamma}(\mathbf x) = 0, ~~\forall \mathbf x \in \mathcal S,
\end{equation}
where $\mathbf c$ is the coefficients vector in the representation of $\psi$ in \eqref{tri}. The above relation is illustrated in Fig. \ref{illusgraph}.
\begin{figure}[!h]
	\centering
	\includegraphics[width=0.3\textwidth]{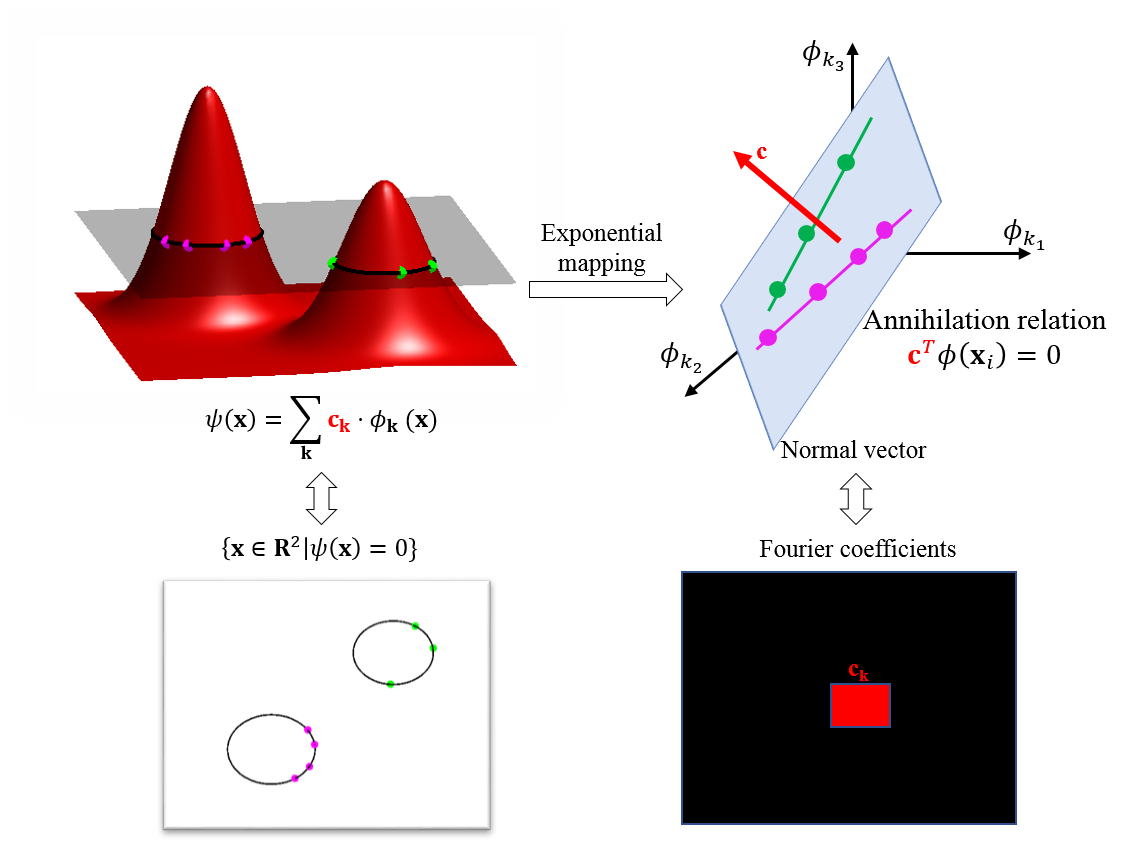}
	\caption{Illustration of the annihilation relations \eqref{ann} in 2D. Assume that the curve is the zero set of a band-limited function $\psi(\mathbf x)$, shown in the top left. The Fourier coefficients of $\psi$, denoted by $\mathbf c$, are bandlimited in $\Lambda$, denoted by the red square in the bottom right. Each point on the curve satisfies $\psi(\mathbf x_i)=0$. Using the representation \eqref{zerolevel}, we have $\mathbf c^T\Phi_{\Lambda}(\mathbf x_i)=0$. This means that the feature map will lift each point in the level set to a $|\Lambda|$ dimensional subspace whose normal vector is specified by $\mathbf c$, as illustrated by the plane and the red vector $\mathbf c$ in the top right. Note that if more than one closed curve are presented, each curve will be lifted to a lower dimensional subspace in the feature space, as shown by the two lines in the plane, and the lower dimensional spaces will span the $|\Lambda|$ dimensional subspace. (Figure courtesy of Q. Zou, reprint from \cite{zou2019sampling} with permission from IEEE).}
	\label{illusgraph}
\end{figure}

The relation \eqref{ann} also implies that $\mathbf c$ is orthogonal to all the feature vectors of points living on $\mathcal S$ and hence a feature matrix constructed from points on the surface is rank deficient by one; i.e., the dimension of the feature space is at most $|\Gamma|-1$. However, we now show that the feature matrix is often significantly low-rank depending on the geometry of the surface and the specific representations of the surface.

\subsection{Shift invariant representation}
We now show that if the level set function is represented by a shift invariant representation (e.g. B-splines), the dimension of the lifted feature points are dependent on the area of the surface. We consider $\varphi_{\mathbf k} = \beta^p\left(\frac{\mathbf x}{T}-\mathbf k\right)$ to be the $p^{\rm th}$-degree tensor-product B-spline function. Note that $\beta^p(\mathbf x)$ is support limited in $[-(p+1)/2,(p+1)/2]^n$. If the support  of $\varphi_k(\mathbf x)$ does not overlap with $\mathcal S$, we have $\varphi_k(\mathbf x)=0; \forall \mathbf x \in \mathcal S$. Hence, we have 
\begin{equation}\label{nullspace1}
	\mathbf i_{\mathbf k}^T~\Phi_{\Gamma}(\mathbf x) = 0,~~\forall \mathbf x \in \mathcal S
\end{equation}
where $\mathbf i_{\mathbf k}$ is the indicator vector whose $k^{\rm th}$ entry is one and the rest of the entries are zeros.  Note that all of these indicator vectors are linearly independent. The number of basis vectors whose bandwidth does not overlap with $\mathcal S$ is dependent on the area of $\mathcal S$ as well as the support of $\psi$. Thus, the dimension of $\mathcal V_{\Gamma}(\mathcal S)$ is a measure of the area of the surface $\mathcal S$, and satisfies 
\begin{equation}\label{rank1}
{\rm{dim}} \left(\mathcal V_{\Gamma}(\mathcal S)\right) \leq |\Gamma| - (P+1) = A, 
\end{equation}
where $P$ is the number of basis functions whose support does not overlap with $\mathcal S$.

\subsection{Band-limited surface representation}
\label{liftsection}

We now consider the case of  an arbitrary point $\mathbf x$ on the zero level set of $\psi(\mathbf x)$ with bandwidth $\Lambda$. Using \eqref{tri1}, the lifting is specified by:
\begin{equation}\label{lifting}
\Phi_{\Lambda}(\mathbf{x})=\begin{bmatrix}\exp(j2\pi \mathbf {k}_1^T\mathbf {x}) \\ \exp(j2\pi \mathbf {k}_2^T\mathbf {x})\\ \vdots\\\quad \exp(j2\pi \mathbf {k}_{|\Lambda|}^T\mathbf {x})\end{bmatrix}.
\end{equation}
We note from \eqref{lifting} that the lifting $\Phi$ can be evaluated with a larger bandwidth $\Gamma \supset \Lambda$. When the lifting is performed with the minimal bandwidth (i.e., $\Gamma=\Lambda$), we term the corresponding lifting as the \emph{minimal lifting}. 

We now analyze the dimension of the feature space $\mathcal V_{\Lambda}(\mathcal S)$ for the minimal ($\Gamma=\Lambda$) and non-minimal lifting ( $\Lambda\subset \Gamma$) cases. In both cases, we will show that the feature space is low-dimensional and is a subspace of $\mathbb C^{|\Lambda|}$.

\subsubsection{Irreducible surface with minimal lifting ($\Gamma=\Lambda$)}
We first focus on the case where $\psi$ is an irreducible trigonometric polynomial and the bandwidth of the lifting is specified by $\Lambda$, which is the bandwidth of the minimal polynomial. The annihilation relation \eqref{ann} implies that $\mathbf c$ is orthogonal to the feature vectors $\Phi_{\Lambda}(\mathbf {x})$. This implies that 
\begin{equation}\label{dimension}
{\dim}(\mathcal V_{\Lambda}) \leq |\Lambda|-1.
\end{equation}

\subsubsection{Irreducible surface with non-minimal lifting ($\Gamma \supset \Lambda$)}
\label{nonminlifting}
We now consider the setting where the non-linear lifting is specified by $\Phi_{\Gamma}(\mathbf {x})$, where $\Lambda \subset \Gamma$. 
Because of the annihilation relation, we have 
\[\tilde{\mathbf{c}}^T~\Phi_\Gamma(\mathbf{x})=0,\]
where $\tilde{\mathbf{c}}$ is the zero filled coefficients in \eqref{tri1}. Since the zero set of the function $\psi_{\mathbf k_0}(\mathbf x) = \psi(\mathbf x)\cdot\exp(j2\pi \mathbf {k}_0^T\mathbf {x})$
is exactly the same as that of $\psi$, we have 
\begin{equation}
\sum_{\mathbf {k}}\mathbf {c_{\mathbf k-\mathbf k_0}}\exp(j2\pi \mathbf {k}^T\mathbf {x}) = 0; ~~\forall \mathbf x \in \mathcal S[\psi].
\end{equation}
This implies that any shift of $\tilde{\mathbf c}$ within $\Gamma \ominus \Lambda$, denoted by  $\tilde{\mathbf d}_{\mathbf k} = \mathbf c_{\mathbf k-\mathbf k_0}$ will satisfy $\tilde{\mathbf{d}}^T~\Phi_\Gamma(\mathbf{x})=0.$ It is straightforward to see that $\tilde{\mathbf d}$ and $\tilde{\mathbf c}$ are linearly independent for all values of $\mathbf k_0$. We denote the number of possible shifts such that the shifted set $\Lambda + \mathbf k_0$ is still within $\Gamma$ (i.e., $\Lambda + \mathbf k_0 \subseteq \Gamma$ ) by $|\Gamma\ominus\Lambda|$:
\begin{equation}
\Gamma\ominus\Lambda=\{\mathbf{l}\in\Gamma \mid \mathbf{l}-\mathbf{k}\in\Gamma, \forall \mathbf{k}\in\Lambda\}.
\end{equation}
This set is illustrated in Fig. \ref{ominusfig} along with $\Gamma$ and $\Lambda$.
\begin{figure}[!h]
	\centering
	\includegraphics[width=0.3\textwidth]{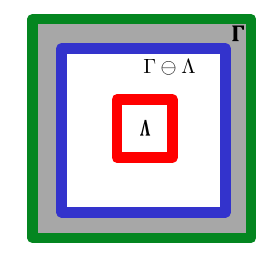}	
	\caption{The non-minimal filter bandwidth $\Gamma$ (green) is illustrated along with the minimal filter bandwidth $\Lambda$ (red). The set $\Gamma\ominus\Lambda$ (blue) contains all indices at which $\Lambda$ can be centered, while remaining inside $\Gamma$.}
	\label{ominusfig}
\end{figure}
Since the vectors $\mathbf c_{\mathbf k-\mathbf k_0}$ are linearly independent and are orthogonal to any feature vector $\Phi_\Gamma(\mathbf{x})$ on $\mathcal S[\psi]$, the dimension of the subspace is bounded by 
\begin{equation}\label{dimnonmin}
{\dim}(\mathcal V_{\Gamma}) \leq |\Gamma|-|\Gamma\ominus\Lambda|.
\end{equation}

\subsubsection{Union of irreducible surfaces with $\Gamma \supset \Lambda_i$}
When $\psi = \prod_{i=1}^{m}\psi_i$, each irreducible surface $\mathcal{S}[\psi_i]$ will be mapped to a subspace of dimension $|\Gamma|-|\Gamma\ominus\Lambda_i|$. This implies that the non-linear lifting transforms the union of irreducible surfaces to the well-studied union of subspace model \cite{lu2008theory, gedalyahu2010time, mishali2011xampling}.

\section{Surface recovery from samples}
\label{sampling}
In this section, we will use the low-rank structure of the feature maps of the points to recover the surface. As discussed in the introduction, the recovery of a surface/manifold from point clouds is an important problem in denoising, machine learning, shape recovery from point clouds, and image segmentation. For presentation purposes, we consider different cases in the increasing order of complexity. In particular, we consider irreducible (single connected component) surfaces with minimal lifting, union of irreducible components with minimal lifting, and finally the case with non-minimal lifting. Note that in practice, the bandwidth of the surface is not known apriori, and hence one has to over-estimate the bandwidth; this translates to the non-minimal lifting setting. Our results in this section show that irreducible surfaces can be recovered from very few samples, as long as the number of samples exceed a number proportional to the bandwidth. Union of irreducible surfaces can also be recovered from few samples, but each of the irreducible components need to be sampled adequately to guarantee perfect recovery.

\subsection{Sampling theorems}
We consider the recovery of the surface $\mathcal{S}$ from its samples $\mathbf{x}_i; i=1, \cdots ,N$. According to the analysis in the previous section, if the sampling point $\mathbf{x}_i$ is located on the zero level set of $\psi(\mathbf{x})$, we will then have the annihilation relation specified by \eqref{ann}. Notice that equation \eqref{ann} is a linear equation with $\mathbf{c}$ as its unknowns. Since all the samples $\mathbf{x}_i; i=1,..,N$ satisfy the annihilation relation \eqref{ann}, we have
\begin{equation}\label{nullspacecond}
\mathbf{c}^T \underbrace{\left[\Phi_{\Gamma}(\mathbf{x}_1) \quad\cdots\quad \Phi_{\Gamma}(\mathbf{x}_N)\right]}_{\Phi_{\Gamma}(\mathbf{X})}=0.
\end{equation}
We call $\Phi_{\Gamma}(\mathbf{X})$ the feature matrix of the sampling set $\mathbf{X}=\{\mathbf{x}_1, \cdots, \mathbf{x}_N\}$. We propose to estimate the coefficients $\mathbf c$, and hence the surface $\mathcal S[\psi]$ using the above linear relation \eqref{nullspacecond}. Note that $\mathcal S[\psi]$ is invariant to the scale of $\mathbf c$; without loss of generality, we reformulate the estimation of the surface as the solution to the system of equations
\begin{equation}
\label{optcoeffs}
\mathbf c^T~ \boldsymbol\Phi_{\Gamma}(\mathbf X)=0; ~~\|\mathbf c\|_F =1.
\end{equation}
We note that without the constraint $\|\mathbf c\|_F =1$, $\mathbf c^T~ \boldsymbol\Phi_{\Gamma}(\mathbf X)=0$ will have a trivial solution with $\mathbf c=0$. The use of the Frobenius norm constraint enables us to solve the problem using eigen decomposition.
The above estimation scheme yields a unique solution, if the matrix $\Phi_{\Lambda}(\mathbf{X})$ has a unique null-space basis vector. We will now focus on the number of samples $N$ and its distribution on $\mathcal S[\psi]$, which will guarantee the unique recovery of $\mathcal S[\psi]$. We will consider different lifting scenarios introduced in Section \ref{liftingsection} separately. As we will see, in some cases considered below, the null-space has a large dimension. However, the minimal null-space vector (coefficients with the minimal bandwidth) will still uniquely identify the surface, provided the sampling conditions are satisfied. 

\subsubsection{Case 1: Irreducible surfaces with minimal lifting}
\label{recoverysection}

Suppose $\psi(\mathbf{x})$ is an irreducible trigonometric polynomial with bandwidth $\Lambda$. Consider the lifting which is specified by the minimal bandwidth $\Lambda$. We see from \eqref{dimension} that $\rm{rank}\left(\Phi_{\Lambda}(\mathbf{X})\right) \leq |\Lambda|-1$. The following result shows when the inequality is replaced by an equality. 

\begin{prop}\label{rankprop}
Let $\{\mathbf{x}_1,\cdots, \mathbf{x}_N\}$ be $N$ independent and uniformly distributed random samples on the surface $\mathcal{S}[\psi]$, where $\psi(\mathbf{x})$ is an irreducible (minimal) trigonometric polynomial with bandwidth $\Lambda$. The feature matrix $\Phi_{\Lambda}(\mathbf X)$ will have rank $|\Lambda|-1$, if $$N\geq  |\Lambda|-1$$
for almost all surfaces $S[\psi]$.
\end{prop}

We note that the above results are true for almost all surfaces. This implies that the surfaces for which the above results do not hold correspond to a set of measure zero \cite{federer2014geometric}. 
The above proposition guarantees that the solution to the system of equations specified by \eqref{optcoeffs} is unique (up to scaling) when the number of samples exceeds $N = |\Lambda|-1$ with unit probability. The proof of this proposition can be found in Appendix \ref{proofrankprop}. With Proposition \ref{rankprop}, we obtain the following sampling theorem.


\begin{thm} [Irreducible surfaces of any dimension]\label{irred}
Let $\psi(\mathbf{x}), \mathbf{x}\in [0,1]^n, n\ge 2$ be an irreducible trigonometric polynomial whose bandwidth is given by $\Lambda$. The zero level set of $\psi(\mathbf{x})$ is denoted as $\mathcal{S}[\psi]$. If we are randomly given $N\ge |\Lambda|-1$ samples on $\mathcal{S}[\psi]$, then almost all surfaces $\mathcal{S}[\psi]$ can be recovered.
\end{thm}

This theorem generalizes the results in \cite{zou2019sampling} to any dimension $n\ge2$ and is illustrated in Fig. \ref{illu2d} and Fig. \ref{illu3d}. 
\begin{figure}[!h]
\centering
\subfloat[Curve]{\includegraphics[width=0.2\textwidth]{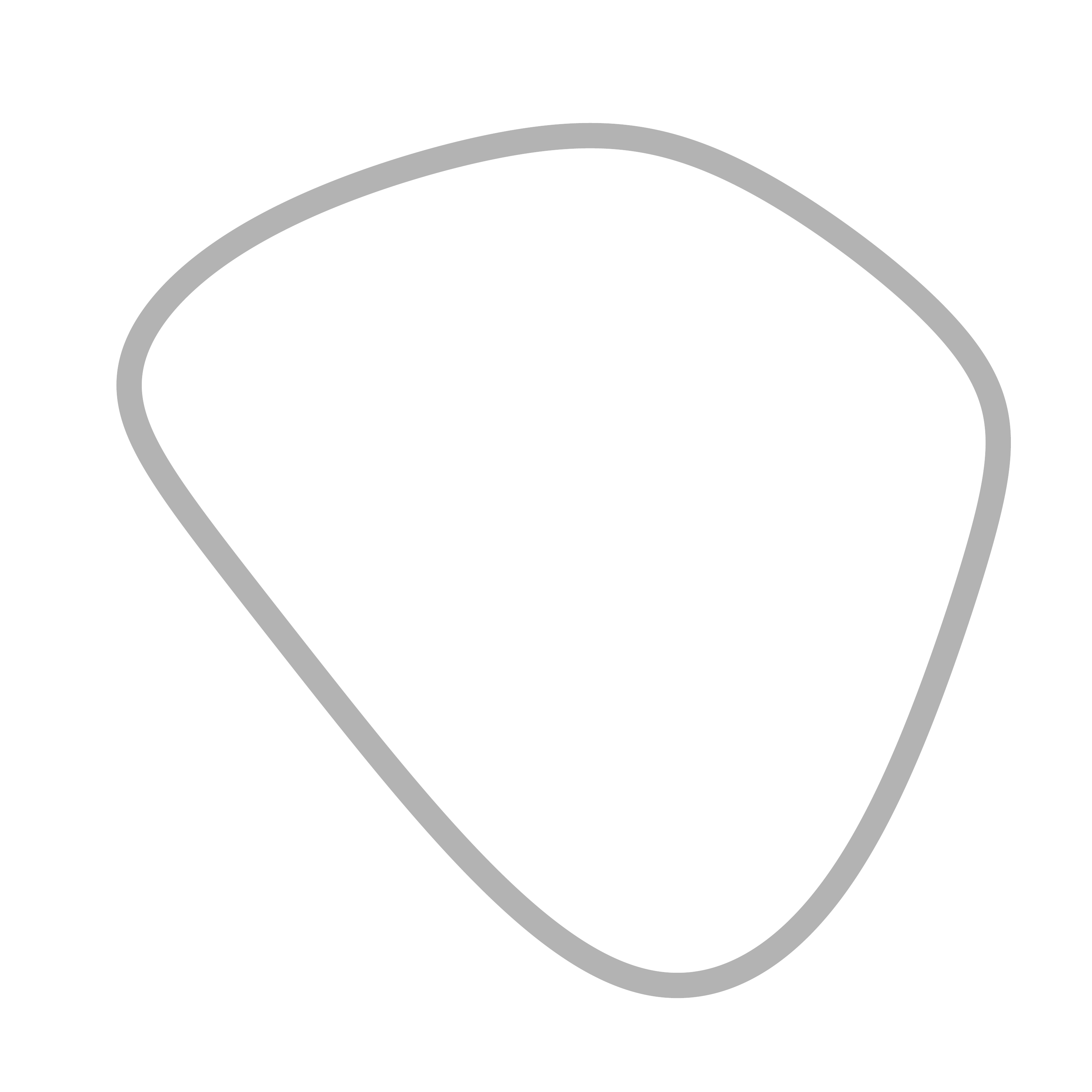}}\hspace{3ex}
\subfloat[Recovery \#1]{\includegraphics[width=0.2\textwidth]{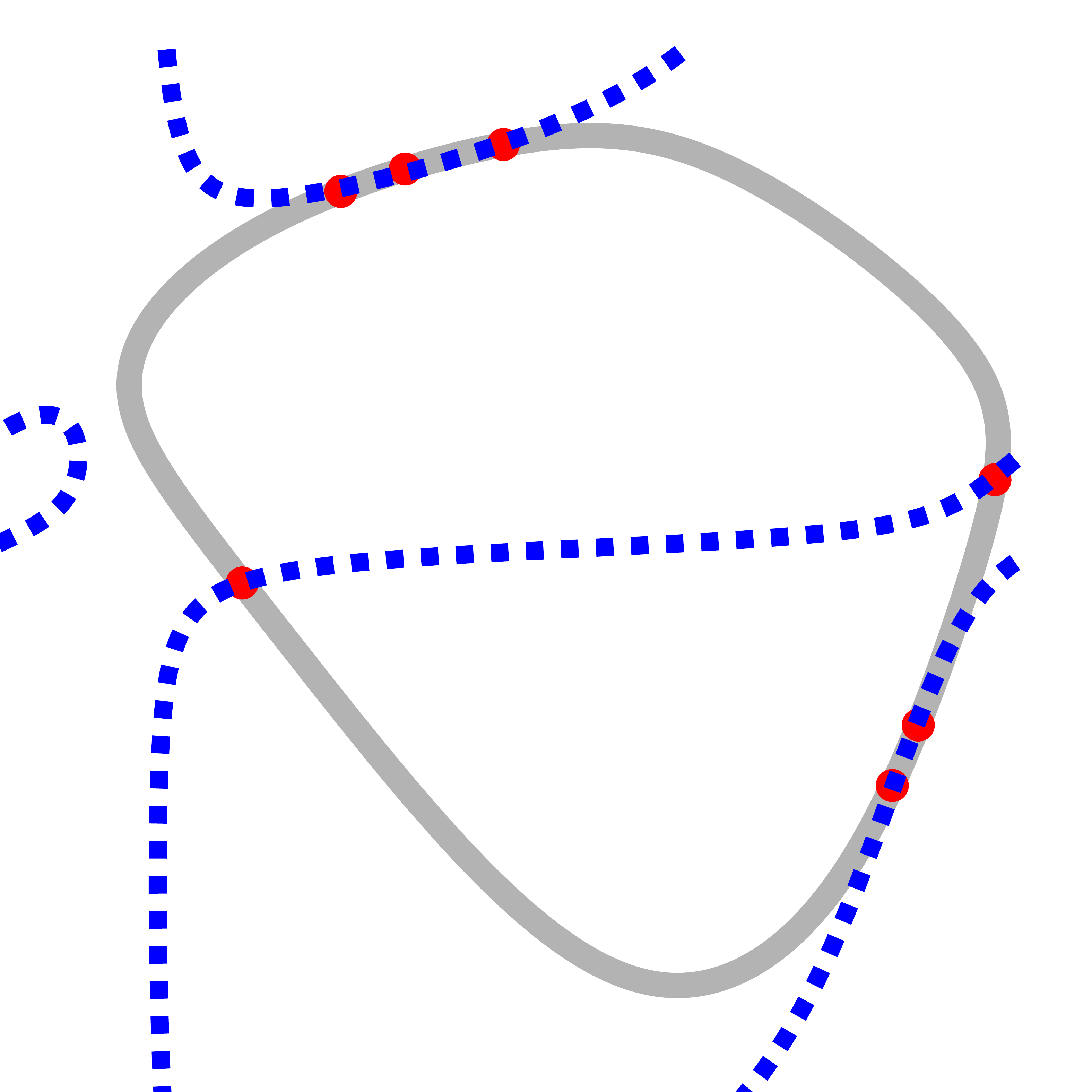}}\hspace{3ex}
\subfloat[Recovery \#2]{\includegraphics[width=0.2\textwidth]{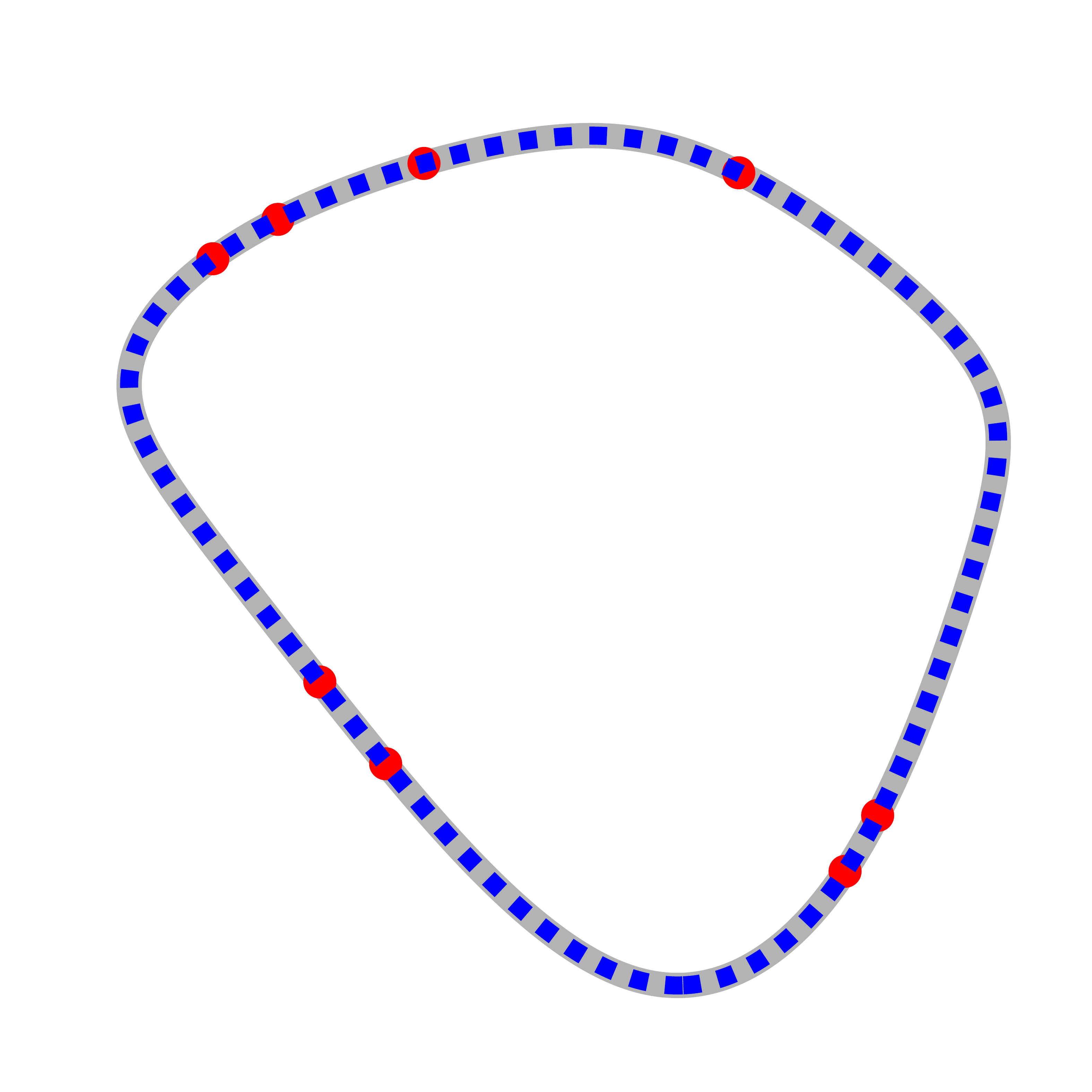}}\\
\subfloat[original $\psi(x,y)$]{\includegraphics[width=0.2\textwidth]{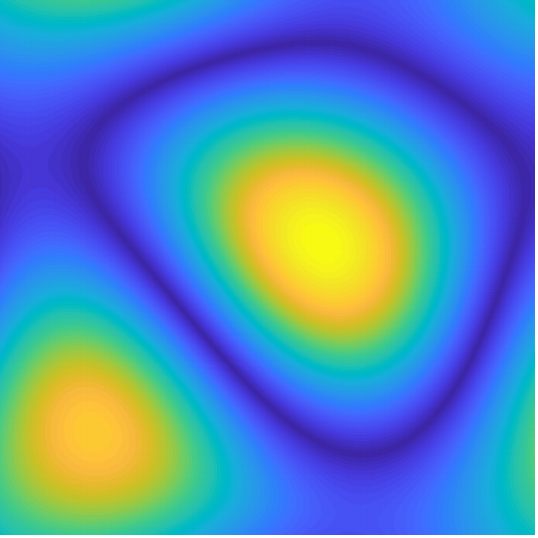}}\hspace{3ex}
\subfloat[$\psi(x,y)$ with 7 samples]{\includegraphics[width=0.2\textwidth]{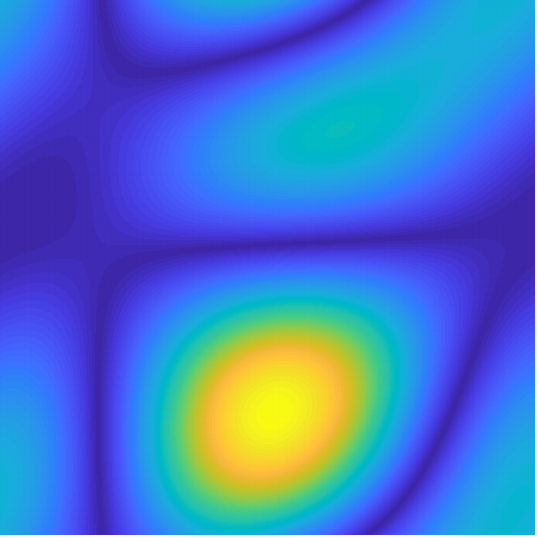}}\hspace{3ex}
\subfloat[$\psi(x,y)$ with 8 samples]{\includegraphics[width=0.2\textwidth]{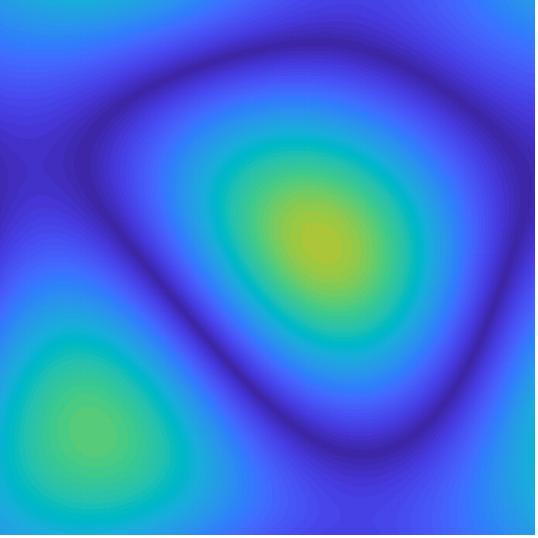}}\\
\caption{Illustration of Theorem \ref{irred} in 2D. The irreducible curve given by (a) is the original curve, which is obtained by the zero level set of a trigonometric polynomial whose bandwidth is $3\times 3$. According to Theorem \ref{irred}, we will need at least 8 samples to recover the curve. In (b), we randomly choose 7 samples (the red dots) on the original curve (the gray curve). The blue dashed curve shows the recovered curve from this 7 samples. Since the sampling condition is not satisfied, the recovery failed. In (c), we randomly choose 8 points (the red dots). From (c), we see that the blue dashed curve (recovered curve) overlaps the gray curve (the original curve), meaning that we recover the curve perfectly. In (d) - (f), we showed the original trigonometric polynomial, the polynomial obtained from 7 samples and the polynomial obtained from 8 samples.}
\label{illu2d}
\end{figure}

\begin{figure}[!h]
\centering
\subfloat[Surface]{\includegraphics[width=0.2\textwidth]{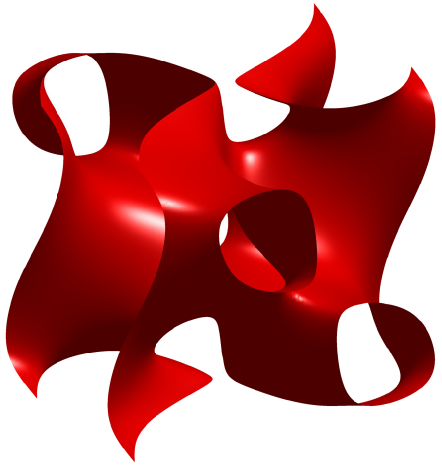}}\hspace{3em}
\subfloat[Recovery \#1]{\includegraphics[width=0.2\textwidth]{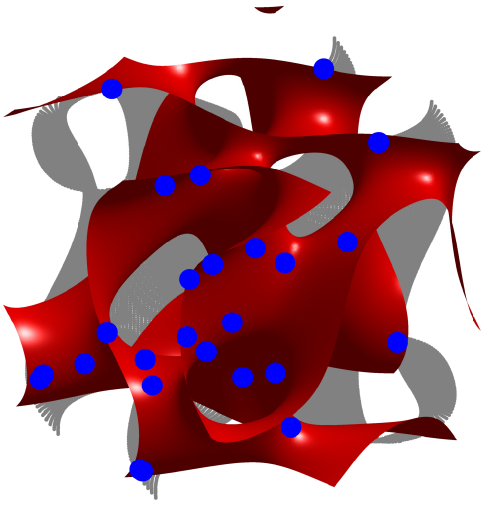}}\hspace{3em}
\subfloat[Recovery \#2]{\includegraphics[width=0.2\textwidth]{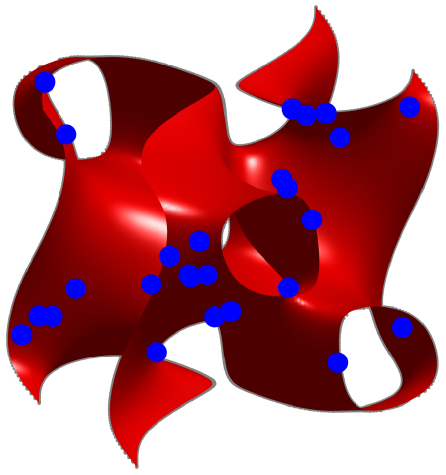}}
\caption{Illustration of Theorem \ref{irred} in 3D. The irreducible surface given by (a) is the original surface, which is given by the zero level set of a trigonometric polynomial whose bandwidth is $3\times 3\times 3$. According to Theorem \ref{irred}, we will need at least 26 samples to recover the surface. In (b), we randomly choose 25 samples (the blue dots) on the original surface (the gray part). The red surface is what we recovered from the 25 samples. Since the sampling condition is not satisfied, the recovery failed. In (c), we randomly choose 26 points (the blue dots). From (c), we see that the red surface (recovered surface) overlaps the gray surface (the original surface), meaning that we recover the surface perfectly.}
\label{illu3d}
\end{figure}

In the theorem, when $n=2$, then $\mathcal{S}$ is a planar curve. In this setting, if the bandwidth of $\psi$ $\Lambda$ is a rectangular region with dimension $k_1\times k_2$. Then by this sampling theorem, we get perfect recovery with probability one, when the number of random samples on the curve exceeds $k_1\cdot k_2-1$. Note that the degrees of freedom in the representation \eqref{tri} is $k_1\cdot k_2-1$, when we constrain $\|\mathbf c\|_F=1$. This implies that if the number of samples exceed the degrees of freedom, we get perfect recovery.  Note that these results are significantly less conservative than the ones in \cite{zou2019sampling}, which required a minimum of $(k_1+k_2)^2$ samples. We note that the results in \cite{zou2019sampling} were the worst case guarantees, and will guarantee the recovery of the curve from any  $(k_1+k_2)^2$ samples. By contrasts, our current results are high probability results; there may exist a set of $N \ge k_1\cdot k_2-1$ samples from which we cannot get unique recovery.

We note that the current work is motivated by the phase transition experiments (Fig. 5) in \cite{zou2019sampling}, which shows that one can recover the curve in most cases when the number of samples exceeds $k_1\cdot k_2-1$ rather than the conservative bound of $(k_1+k_2)^2$. We also note that it is not straightforward to extend the proof in \cite{zou2019sampling} to the cases beyond $n=2$. Specifically, we relied on Bezout's inequality in \cite{zou2019sampling}, which does not generalize easily to high dimensional cases. 


\subsubsection{Case 2: Union of irreducible surfaces with minimal lifting}\label{recoverysection2}
We now consider the union of irreducible surfaces $\mathcal S[\psi]$, where $\psi$ has several irreducible factors $\psi(\mathbf{x}) = \psi_1(\mathbf{x})\cdots\psi_M(\mathbf{x})$. Then we have $\mathcal{S}[\psi]=\bigcup_{i=1}^{M}\mathcal{S}[\psi_i]$. Suppose the bandwidth of $\psi(\mathbf{x})$ is given by $\Lambda$ and the bandwidth of each factor $\psi_i(\mathbf{x})$ is given by $\Lambda_i$. We have the  following result for this setting.

\begin{prop}\label{rankprop2}
Let $\psi(\mathbf{x})$ be a trigonometric polynomial with $M$ irreducible factors, i.e., 
\begin{equation}\label{reduciblepoly}
\psi(\mathbf{x}) = \psi_1(\mathbf{x})\cdots\psi_M(\mathbf{x}).
\end{equation}
Suppose the bandwidth of each factor $\psi_i(\mathbf{x})$ is given by $\Lambda_i$ and the bandwidth of $\psi$ is $\Lambda$. Assume that $\{\mathbf{x}_1,\cdots, \mathbf{x}_N\}$ are $N$ uniformly distributed random samples on $\mathcal S[\psi]$, which are chosen independently. Then with probability 1 that the feature matrix $\Phi_{\Lambda}(\mathbf{X})$ will be of rank $|\Lambda|-1$ for almost all $\psi$ if 
\begin{enumerate}
	\item each irreducible factor is randomly sampled with $N_i\ge|\Lambda_i|-1$ points, and
	\item the total number of samples satisfy $ N \ge |\Lambda|-1$.
\end{enumerate}
\end{prop}

Similar to previous propositions, the above results are valid for almost all $\psi$, which implies that the set of $\psi$ for which the above results do not hold is a set of measure zero \cite{federer2014geometric}.
The proof of this result can be seen in Appendix \ref{proofrankprop2}. Based on this proposition, we have the following sampling conditions.

\begin{thm} [Union of irreducible surfaces of any dimension]\label{reducible}
Let $\psi(\mathbf{x})$ be a trigonometric polynomial with $M$ irreducible factors as in \eqref{reduciblepoly}. If the samples $\mathbf x_1,.,\mathbf x_N$ satisfy the conditions in Proposition \ref{rankprop2}, then the surface can be uniquely recovered by the solution of \eqref{optcoeffs} for almost all $\psi$. 
\end{thm}

Unlike the sampling conditions in Theorem \ref{irred} that does not impose any constraints on the sampling, the above result requires each component to be sampled with a minimum rate specified by the degrees of freedom of that component. We illustrate the above result in Fig. \ref{illuprop2} in 2D ($n=2$), where $\mathcal S$ is the union of two irreducible curves with bandwidth of $3\times 3$, respectively. The above results show that if each of these simply connected curves are sampled with at least eight points and if the total number of samples is no less than 24, we can uniquely identify the union of curves. The results show that if any of the above conditions are violated, the recovery fails; by contrast, when the number of randomly chosen points satisfy the conditions, we obtain perfect recovery. 

\begin{figure}[!h]
\centering
\subfloat[Curve]{\includegraphics[width=0.18\textwidth]{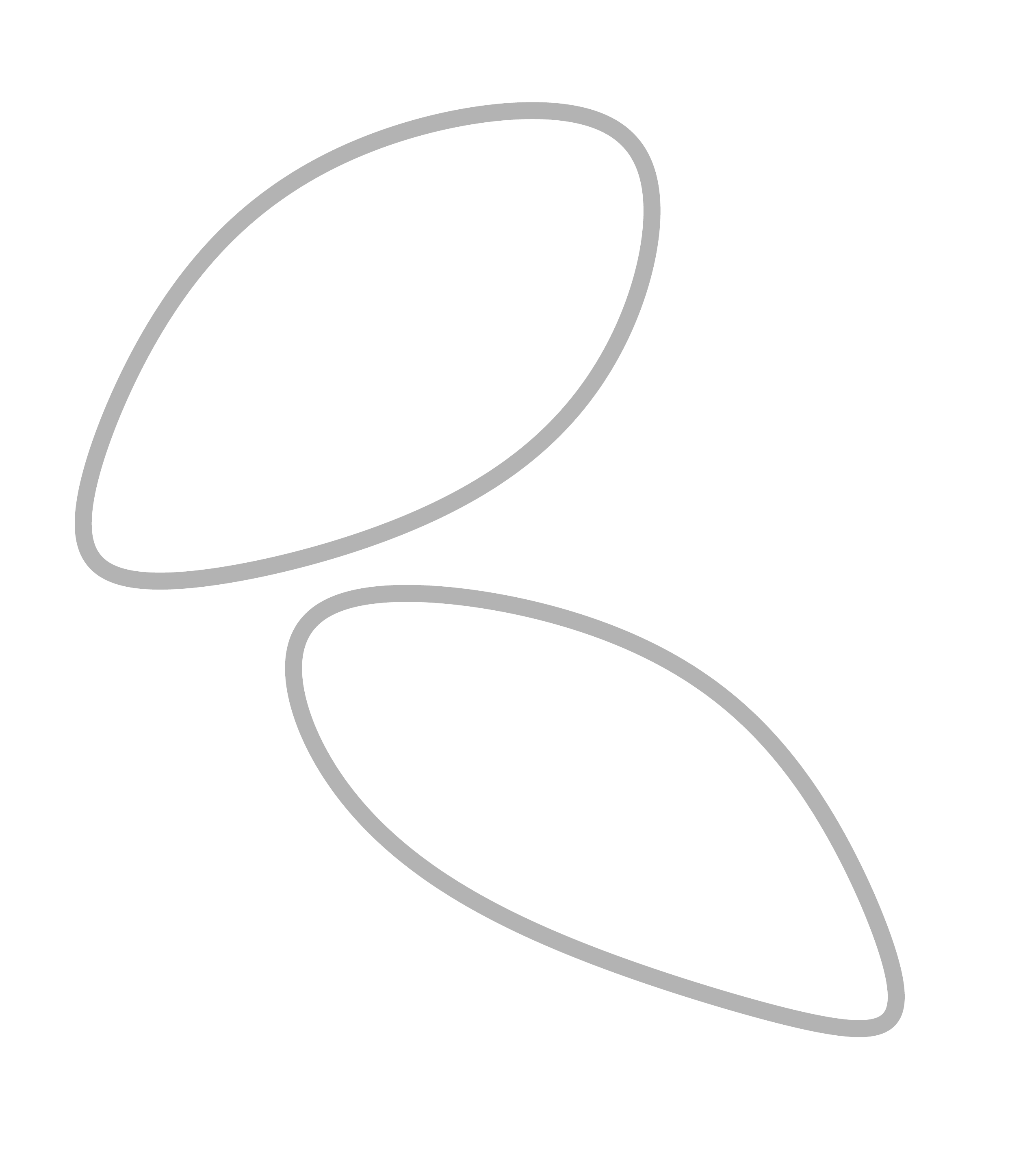}}\hspace{3em}
\subfloat[$7+17$]{\includegraphics[width=0.18\textwidth]{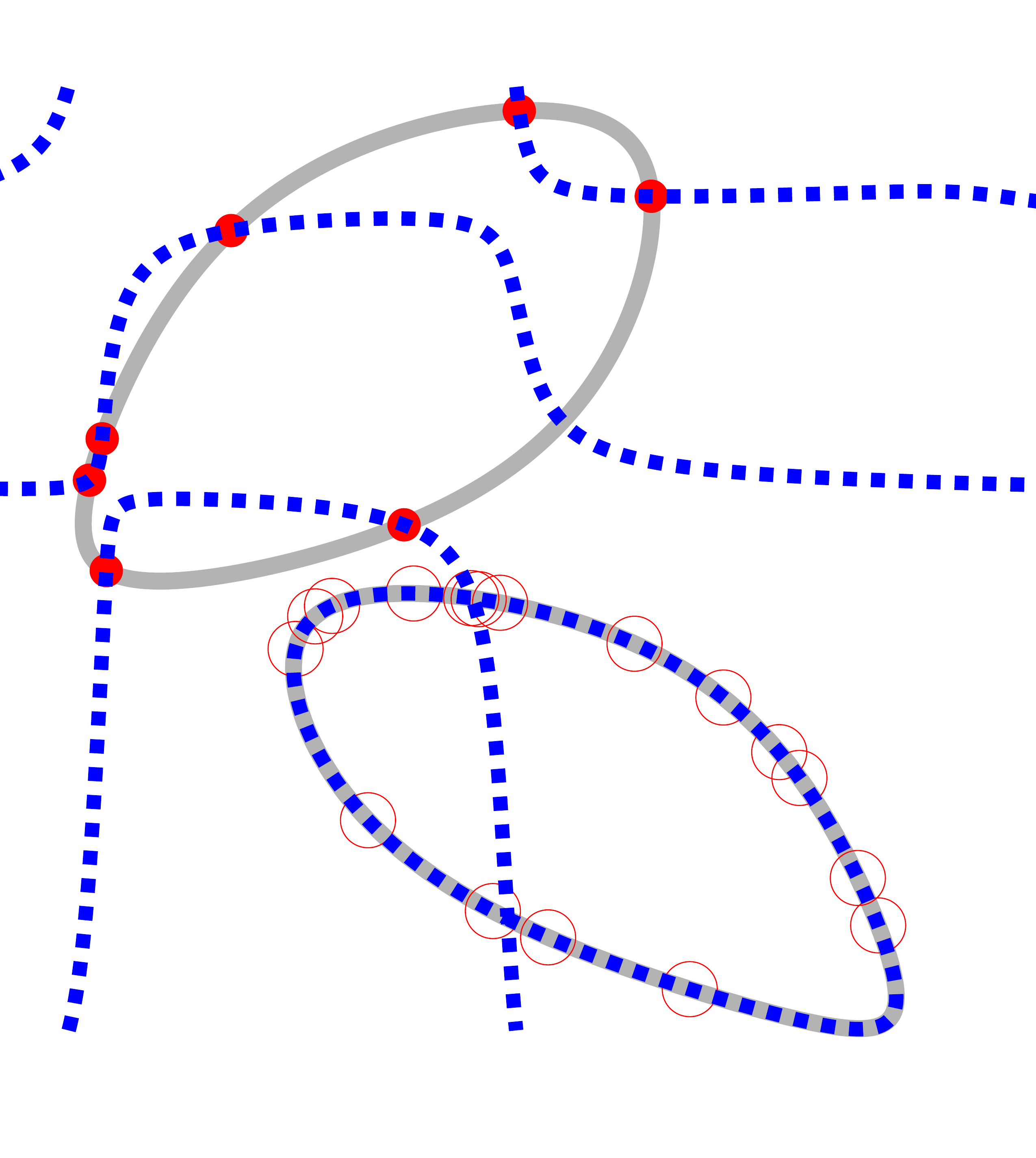}}\hspace{3em}
\subfloat[$8+16$]{\includegraphics[width=0.18\textwidth]{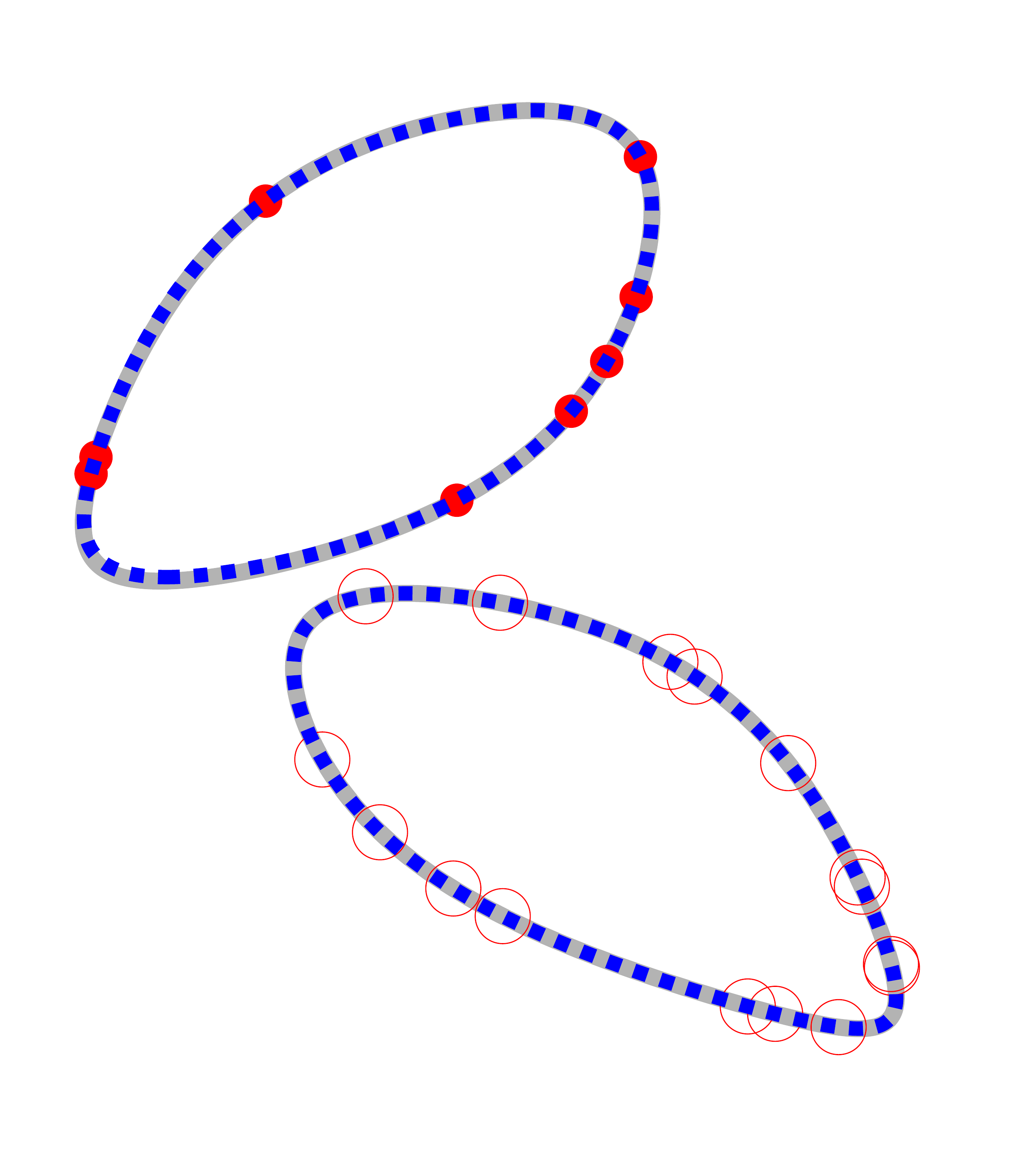}}\\
\subfloat[$17+7$]{\includegraphics[width=0.18\textwidth]{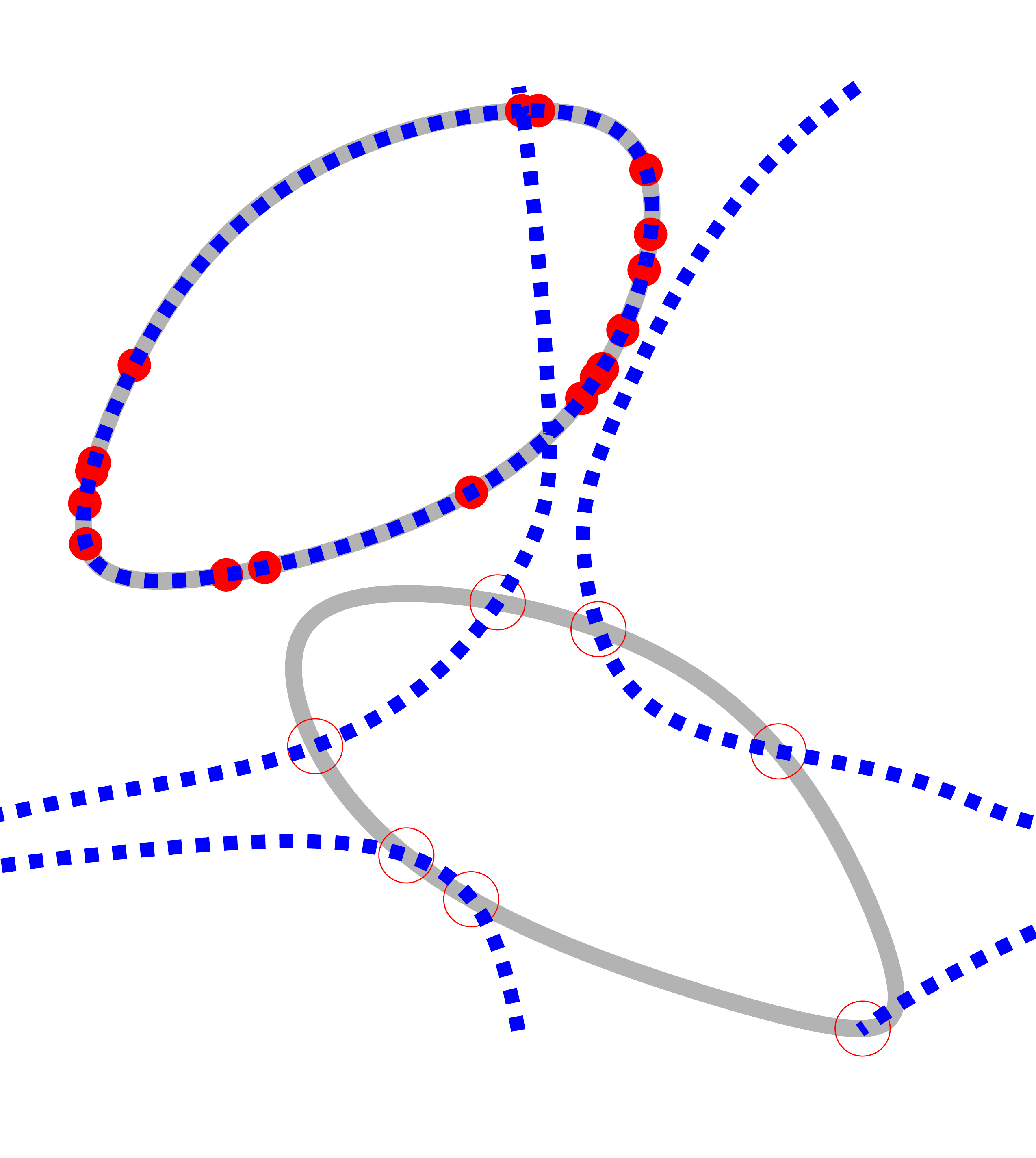}}\hspace{3em}
\subfloat[$16+8$]{\includegraphics[width=0.18\textwidth]{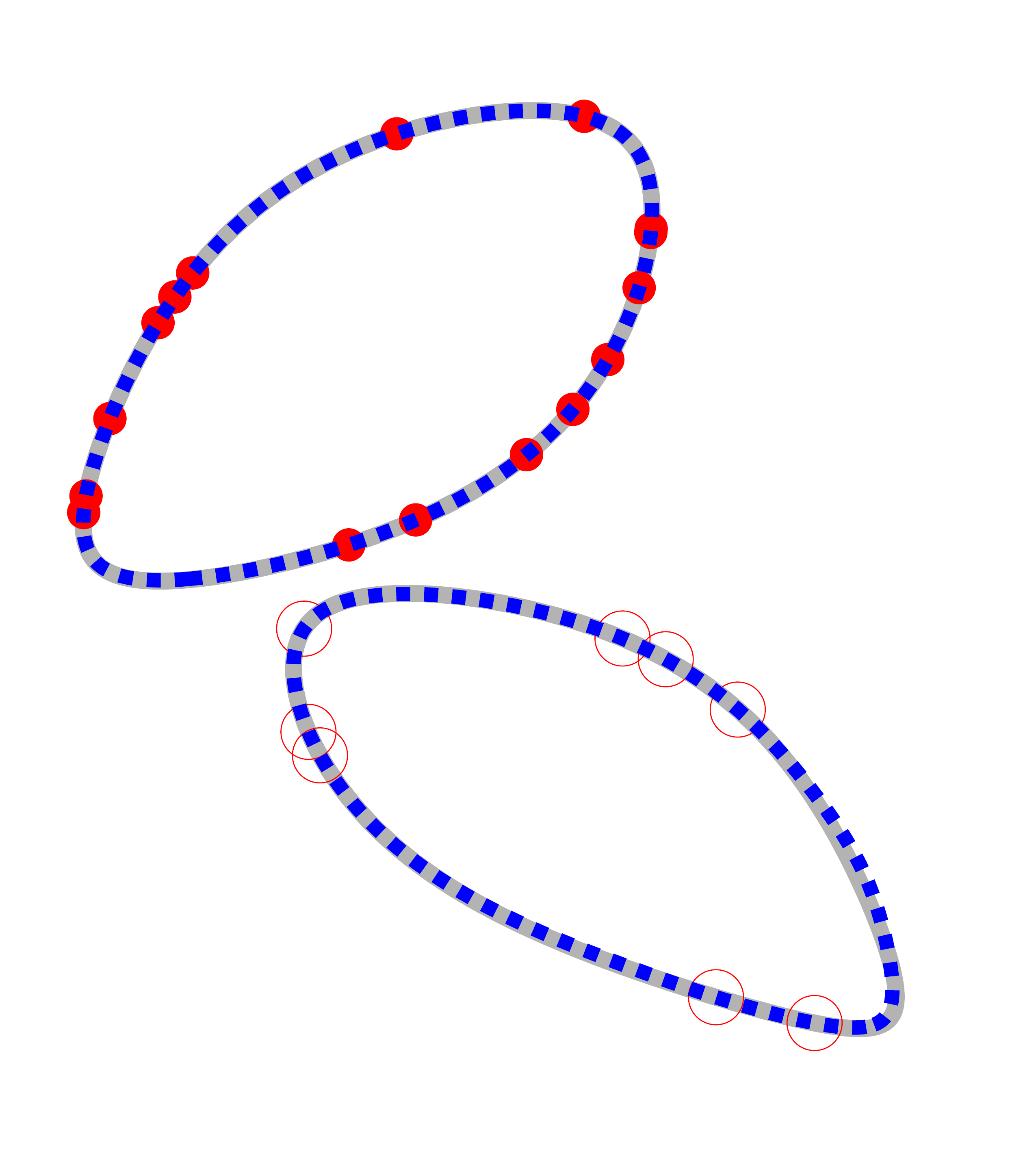}}\hspace{3em}
\subfloat[$8+8$]{\includegraphics[width=0.18\textwidth]{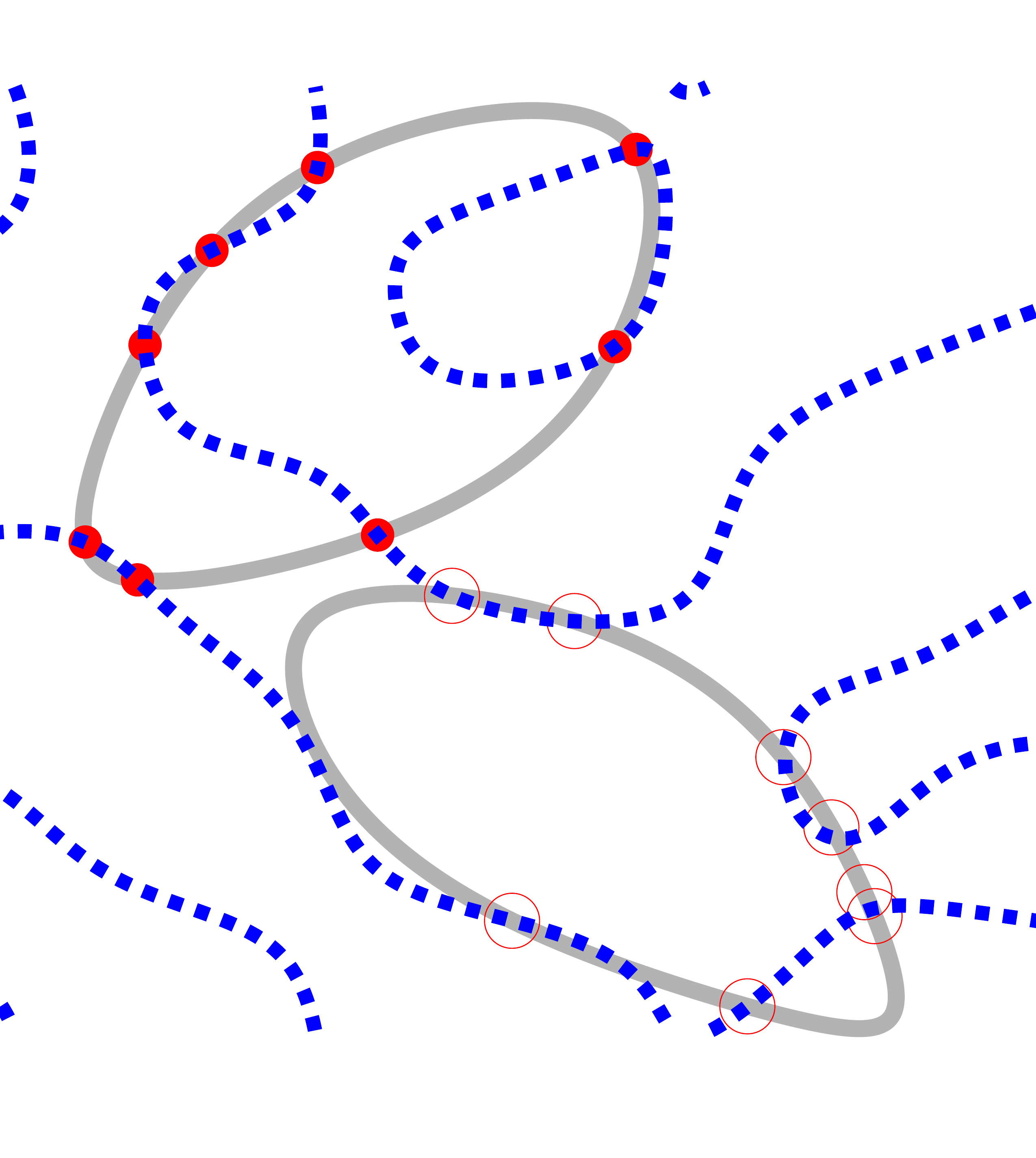}}
\caption{Illustration of Theorem \ref{reducible}. The original curve (a) is given by the zero set of a reducible trigonometric polynomial with bandwidth $5\times 5$, which is the product of two trigonometric polynomials with bandwidth $3\times 3$. According to the sampling theorem, we totally need at least 24 samples and each of the components needs to be sampled for at least 8 samples. We first choose 7 samples (red dots) on the first component and 17 samples (red circles) on the second one. The gray curve in (b) is the original curve and the blue dashed curve is what we recovered from the $7+17=24$ samples. Since the sampling condition is not satisfied, the recovery failed. In (c), we choose 8 samples (red dots) on the first component and 16 samples (red circles) on the second one. From (c), we see that the gray curve (the original curve) overlaps the blue dashed curve (recovered curve), meaning that we recovered the curve successfully. In (d), we choose 17 samples on the first component and 7 samples on the other one. From (d), we see that the recovery is not successful. In (e), we have 16 samples on the first component and 8 samples on the second one. The original curve overlaps the recovered one. So we recovered it perfectly. Lastly, we choose 8 samples on each of the component and we failed to recover the curve as shown in (f). Note that the recovered curves pass through the samples in all cases.}
\label{illuprop2}
\end{figure}

\subsubsection{Case 3: Non-minimal lifting}\label{nonminimal}
In Section \ref{recoverysection} and \ref{recoverysection2}, we introduced theoretical guarantees for the perfect recovery of the surface in any dimensions. The sampling theorems introduced in Section \ref{recoverysection} and \ref{recoverysection2} assume that we know exactly the bandwidth of the surface or the union of surfaces. However, in practice, the true bandwidth of the surface is usually unknown. We now consider the recovery of the surface, when the bandwidth is over-estimated, or equivalently the lifting is performed assuming $\Gamma \supset \Lambda$. As discussed in Section \ref{nonminlifting}, the dimension of $\mathcal V_{\Gamma}$ is upper bounded by $|\Gamma|-|\Gamma\ominus\Lambda|$, which implies that 
\[\rank(\Phi_\Gamma(\mathbf{X}))\le|\Gamma|-|\Gamma\ominus\Lambda|,\]
where $\Gamma\ominus\Lambda$ represents the number of valid shifts of $\Lambda$ within $\Gamma$ as discussed in Section \ref{nonminlifting}. 

The following two propositions show when the inequality in the rank relation above can be an equality and hence we can recover the surface.

\begin{prop}[Irreducible surface with non-minimal lifting]\label{nonminimal1}
Let $\{\mathbf{x}_1, \cdots, \mathbf{x}_N\}$ be $N$ random samples on the surface $\mathcal{S}[\psi]$, chosen independently. The trigonometric polynomial $\psi(\mathbf{x})$ is irreducible whose true bandwidth is $\Lambda$. Suppose the lifting mapping is performed using bandwidth $\Gamma\supset \Lambda$. Then $\rank(\Phi_\Gamma(\mathbf{X})) = |\Gamma|-|\Gamma\ominus\Lambda|$ for almost all $\psi$, if
\[N\ge |\Gamma|-|\Gamma\ominus\Lambda|.\]
\end{prop}

The proof of this proposition can be found in Appendix \ref{proofrankprop3}.

\begin{prop}[Union of irreducible surfaces with non-minimal lifting]\label{nonminimal2}
Let $\psi(\mathbf{x})$ be a randomly chosen trigonometric polynomial with $M$ irreducible factors, i.e., 
\begin{equation}\label{reduciblepoly}
\psi(\mathbf{x}) = \psi_1(\mathbf{x})\cdots\psi_M(\mathbf{x}).
\end{equation}
Suppose the bandwidth of each factor $\psi_i(\mathbf{x})$ is given by $\Lambda_i$ and the bandwidth of $\psi$ is $\Lambda$. Let $\Gamma_i\supset\Lambda_i$ be the non-minimal bandwidth of each factor $\psi_i(\mathbf{x})$ and $\Gamma\supset\Lambda$ is the bandwidth of the non-minimal lifting. Assume that $\{\mathbf{x}_1,\cdots, \mathbf{x}_N\}$ are $N$ random samples on $\mathcal S[\psi]$ that are chosen independently. Then, the feature matrix $\Phi_{\Lambda}(\mathbf{X})$ will be of rank $|\Gamma|-|\Gamma\ominus\Lambda|$ for almost all $\psi$ if 
\begin{enumerate}
	\item each irreducible factor is randomly sampled with $N_i\ge|\Gamma_i|-|\Gamma_i\ominus\Lambda_i|$ points, and
	\item the total number of samples satisfy $ N \ge|\Gamma|-|\Gamma\ominus\Lambda|$.
\end{enumerate}
\end{prop}

We prove this result in Appendix \ref{proofrankprop4}. Note that in practice, when non-minimal lifting mapping is performed, we then randomly sample approximately $|\Gamma|-|\Gamma\ominus\Lambda|$ positions on $\mathcal{S}$. This random strategy ensures that the samples are distributed to the factors, roughly satisfying the conditions in Proposition \ref{nonminimal2}. We further studied this proposition in Fig. \ref{phasetrans}. We considered several random surfaces obtained by choosing random coefficients, each with different bandwidth and considered their recovery from different number of samples. From which, we obtained the phase transition plot given in Fig. \ref{phasetrans}, which agrees well with the theory.

\begin{figure}[h!]
	\centering
	\includegraphics[width=0.4\textwidth]{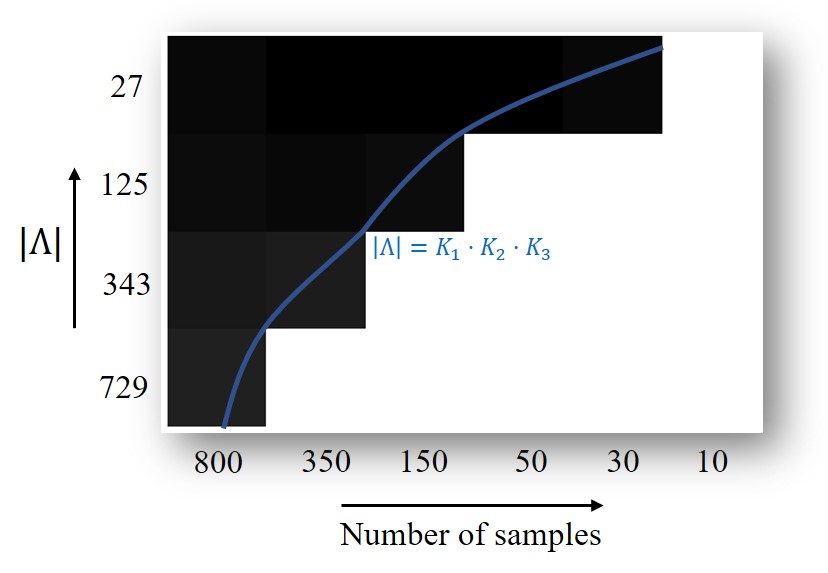}
	\caption{Effect of number of sampled points on surfaces reconstruction error. We randomly generated several surfaces with different bandwidths and number of sampled points, and tried to recover the surfaces from these samples. The reconstruction errors of the surfaces averaged over several trials are shown in the above phase transition plot, as a function of bandwidth and number of sampled entries.  the color black indicates that the true surfaces can be recovered in any of the experiments, while the color white represents that the true surfaces are not recovered in all the experiments. It is seen that we can almost recover the surfaces with $|\Lambda| = k_1\cdot k_2\cdot k_3$ samples.}
	\label{phasetrans}
\end{figure}

\subsection{Surface recovery algorithm for the non-minimal setting}
\label{sos}

The two propositions in Section \ref{nonminimal} show that $\Phi_{\Gamma}(\mathbf{X})$ has $|\Gamma\ominus\Lambda|$ null space basis vectors $\mathbf n_{i} \leftrightarrow \mu_{i}$, when the non-minimal lifting with bandwidth $\Gamma$ is performed. The following result from \cite{zou2019sampling} shows that the null-space vectors are related to the minimal polynomial of the surface. In particular, all null-space vectors have the minimal polynomial as a factor.  We will use this property to extract the surface from the null-space vectors as their greatest common divisor. We also introduce a simpler computational strategy which relies on the sum of squares of the null-space vectors.

\begin{prop}[Proposition 9 in \cite{zou2019sampling}]
\label{shift}
The coefficients of the trigonometric polynomials of the form
\[\theta_{\mathbf{k}}(\mathbf x)=\exp(j2\pi\mathbf{l}^T\mathbf{x})\psi(\mathbf{x}),\quad\forall \mathbf{k}\in\Gamma\ominus\Lambda.\]
 is a null space vector of  $\Phi_{\Gamma}(\mathbf{X})$. 
\end{prop}

Note that the coefficients of $\theta_{{\mathbf k}}(\mathbf x)$ correspond to the shifted versions of the coefficients of $\psi$ and hence are linearly independent. We also note that any such function is a valid annihilating functions for points on $\mathcal S$. When the dimension of the null-space is $|\Gamma\ominus \Lambda|$, these corresponding coefficients form a basis for the null-space. Therefore, we have that any function in the null-space can be expressed as
\begin{eqnarray}
\eta(\mathbf x) &=& \sum_{\mathbf k \in \Gamma \ominus \Lambda} \alpha_{\mathbf k}~\psi(\mathbf x) \exp(j2\pi\mathbf k^{T} \mathbf x)\\
&=& \psi(\mathbf x) \underbrace{\sum_{\mathbf k \in \Gamma \ominus \Lambda} \alpha_{\mathbf k}\exp(j2\pi\mathbf k^{T} \mathbf x)}_{\gamma(\mathbf x)} = \psi(\mathbf x)\gamma(\mathbf x),
\end{eqnarray}
 where $\alpha_{\mathbf k}$ and $\gamma$ are arbitrary coefficients and function, respectively. Note that all of the functions obtained by the null-space vectors have $\psi$ as a common factor. 
 
Accordingly, we have that $\psi(\mathbf{x})$ is the greatest common divisor of the polynomials $\mu_i(\mathbf{x}) \leftrightarrow \mathbf n_{i}$, where $\mathbf n_{i}$ are the null-space vectors of $\Phi_{\Gamma}(\mathbf X)$, which can be estimated using singular value decomposition (SVD). Since we consider polynomials of several variables, it is not computationally efficient to find the greatest common divisor. We note that we are not interested in recovering the minimal polynomial, but are only interested in finding the common zeros of $\mu_{i}(\mathbf x)$. We hence propose to recover the original surface as the zeros of the sum of squares (SoS) polynomial
\[\sigma(\mathbf{x}) = \sum_{i=1}^{|\Gamma\ominus\Lambda|}|\mu_i(\mathbf{x})|^2.\]

Note that rank guarantees in Propositions \ref{proofrankprop3} and \ref{proofrankprop4} ensure that the entire null-space will be fully identified by the feature matrix. Coupled with Proposition \ref{shift}, we can conclude that the recovery using the above algorithm (SVD, followed by the sum of squares of the inverse Fourier transforms of the coefficients) will give perfect recovery of the surface under noiseless conditions. The algorithm is illustrated in Fig. \ref{prop_over}.

\begin{figure}[h!]
	\centering
          \subfloat[The original curve]{\includegraphics[width=0.23\textwidth]{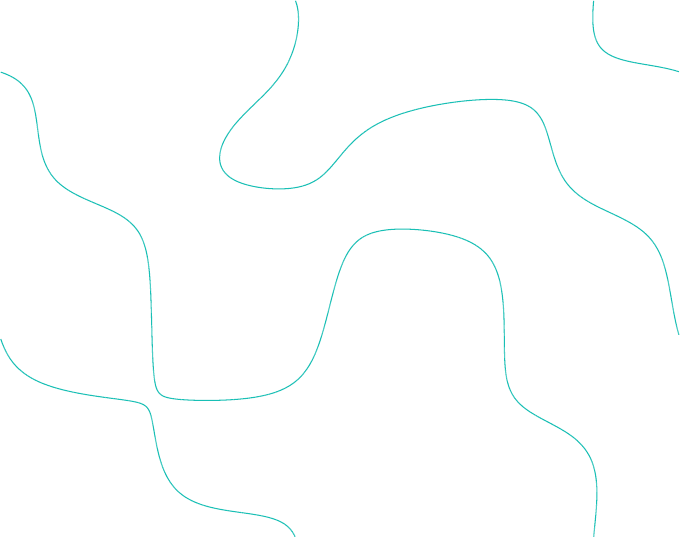}}\hspace{4ex}
          \subfloat[The sampling points]{\includegraphics[width=0.23\textwidth]{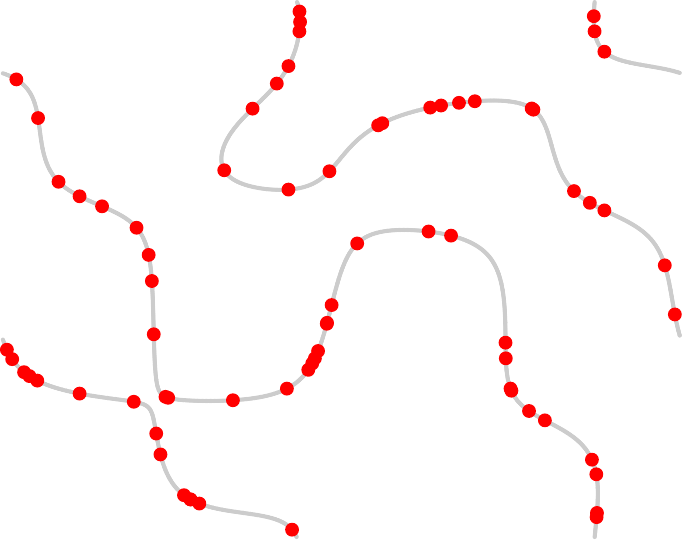}}\\
          \subfloat[1st null-space function]{\includegraphics[width=0.23\textwidth]{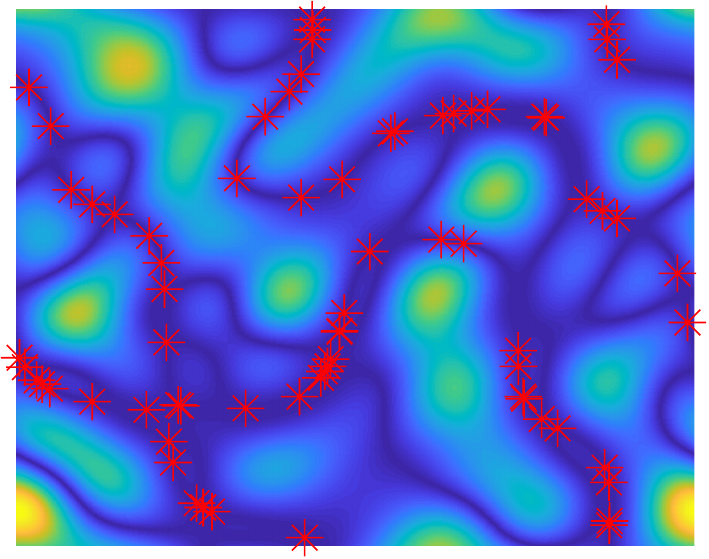}}\hspace{4ex}
          \subfloat[2nd null-space function]{\includegraphics[width=0.23\textwidth]{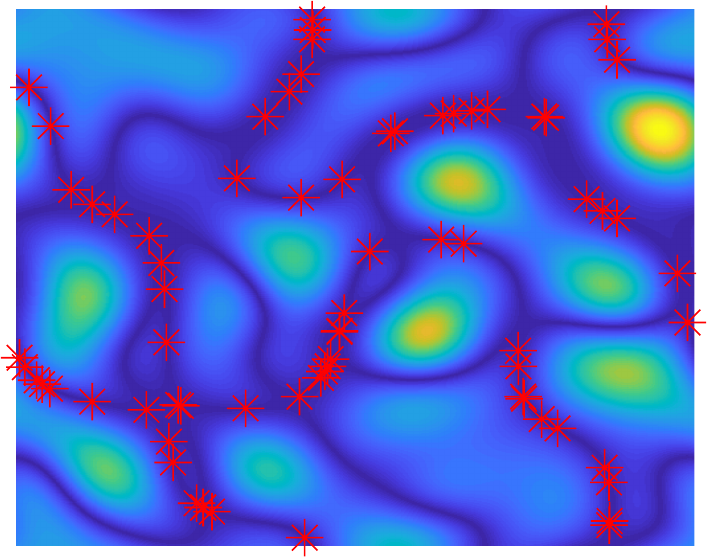}}\hspace{4ex}
          \subfloat[sum-of-squares polynomial]{\includegraphics[width=0.23\textwidth]{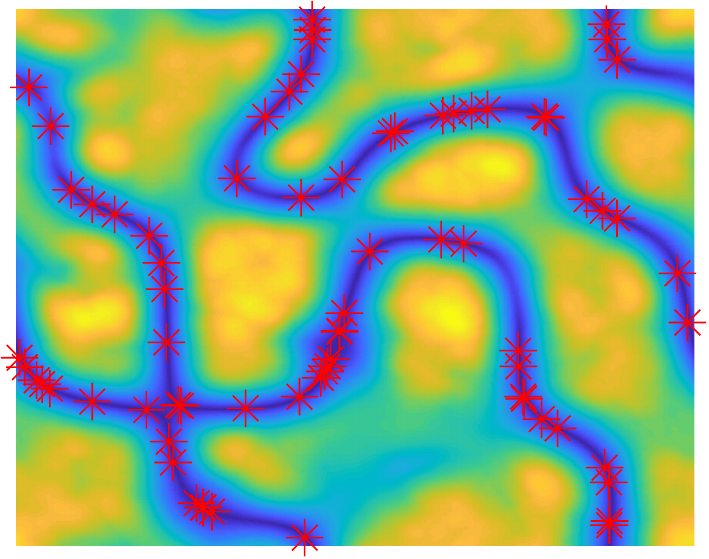}}\\
\caption{Illustration of the sampling fashion for non-minimal bandwidth. We consider the curve as shown in (a), which is given by the zero level set of a trigonometric polynomial of bandwidth $5\times 5$. We choose the non-minimal bandwidth $\Gamma$ as $11\times 11$. According to the sampling condition for non-minimal bandwidth, we sampled on 72 random locations. We randomly chose two null-space vectors for the feature matrix of the sampling set, which gave us functions (c) and (d). We can see that all of these functions have zeros on the original zero set, in addition to processing several other zeros. The sum of squares function is shown in (e), showing the common zeros, which specifies the original curve.}
\label{prop_over}
\end{figure}

\subsection{Surface recovery from noisy samples}\label{noisy}
The analysis in Section \ref{nonminimal} shows that when the bandwidth of the surface is small, the feature matrix is low rank. In practice, the sampling points are usually corrupted with some noise. We denote the noisy sampling set by $\mathbf{Y}=\mathbf{X}+\mathbf{N}$, where $\mathbf{N}$ is the noise. We propose to exploit the low-rank nature of the feature matrix to recover it from noisy measurements. Specifically, when the sampling set $\mathbf{X}$ is corrupted by noise, the points will deviate from the original surface, and hence the features will cease to be low rank. We impose a nuclear norm penalty on the feature maps that will push the feature vectors to a subspace. Since the feature vectors are related to the original points by the exponential mapping, the original points will move to the surface. In practice it is difficult to compute the feature map. We hence rely on an iterative reweighted least-squares algorithm, coupled with the \emph{kernel-trick}, to avoid the computation of the features. Since the cost function is non-linear (due to the non-linear kernel), we use steepest descent-like algorithm to minimize the cost function. We note that each iteration of this algorithm has similarities to non-local means algorithms, which first estimate the weight/Laplacian matrix from the patches, followed by a smoothing. We also note that this approach has conceptual similarities to kernel low-rank algorithms used in MRI and computer vision \cite{muller2001introduction,nakarmi2017kernel}.
These algorithms rely on explicit polynomial mappings, low-rank approximation of the features, followed by the analytical evaluation of the pre-images that is possible for polynomial kernels.

We pose the denoising as:
\begin{equation}\label{nucnorm}
\mathbf{X}^*=\arg\min_{\mathbf X}||\mathbf X-\mathbf Y||^2+\lambda||\Phi(\mathbf{X})||_{*}
\end{equation}
where we use  the nuclear norm of the feature matrix of the sampling set as a regularizer. Unlike traditional convex nuclear norm formulations, the above scheme is non-convex. 

We adapt the kernel low-rank algorithm in \cite{poddar2018recovery,ongie2017algebraic} to the high dimensional setting to solve \eqref{nucnorm}. This algorithm relies on an iteratively reweighted least squares (IRLS) approach \cite{fornasier2011low,mohan2012iterative} which alternates between the following two steps:
\begin{equation}\label{min_irls}
\mathbf{X}^{(n)}=\arg\min_{\mathbf X}||\mathbf{X}-\mathbf{Y}||^2+\lambda{\mathrm{trace}}[\mathcal{K}(\mathbf{X})\mathbf{P}^{(n-1)}],
\end{equation}
and
\begin{equation}
\label{pn}
\mathbf{P}^{(n)}= \left[\mathcal{K}(\mathbf{X}^{(n)}) +\gamma^{(n)}\mathbf{I}\right]^{-1/2}
\end{equation}
where $\gamma^{(n)}=\frac{\gamma^{(n-1)}}{\eta}$ and $\eta>1$ is a constant. Here, $\mathcal{K}(\mathbf{X}) = \Phi_{\Gamma}(\mathbf X)^T\Phi_{\Gamma}(\mathbf X)$. We use the \emph{kernel-trick} to evaluate $\mathcal K\left(\mathbf X\right)$. The kernel-trick suggests that we do not need to explicitly evaluate the features. Each entry of the matrices $\mathcal K\left(\mathbf X\right)$ correspond to inner-products in feature space:
\begin{eqnarray}
\label{kernelmtx}
\left(\mathcal K\left(\mathbf X\right)\right)_{(i,j)} &=& \underbrace{\Phi(\mathbf{x}_i)^H \Phi(\mathbf{x}_j)}_{\kappa(\mathbf x_i,\mathbf x_j)}
\end{eqnarray}
which can be evaluated efficiently using the nonlinear function $\kappa$ (termed as kernel function) of their inner-products in $\mathbb R^n$. 

The dependence of the kernel function on the lifting is detailed in Section \ref{kerneltrick}. Since the above problem in \eqref{min_irls} is not quadratic, we propose to solve it using gradient descent as in \cite{zou2019sampling}. We note that the cost function in \eqref{min_irls} can be rewritten as 
	\begin{equation}\label{cost}
	C(\mathbf X) = \|\mathbf X-\mathbf Y\|^2 + \lambda ~ \sum_{i,j}~\mathbf P_{ij}^{(n-1)} ~ \kappa\left(\mathbf x_i,\mathbf x_j\right),
	\end{equation}	
	where $\mathbf{P}_{i,j}$ are the entries of the matrix $\mathbf P^{(n-1)}$.
	As will be discussed in detail in Section \ref{kerneltrick}, the exponential kernel for a circular support as in Fig. \ref{kernelshape}.(b) can be approximated as a circularly symmetric kernel $\kappa(\mathbf x_i,\mathbf x_j) = k(\|\mathbf x_i-\mathbf x_j\|^2)$. 
In this case, the partial derivatives of \eqref{min_irls} with respect to one of the vectors $\mathbf x_i$ is 
\begin{eqnarray}
	\partial_{\mathbf x_i}\mathcal C &=& 2(\mathbf x_i - \mathbf y_i) + 2\lambda \sum_j \underbrace{\mathbf P_{ij}^{(n-1)} ~k'(\|\mathbf x_i-\mathbf x_j\|^2)}_{w_{i,j}}~(\mathbf x_i-\mathbf x_j)\\
	&=& 2(\mathbf x_i - \mathbf y_i) + 2\lambda \underbrace{\left(\sum_j w_{i,j}\right)}_{d_i}\mathbf x_i - \mathbf W\mathbf X.
\end{eqnarray}
Here, 
\begin{equation}
\label{wn}
\mathbf W_{ij} = \mathbf P_{ij}^{(n-1)} ~k'(\|\mathbf x_i-\mathbf x_j\|^2.
\end{equation}
Thus, the gradient of the cost function \eqref{cost} is :
\begin{equation}\label{gradient}
\nabla_{\mathbf X} \mathcal C \approx 2(\mathbf X - \mathbf Y) +2\lambda \underbrace{
	\left(\mathbf D-\mathbf W\right)}_{\mathbf L} \mathbf X.
\end{equation}
Here, $\bf{L} = \bf{D}-\bf{W}$ is the matrix obtained from the weights $\mathbf W$ and  $\bf{D}$ is a diagonal matrix with diagonal entries $d_{i} = \sum_{j}{\bf{W}}_{ij}$.

We note that the gradient of \eqref{min_irls} specified by \eqref{gradient} is also the gradient of the cost function
\begin{equation}
\label{Req}
\mathcal D = \left\|\mathbf X - \mathbf Y\right\|_F^2 + \lambda~{\rm trace}\left(\mathbf X~ \mathbf L ~\mathbf X^{H} \right),
\end{equation}
which is used in approaches such as non-local means (NLM) \cite{buades2011non} and graph regularization \cite{smola2003kernels}. We note that the above optimization problem is quadratic and hence has an analytical solution. We thus alternate between the solution of \eqref{Req} and updating the weights, and hence the Laplacian matrix using \eqref{wn}, where $\mathbf P$ is specified by \eqref{pn}. Despite the similarity to NLM, we note that NLM approaches use a fixed Laplacian unlike the iterative approach in our work. In addition, the expression of the Laplacian is also very different. We refer the readers to \cite{zou2019sampling} for comparison of the proposed scheme with the above graph regularized algorithm. Once the denoised null-space matrix is obtained from the above algorithm, we can use the sum of square approach described in Section \ref{nonminimal} to recover the surfaces. We note that the algorithm is not very sensitive to the true bandwidth of the kernel $\Gamma$, as long as it over-estimates the true bandwidth of the surface $\Lambda$.

We illustrate this approach in the context of recovering 3D shapes from noisy point clouds in Fig. \ref{illu3dsurf}. The data sets are obtain from AIM@SHAPE \cite{AIM}. We note that the direct approach, where the null-space vector is calculated from the noisy feature matrix, often results in perturbed shapes. By contrast, the nuclear norm prior is able to regularize the recovery.

\begin{figure}[!h]
\centering
\subfloat[1500 samples]{\includegraphics[width=0.17\textwidth]{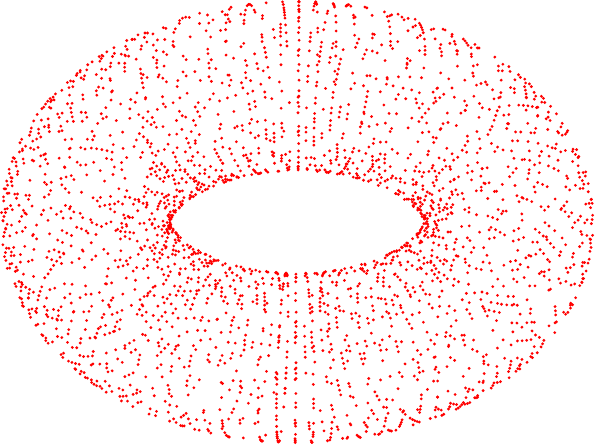}}\hspace{4em}
\subfloat[2500 samples]{\includegraphics[width=0.155\textwidth,angle=90]{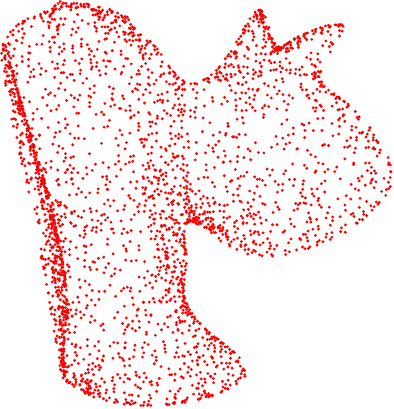}}\hspace{4em}
\subfloat[2500 samples]{\includegraphics[width=0.18\textwidth]{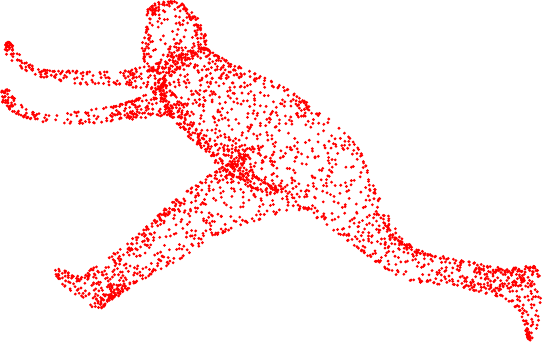}}\hspace{3em}\\
\subfloat[Noisy samp.]{\includegraphics[width=0.15\textwidth]{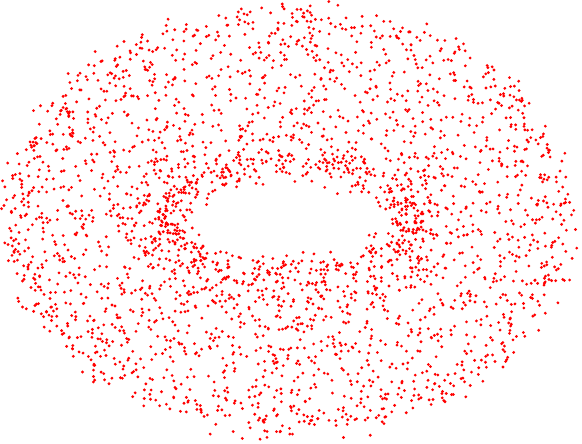}}\hspace{1.5em}
\subfloat[Denoised samp.]{\includegraphics[width=0.15\textwidth]{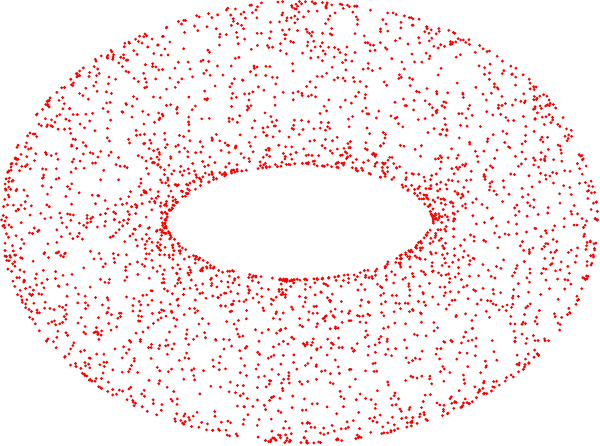}}\hspace{1.5em}
\subfloat[Noisy recon.]{\includegraphics[width=0.15\textwidth]{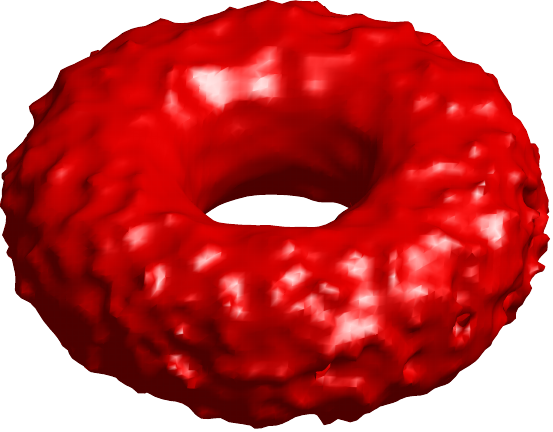}}\hspace{1.5em}
\subfloat[Denoised recon.]{\includegraphics[width=0.15\textwidth]{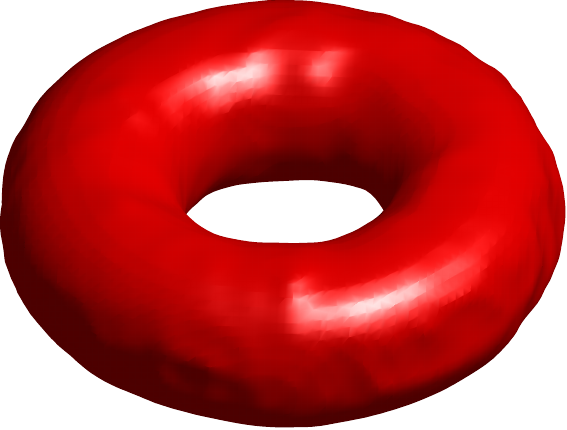}}\hspace{1.5em}
\subfloat[Original recon.]{\includegraphics[width=0.15\textwidth]{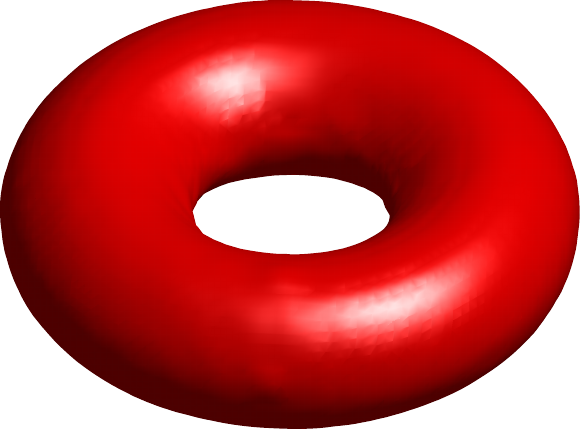}}\\
\subfloat[Noisy samp.]{\includegraphics[width=0.14\textwidth,angle=90]{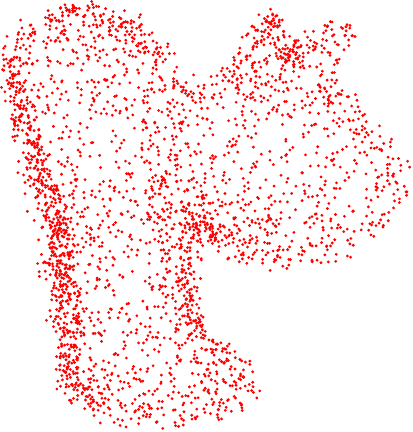}}\hspace{1.55em}
\subfloat[Denoised samp.]{\includegraphics[width=0.14\textwidth,angle=90]{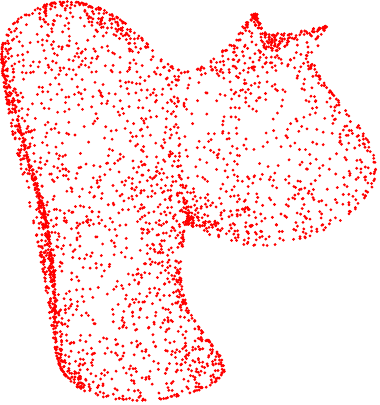}}\hspace{1.55em}
\subfloat[Noisy recon.]{\includegraphics[width=0.14\textwidth,angle=90]{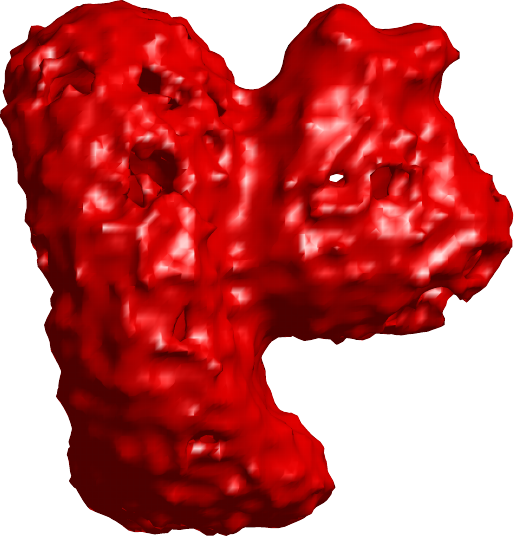}}\hspace{1.55em}
\subfloat[Denoised recon.]{\includegraphics[width=0.14\textwidth,angle=90]{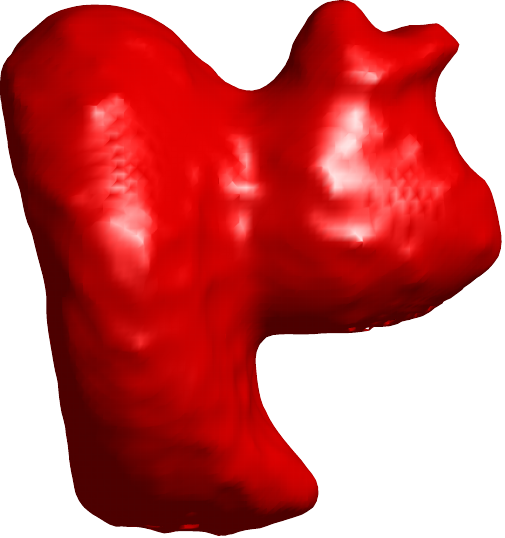}}\hspace{1.55em}
\subfloat[Original recon.]{\includegraphics[width=0.14\textwidth,angle=90]{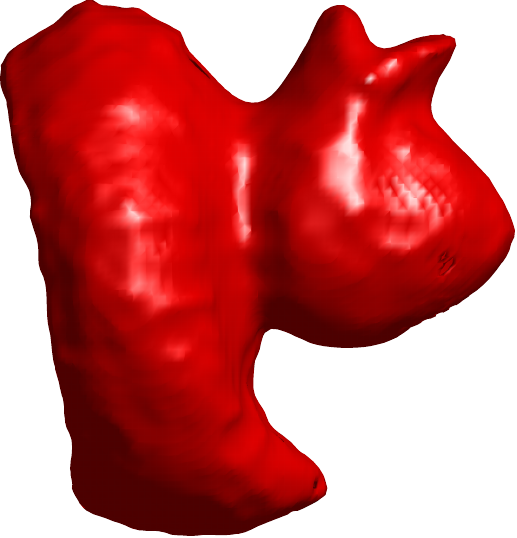}}\\
\subfloat[Noisy samp.]{\includegraphics[width=0.16\textwidth]{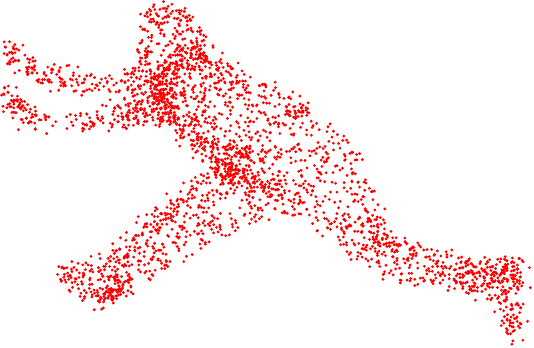}}\hspace{1.1em}
\subfloat[Denoised samp.]{\includegraphics[width=0.16\textwidth]{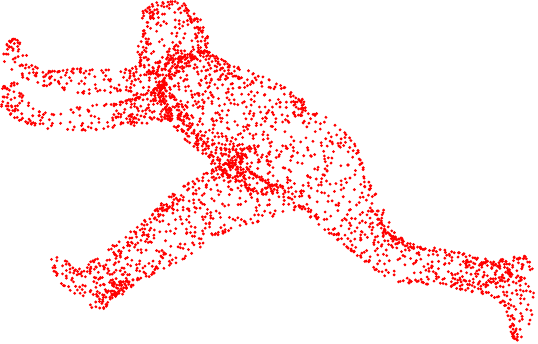}}\hspace{1.1em}
\subfloat[Noisy recon.]{\includegraphics[width=0.16\textwidth]{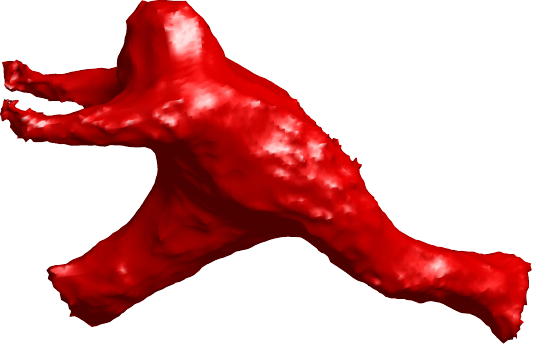}}\hspace{1.1em}
\subfloat[Denoised recon.]{\includegraphics[width=0.16\textwidth]{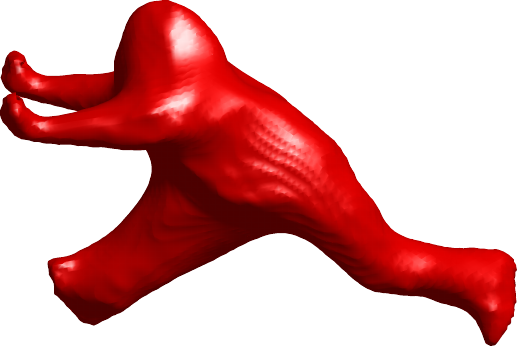}}\hspace{1.1em}
\subfloat[Original recon.]{\includegraphics[width=0.16\textwidth]{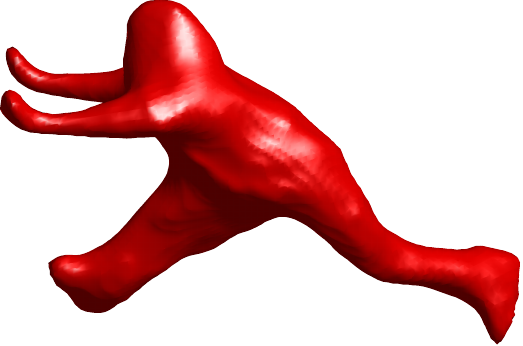}}\\
\caption{Illustration of the points cloud denoising algorithm and surface recovery algorithm with unknown bandwidth. The first row shows the samples drawn from three surfaces. Noise is added to the samples (see (d), (i), (n)). Then we use the proposed algorithm to denoise the points. The parameter $\lambda$ in \eqref{nucnorm} is chosen as 1.4 for the denoising algorithm. The number of iterations for the denoising algorithm is 30. The surfaces that are recovered from noisy samples and denoised samples are also presented for comparison. The bandwidth was chosen as $31\times 31\times 31$ for all the experiments.}
\label{illu3dsurf}
\end{figure}



\section{Recovery of functions on surfaces}
\label{funlearn}
As discussed in the introduction, modern machine learning algorithms pre-learn functions from given input and output data pairs \cite{koulamas2007single}. For example, CNN based denoising approaches that provide state-of-the-art results essentially learn to generate noise-free pixels or patches from given training data with several noisy and noise-free patch pairs \cite{zhang2017beyond,tian2019enhanced}. The problem can be formulated as estimating a nonlinear function $\mathbf y=f(\mathbf x)$, given input and output data pairs $(\mathbf{x}_i,\mathbf y_i);i=1,..,N_{\rm train}$. A challenge in the representation of such high dimensional function is the large number of parameters, which is also termed as the curse of dimensions. Kernel methods \cite{pillonetto2014kernel}, random forests \cite{schroff2008object} and neural networks \cite{zhang1995wavelet} provide a powerful class of machine learning models that can be used in learning highly nonlinear functions. These models have been widely used in many machine learning tasks \cite{hastie2005elements}.



We now show that the results shown in the previous sections provide an attractive option to compactly represent functions, when the data lie on a smooth surface or manifold in high dimensional spaces. We note that the manifold assumption is widely assumed in a range of machine learning problems \cite{gallese2003roots, fefferman2016testing}. We now show that if the data lie on a smooth surface in high dimensional space, one can represent the multidimensional functions very efficiently using few parameters. 

We model the function using the same basis functions used to represent the level set function. In our case\footnote{We note that similar results can be obtained when the function $f$ and the level set function are represented as a linear combination of shift-invariant functions or polynomials.}, we model it as a band-limited multidimensional function:
\begin{equation}\label{blfn}
f(\mathbf{x})=\sum_{\mathbf{k}\in\Gamma}\beta_{\mathbf{k}}\exp(j2\pi\mathbf{k}^T\mathbf{x}) = \boldsymbol{\beta}^T\Phi_{\Gamma}(\mathbf{x}), 
\end{equation}
where $\mathbf x \in \mathbb R^n$. 
The number of free parameters in the above representation is $|\Gamma|$, where $\Gamma \subset \mathbb Z^n$ is the bandwidth of the function. Note that $|\Gamma|$ grows rapidly with the dimension $n$. The large number of parameters needed for such a representation makes it difficult to learn such functions from few labeled data points. We now show that if the points lie on the union of irreducible surfaces as in \eqref{uim}, where the bandwidth of $\psi$ is given by $\Lambda\subset\Gamma$, we can represent functions of the form \eqref{blfn} efficiently.

\subsection{Compact representation of features using anchor points}
We use the upper bound of the dimension of the feature matrix in \eqref{dimnonmin} to come up with an efficient representation of functions of the form \ref{blfn}. The dimension bound \eqref{dimnonmin} implies that the features of points on $S[\psi]$ lie in a subspace of dimension $r=|\Gamma|-|\Gamma\ominus\Lambda|$, which is far smaller than $|\Gamma|$ especially when the dimension $n$ is large. We note that kernel methods often approximate the feature space using few eigen vectors of kernel PCA. However, there is no guarantee that these basis vectors are mappings of some points on $\mathcal S$. Hence, it is a common practice to consider all the training  samples to capture the low-dimensional feature vectors in kernel PCA. We now show that it is possible to find a set of $N\geq r$ anchor points $\mathbf{a}_1,\cdots ,\mathbf{a}_N \in \mathcal S[\psi]$, such that the feature space $\mathcal V_{\Gamma}(\mathcal S)$ is in ${\rm{span}}\{\Phi_{\Gamma}(\mathbf{a}_1),\cdots,\Phi_{\Gamma}(\mathbf{a}_N)\}$. This result is a Corollary of Proposition \ref{nonminimal2}.
\begin{cor}
	\label{corrollary}
Let $\psi(\mathbf{x})$ be a randomly chosen trigonometric polynomial with $M$ irreducible factors as in \eqref{reduciblepoly}. Suppose $\Gamma_i\supset\Lambda_i$ is the non-minimal bandwidth of each factor $\psi_i(\mathbf{x})$ and $\Gamma\supset\Lambda$ is the total bandwidth. Let $\{\mathbf{a}_1,\cdots, \mathbf{a}_N\}$ be $N$ randomly chosen anchor points on $\mathcal S[\psi]$ satisfying
\begin{enumerate}
	\item each irreducible factor $\mathcal S[\psi_i]$ is sampled with $N_i\ge|\Gamma_i|-|\Gamma_i\ominus\Lambda_i|$ points, and
	\item the total number of samples satisfy $ N \ge|\Gamma|-|\Gamma\ominus\Lambda|$.
\end{enumerate}
Then, 
	\begin{equation}\label{dimbound}
	\mathcal V_{\Gamma}(\mathcal S) \subseteq {\rm span}\left\{ \Phi_{\Gamma}(\mathbf a_i); i=1,\cdots,N \right\}
	\end{equation}
with probability 1.
\end{cor}

As discussed in Section \ref{nonminimal}, if we randomly choose $N \ge |\Gamma|-|\Gamma\ominus\Lambda| =  r$ points on  $\mathcal{S}[\psi]$, the feature matrix will satisfy the conditions in Corollary \ref{corrollary} and hence \eqref{dimbound} with unit probability. This relation implies that the feature vector of any point $\mathbf x \in \mathcal S[\psi]$ can be expressed as the linear combination of the features of the anchor points $\Phi_{\Gamma}(\mathbf a_i); i=1,\cdots,N$:
\begin{eqnarray}
\Phi_{\Gamma}(\mathbf x) &=& \sum_{i=1}^{N} \mathbf \alpha_i(\mathbf x) ~\Phi_{\Gamma}(\mathbf a_i)\\\label{linearcomb}
&=& \underbrace{\begin{bmatrix}\Phi_{\Gamma}(\mathbf{a}_1)  & \cdots & \Phi_{\Gamma}(\mathbf{a}_{N})\end{bmatrix}}_{\Phi(\mathbf{A})} \underbrace{\begin{bmatrix} \alpha_1(\mathbf{x}) \\  \vdots \\ \alpha_{N}(\mathbf{x})\end{bmatrix}}_{\boldsymbol \alpha(\mathbf x)}
\end{eqnarray}
Here, $\alpha_i(\mathbf x)$ are the coefficients of the representation.
Note that the complexity of the above representation is dependent on $N$, which is much smaller than $|\Gamma|$, when the surface is highly band-limited. We note that the above compact representation is exact only for $\mathbf x \in \mathcal{S}[\psi]$ and not for arbitrary $\mathbf x\in \mathbb R^{n}$; the representation in \eqref{linearcomb} will be invalid for $\mathbf x \notin \mathcal{S}[\psi]$. 

However, this direct approach requires the computation of the high dimensional feature matrix, and hence may not be computationally feasible for high dimensional problems. We hence consider the normal equations and solve for $\boldsymbol\alpha(\mathbf x)$ as 
\begin{equation}\label{coeffs}
\boldsymbol\alpha(\mathbf x) = \left(\underbrace{\Phi(\mathbf{A})^H\Phi(\mathbf{A})}_{\mathcal K\left(\mathbf A\right)}\right)^{\dag} \underbrace{\left(\Phi(\mathbf{A})^H\Phi_{\Gamma}(\mathbf x)\right)}_{\mathbf k_{\mathbf A}(\mathbf x)},
\end{equation}
where $(\cdot)^\dag$ denotes the pseudo-inverse.

\subsection{Representation and learning of functions}
Using \eqref{blfn}, \eqref{linearcomb}, and \eqref{coeffs}, the function $\mathbf f:\mathbb R^n \rightarrow \mathbb R^m$ can be written as
\begin{eqnarray}\label{fnrep}
\mathbf f(\mathbf{x})&=& \boldsymbol{\beta}^T~\Phi({\mathbf A}) ~\mathcal K\left(\mathbf A\right)^{\dag} ~\mathbf k_{\mathbf A}(\mathbf x)\\\label{mtxform}
&=& \underbrace{\left[\overbrace{\boldsymbol{\beta}^T~\Phi_{\Gamma}(\mathbf a_1)}^{\mathbf f(\mathbf a_1)},\ldots,\overbrace{\boldsymbol{\beta}^T~\Phi_{\Gamma}(\mathbf a_N)}^{\mathbf f(\mathbf a_N)}\right]}_{\mathbf F} ~\underbrace{\mathcal K\left(\mathbf A\right)^{\dag} ~\mathbf k_{\mathbf A}(\mathbf x)}_{\boldsymbol{\alpha}(\mathbf x)}
\end{eqnarray}
Here, $\mathbf f(\mathbf x)$ is an $M\times 1$ vector, while $\mathbf F$ is an $M\times N$ matrix. $\mathcal{K}(\mathbf{A})$ is an $N\times N$ matrix and $\mathbf {k_{A}}(\mathbf x)$ is an $N\times 1$ vector. Thus, if the function values at the anchor points, specified by $\mathbf f(\mathbf a_i); i=1,\cdots,N$ are known, one can compute the function for any point $\mathbf x\in \mathcal S[\psi]$. 

We note that the direct representation of a function $f:\mathbb R^n\rightarrow \mathbb R$ in \eqref{blfn} requires $|\Gamma|$ parameters, which can be viewed as the area of the green box in Fig. \ref{ominusfig}. By contrast, the above representation only requires $|\Gamma|\ominus|\Gamma:\Lambda|$ anchor points, which can be viewed as the area of the gray region in Fig. \ref{ominusfig}. The more efficient representation allows the learning of complex functions from few data points, especially in high dimensional applications.  

We demonstrate the above local function  representation result in a 2D setting in Fig. \ref{func_curve}. Specifically, the original band-limited function is with bandwidth $13\times 13$. The direct representation of the function has $13\times 13 = 169$ degrees of freedom. Now, if we only care about points on a curve which is with bandwidth $3\times 3$, then the same function living on the curve can be represented exactly using 48 anchor points, thus significantly reducing the degrees of freedom. However, note that the above representation is only exact on the curve. We note that the function goes to zero as one moves away from the curve.

\begin{figure}[!h]
		\centering
		\subfloat[Curve]{\includegraphics[width=0.21\textwidth]{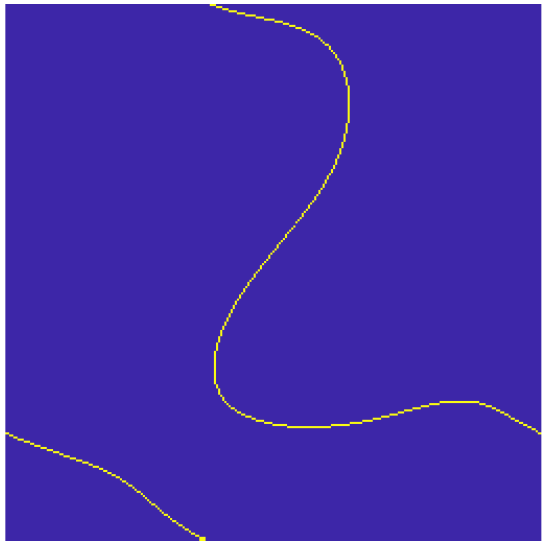}}\hspace{4ex}
		\subfloat[band-limited function]{\includegraphics[width=0.25\textwidth]{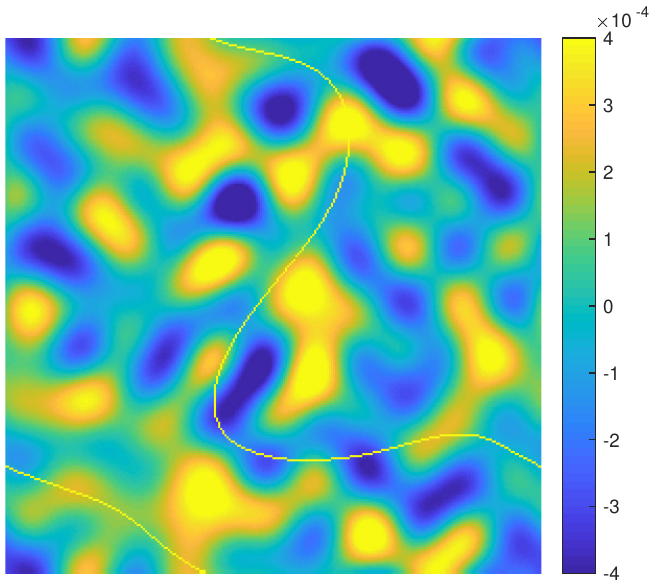}}\hspace{4ex}
		\subfloat[Function on curve]{\includegraphics[width=0.25\textwidth]{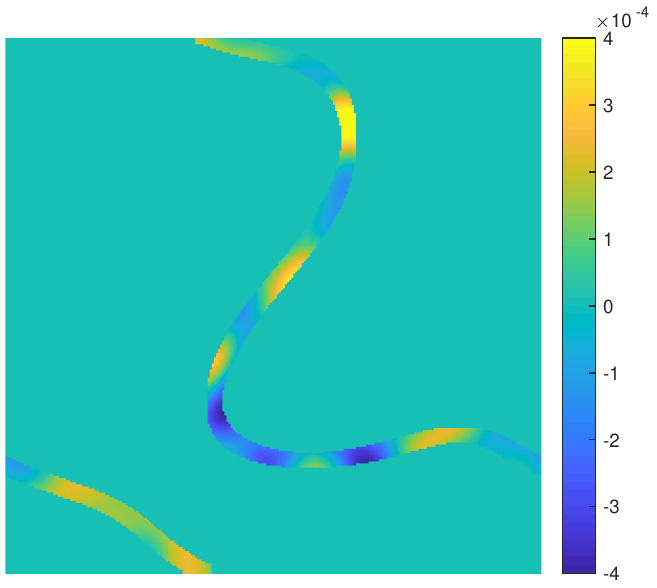}}\\
		\subfloat[Anchor points]{\includegraphics[width=0.21\textwidth]{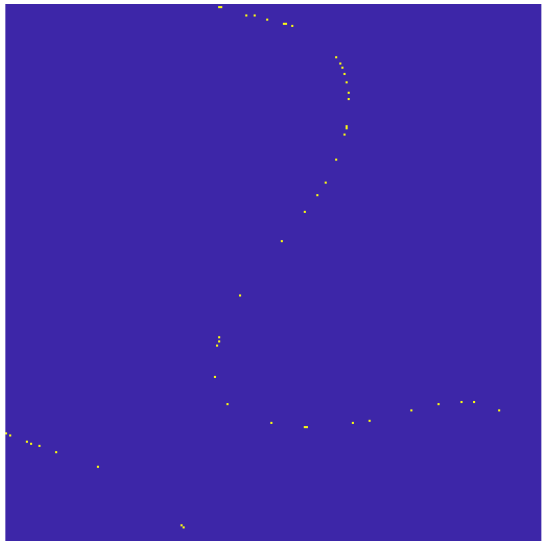}}\hspace{4ex}
		\subfloat[Approximation]{\includegraphics[width=0.25\textwidth]{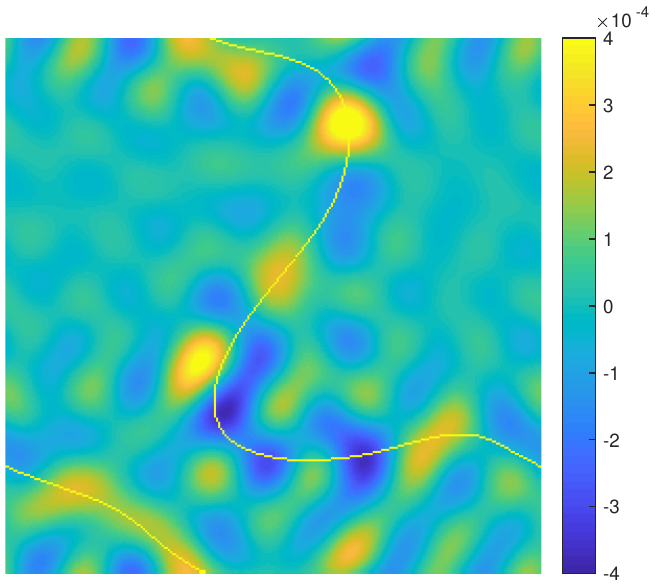}}\hspace{4ex}
		\subfloat[Approx on curve]{\includegraphics[width=0.25\textwidth]{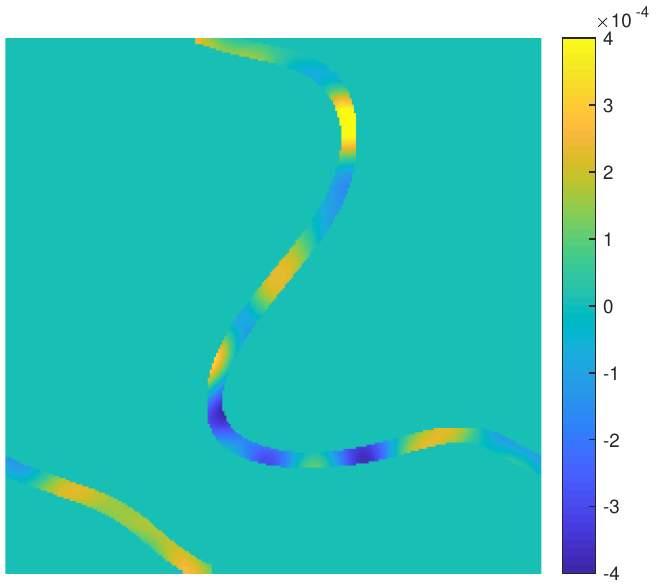}}
		\caption{Illustration of the local representation of functions in 2D. We consider the local approximation of the band-limited function in (b) with a bandwidth of $13\times13$, living on the band-limited curve shown in (a). The bandwidth of the curve is $3\times3$. The curve is overlaid on the function in (b) in yellow. The restriction of the function to the vicinity of the curve is shown in (c). Our results suggest that the local function approximation requires $13^2 - 11^2 = 48$ anchor points. We randomly select the points on the curve, as shown in (d). The interpolation of the function values at these points yields the global function shown in (e). The restriction of the function to the curve in (f) shows that the approximation is good.}
\label{func_curve}
	\end{figure}

The choice of anchor points depends on the geometry of the surface, including the number of irreducible components. For arbitrary training samples, we can estimate the unknowns $\mathbf F$ in \eqref{mtxform} from the linear relations
\begin{eqnarray}\label{key}
\underbrace{\left[\mathbf y_1,..\mathbf y_P\right]}_{\mathbf Y} &=& \mathbf F~ \underbrace{[\boldsymbol \alpha(\mathbf x_1),\ldots,\boldsymbol \alpha(\mathbf x_P)]}_{\mathbf Z}
\end{eqnarray}
as $\mathbf F = \mathbf Y\mathbf Z^H\left(\mathbf Z\mathbf Z^H\right)^{\dag}$. The above recovery is exact when we have $N=r$ achor points  because $\mathbf{Z}$ has full column rank in this case. The reason why $\mathbf{Z}$ has full column rank is due to \eqref{linearcomb} and \eqref{coeffs}. Equation \eqref{linearcomb} suggests that ${\rm{rank}}(\mathbf{Z}) \ge N$, while equation \eqref{coeffs} shows ${\rm{rank}}(\mathbf{Z}) \le N$. Therefore, we have ${\rm{rank}}(\mathbf{Z}) = N$, indicating that $\mathbf{Z}$ has full rank in this case. When $N>r$, the $\mathbf{F}$ is obtained using the pseudo-inverse, which is based on the least square approximation.

\subsection{Efficient computation using \emph{kernel trick}}
\label{kerneltrick}
We use the \emph{kernel-trick} to evaluate $\mathcal K\left(\mathbf A\right)$ and $\mathbf k_{\mathbf A}(\mathbf x)$, thus eliminating the need to explicitly evaluating the features of the anchor points and $\mathbf x$. Each entry of the matrix $\mathcal K\left(\mathbf A\right)$ is computed as in \eqref{kernelmtx}, while the vector 
$\mathbf{k}_{\mathbf A}(\mathbf x)$ is specified by:
\begin{eqnarray}\label{proj}
(\mathbf{k}_{\mathbf A}(\mathbf x))_i &=& \underbrace{\Phi_{\Gamma}(\mathbf{a}_i)^H \Phi_{\Gamma}(\mathbf{x})}_{\kappa(\mathbf a_i,\mathbf x)},
\end{eqnarray}
which can be evaluated efficiently as nonlinear function $\kappa$ (termed as kernel function) of their inner-products in $\mathbb R^n$. We now consider the kernel function $\kappa$ for specific choices of lifting.

Using the lifting in \eqref{lifting}, we obtain the kernel as 
\[\kappa(\mathbf{x},\mathbf{y})= \sum_{\mathbf{k}\in\Gamma}\exp(j2\pi\mathbf{k}^T(\mathbf{y}-\mathbf{x})).\]
Note that the kernel is shift invariant in this setting. Since $\kappa: \mathbb R^n\rightarrow \mathbb R$ is an $n$ dimensional function, evaluating and storing it is often challenging in multidimensional applications. We now focus on approximating the kernel efficiently for fast computation. We consider the impact of the shape of the bandwidth set $\Gamma$ on the shape of the kernel. Specifically, we consider sets of the form
\begin{equation}
\Gamma = \{\mathbf k \in\mathbb{Z}^n,||\mathbf{k}||_q\le d \},
\end{equation}
where $d$ denotes the size of the bandwidth. The integer $q$ specifies the shape of $\Gamma$ \cite{weisz2012summability}, which translates to the shape of the kernel
\begin{equation}\label{dirich}
k_{d,n}^q(\mathbf{x}):=\sum_{\mathbf{k}\in\mathbb{Z}^n,||\mathbf{k}||_q\le d}\exp(j2\pi \mathbf{k}^T\mathbf{x}).
\end{equation}
 We term the $q=1$ case as the diamond Dirichlet kernel. If $q=2$, we call it the circular Dirichlet kernel. We call the Dirichlet kernel the cubic Dirichlet kernel if $q=\infty$. See Figure \ref{kernelshape} for the bandwidth and Figure \ref{kerneltype} to see the associated kernel.
\begin{figure}[!h]
	\centering
	\subfloat[$q=1$]{\includegraphics[width=0.21\textwidth]{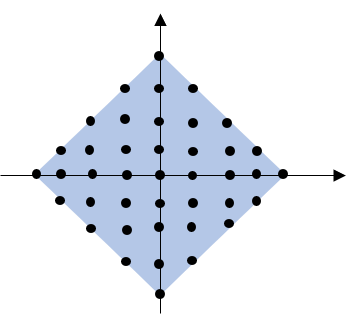}}\hspace{5em}
	\subfloat[$q=2$]{\includegraphics[width=0.21\textwidth]{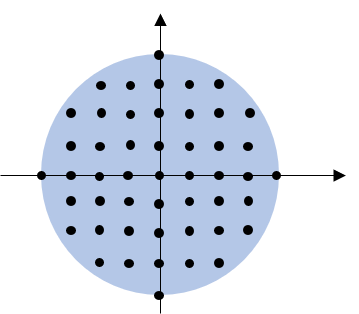}}\hspace{5em}
	\subfloat[$q=\infty$]{\includegraphics[width=0.21\textwidth]{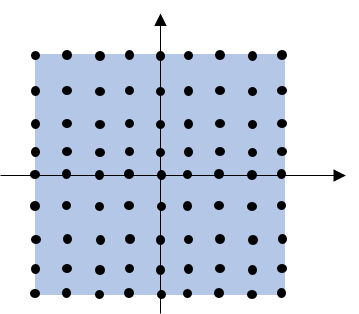}}
	\caption{bandwidth of the set $\Lambda$ with different $q$ values.}
	\label{kernelshape}
\end{figure}

\begin{figure}[!h]
	\centering
	\subfloat[Gaussian kernel]{\includegraphics[width=0.2\textwidth]{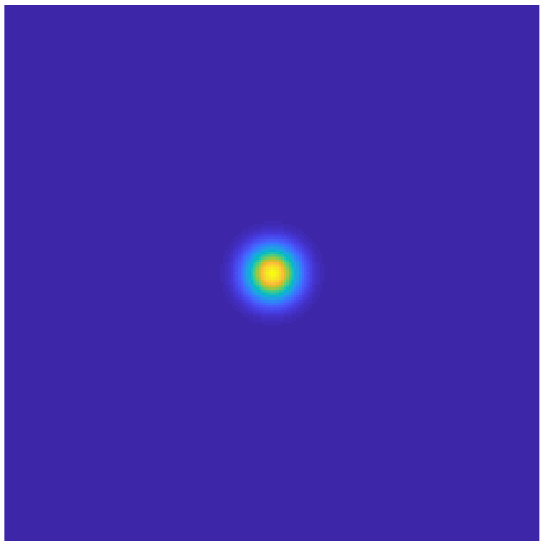}}\hspace{2em}
	\subfloat[Dirichlet with $q = 2$]{\includegraphics[width=0.2\textwidth]{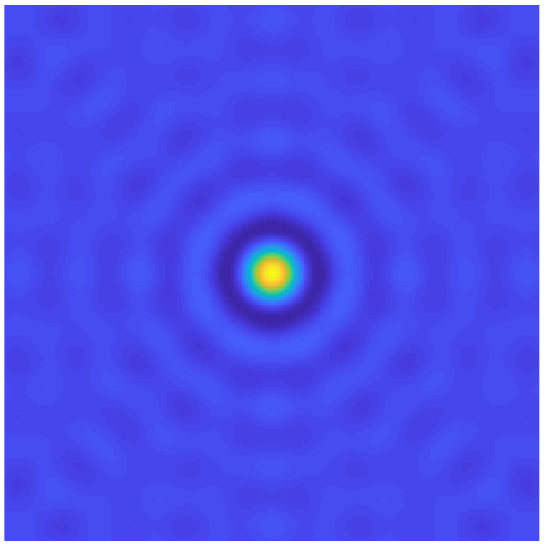}}\hspace{2em}
	\subfloat[Dirichlet with $q=\infty$]{\includegraphics[width=0.2\textwidth]{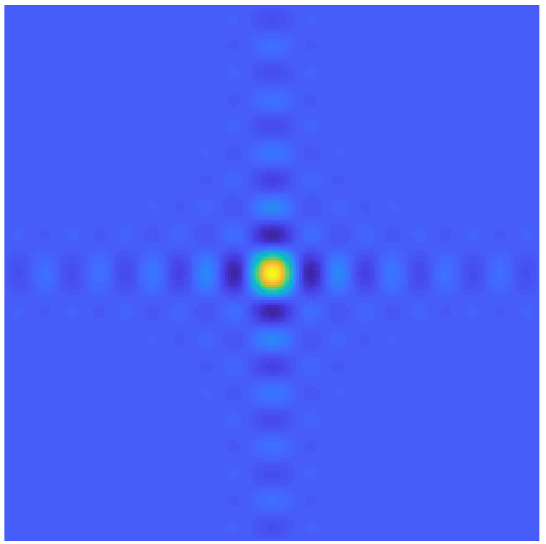}}\hspace{2em}
	\subfloat[Plot of $\gamma$]{\includegraphics[width=0.2\textwidth]{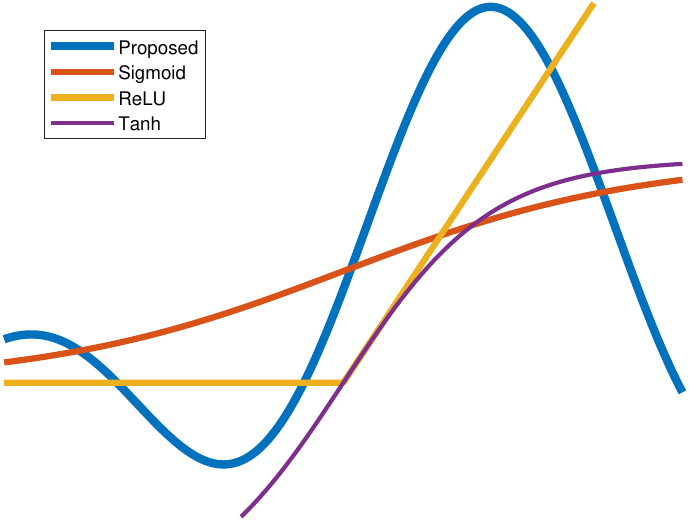}}
	\caption{Visualization of kernels in $\mathbb R^2$ and the non-linear function $\gamma$ with some commonly used activation functions.}
	\label{kerneltype}
\end{figure}

We note from the above figures that the circular Dirichlet kernel ($q=2$) is roughly circularly symmetric, unlike the triangular or diamond kernels. This implies that we can safely approximate it as 
\begin{equation}\label{key}
\kappa(\mathbf x,\mathbf y) \approx g(\|\mathbf x-\mathbf y\|^2)
\end{equation}
where $g:\mathbb R_{+} \rightarrow \mathbb R$. We note that this approximation results in significantly reduced computation in the multidimensional case. The function $g$ may be stored in a look-up table or computed analytically. We use this approach to speed up the computation of multidimensional functions in Section \ref{nn}.

An additional simplification is to assume that $\mathbf x$  and $\mathbf y$ are unit-norm vectors. In this case, we can approximate 
\begin{equation}\label{approx}
 g(||\mathbf{x}_i-\mathbf{y}_i||_2^2) =g(||\mathbf{x}_i||_2^2+||\mathbf{y}_i||_2^2-2\langle\mathbf{x}_i,\mathbf{y}_i\rangle)\approx g(2-2\langle\mathbf{x},\mathbf{y}\rangle)=:\gamma(\langle\mathbf{x},\mathbf{y}\rangle),
\end{equation}
where $\gamma(z) = g(1-z/2)$. Here, we term $\gamma$ as the activation function. While we do not make this simplifying assumption in our computations, it enables us to show the similarity of the computational structure of \eqref{fnrep} to current neural network. The plot of this activation function, along with commonly used activation functions, is shown in Figure \ref{kerneltype} (d).

With the aforementioned analysis, we can then rewrite \eqref{fnrep} as 
\begin{eqnarray}\label{compstruct}
\mathbf f(\mathbf{x}) &=& \underbrace{\left[\mathbf f_1,\ldots,\mathbf f_N\right]}_{\mathbf F} ~\mathcal K\left(\mathbf A\right)^{\dag} ~\underbrace{\begin{bmatrix}
	g(\|\mathbf x-\mathbf a_1\|^2)\\
	\vdots\\
	g(\|\mathbf x-\mathbf a_N\|^2)
	\end{bmatrix}}_{\mathbf k_{\mathbf A}(\mathbf x)}\\\label{nneq}
&\approx& \underbrace{\mathbf F~ \mathcal K(\mathbf A)^{\dag}}_{\widetilde{\mathbf F}}~\underbrace{\begin{bmatrix}
	\gamma\left(\left\langle\mathbf x,\mathbf a_1\right\rangle\right)\\
	\vdots\\
	\gamma\left(\left\langle\mathbf x,\mathbf a_N\right\rangle\right)
	\end{bmatrix}}_{\boldsymbol\Gamma_{\mathbf A}(\mathbf x )}
\end{eqnarray}
In the second step, we used the approximation in \eqref{approx}.

\begin{figure}[b!]
	\centering
	\subfloat[One layer network]{\includegraphics[width=0.4\textwidth]{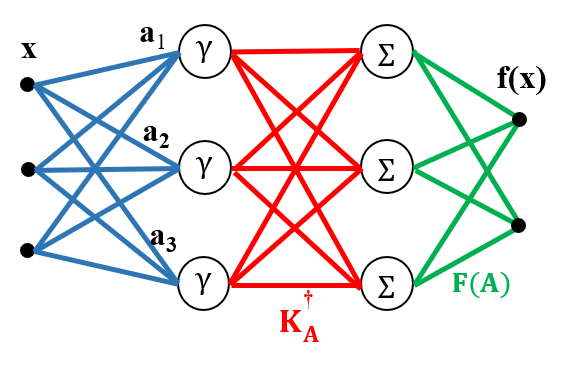}}~~~~~
	\subfloat[Two layers network]{\includegraphics[width=0.5\textwidth]{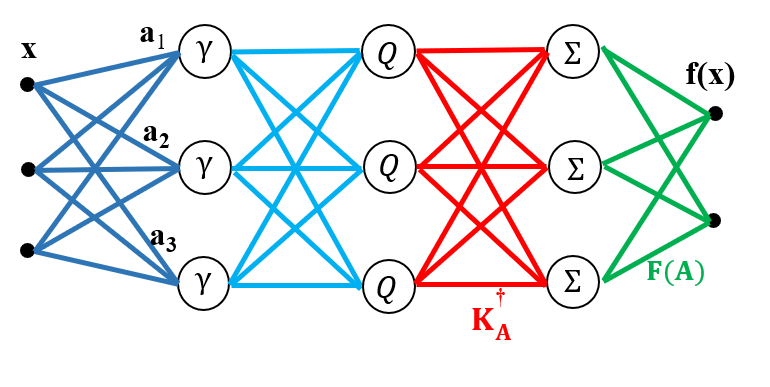}}
	\caption{Computational structure of function evaluation. (a) corresponds to \eqref{fnrep} to compute the band-limited multidimensional function $\mathbf f$ on $\mathcal S[\psi]$. The inner-product between the input vector $\mathbf x$ and the anchor templates on the surface are evaluated, followed by non-linear activation functions $\gamma$ to obtain the coefficients $\alpha_i(\mathbf x)$. These coefficients are operated with the fully connected linear layers $\mathbf K_{\mathbf A}^{\dag}$ and $\mathbf F(\mathbf A)$. The fully connected layers can be combined to obtain a single fully connected layer $\widetilde{\mathbf F}$. Note that this structure closely mimics a neural network with a single hidden layer.  (b) uses an additional quadratic layer, which combines functions of a lower bandwidth to obtain a function of a higher bandwidth.}
	\label{nnetwwork}
\end{figure}

\subsection{Optimization of the anchor points and coefficients}
	\label{optimsection}
	The above results show the existence of a computational structure of the form \eqref{nneq} with $N$ anchor points $\mathbf a_1,..,\mathbf a_N$ on the surface and the corresponding coefficients $\tilde{\mathbf f_1},..,\tilde{\mathbf f_N}$ that can represent the function exactly. We note that the anchor points need not to be selected as a subset of the training data. We note that Corollary \ref{corrollary} guarantees $\mathcal K(\mathbf A)$ to have full column rank as $N=r$. However, the condition number of this matrix may be poor, depending on the choice of the anchor points. It may be worthwhile to choose the anchors such that the condition number of $\mathcal K(\mathbf A)$ is low, which will reduce the noise amplification in \eqref{coeffs}.

We hence propose to solve for the anchor points $\mathbf A$ and the corresponding coefficients $\widetilde{\mathbf F}$ such that it minimizes the least square error evaluated on the training data:
	\begin{equation}\label{training}
	{\widetilde{\mathbf F}^*,\mathbf A^*} = \arg \min_{\widetilde{\mathbf F},\mathbf A} \sum_{i=1}^{N_{\rm train}}\|\widetilde{\mathbf F} ~\boldsymbol\Gamma_{\mathbf A}(\mathbf x_i ) - \mathbf y_i\|^2 
	\end{equation}
	We propose to minimize the above expression using stochastic gradient descent. This approach will allow the choice of the anchor points $\mathbf a_1,..,\mathbf a_N$.

\section{Relation to neural networks}
\label{nn}
We now briefly discuss the close relation of the proposed framework with neural networks. We consider the function learning setting, which is considered in Section \ref{funlearn} and show that the computational structure closely mimics a neural network with one hidden layer. We discuss briefly the benefits of depth in improving the representation. We also show that the above framework can be used to approximate the learning of a manifold from data, which can be viewed as a signal subspace alternative to the null-space approach considered in Section \ref{sampling}. We also show that the computational structure closely mimics an auto-encoder.

\subsection{Task/function learning from input output pairs }
We now focus on the learning of a function \eqref{compstruct} from training data pairs and will show its equivalence with neural networks. 
 Note that the computation involves the inner product of the input signal $\mathbf x$ with templates $\mathbf a_i; i = 1,..,N$, followed by the non-linear activation function $\gamma$ to obtain $\mathbf k_{\mathbf A}(\mathbf x)$. These terms are then weighted by the fully connected layer $\mathcal K\left(\mathbf A\right)^{\dag} $, followed by weighting by the second fully connected layer $\widetilde{\mathbf F}$. See Fig. \ref{nnetwwork} for the visual illustration.

As noted above, the representation using anchor points to reduce the degrees of freedom significantly compared to the direct representation. However, we note that the number of parameters needed to represent a high bandwidth function in high dimensions is still high. We now provide some intuition on how the low-rank tensor approximation of functions and composition can explain the benefit of common operations in deep networks.

We now consider the case when the band-limited multidimensional function $f:\mathbb R^n\rightarrow \mathbb R$ in \eqref{blfn} can be approximated as 
\begin{equation}\label{blfactor}
f(\mathbf x) = \left(\sum w_i~f_i(\mathbf x)\right)^2.
\end{equation}
Clearly, the bandwidth of $f$ is almost twice that of $f_i:\mathbb R^n\rightarrow \mathbb R$, showing the benefit of adding layers. While an arbitrary function with the same bandwidth as $f$ cannot be represented as in \eqref{blfactor}, one may be able to approximate it closely. The new layer will have a quadratic non-linearity $Q$, if the function has the form \eqref{blfactor}. Note that one may use arbitrary non-linearity in place of the quadratic one in \eqref{blfactor}. 

Similarly, one may perform a low-rank tensor approximation of an arbitrary $N$ dimensional function $f:\mathbb R^n\rightarrow \mathbb R$. Specifically, the approximation involves the sum of products of 1-D functions.
\begin{equation}\label{lr}
f(x_1,..,x_N) \approx \sum_{i=1}^{r} h_1^{(i)}(x_1)\cdot ~h_2^{(i)}(x_2)\ldots ~h_N^{(i)}(x_N),
\end{equation}
where $h_i:\mathbb R\rightarrow \mathbb R$. The above sum of products can also be realized by taking weighted linear combination of 1-D functions, followed by a non-linearity as in \eqref{blfactor}. This allows one to have a hierarchical structure, where lower dimensional functions are pooled together to represent a multidimensional function. 

In image processing applications, the functions to be learned are shift-invariant. This allows one to learn functions of small image patches (e.g. $3\times3$) of a specified dimension at each layer. The functions on nearby pixels in the output thus correspond to information from different $3\times3$ neighborhoods. The low-dimensional functions from non-overlapping $3\times3$ neighborhoods could be combined with downsampling as in \eqref{lr} to represent a high dimensional function (e.g. $9\times9$) neighborhoods. The process can be repeated to improve the efficiency of representation.

\subsection{Relation to auto-encoders }

We note that the space of band-limited functions of the form \eqref{blfn} can reasonably approximate lower order polynomials in $\mathbb R^n$ for sufficiently high bandwidth $\Gamma$ \cite{sorevik2016trigonometric}. In particular, let us assume that there exists a set of coefficients $\boldsymbol \beta$ such that
\begin{equation}\label{poly}
\mathbf x \approx \tilde{\mathbf x}= \sum_{\mathbf k\in \Gamma} \beta_{\mathbf{k}} \exp(j2\pi \mathbf k^T \mathbf x)
\end{equation}
In this case, the above results imply that one can represent any point on the surface $\mathcal S[\psi]$ as 
\begin{eqnarray}\label{auto}
\mathbf x&\approx& \underbrace{\left[\mathbf a_1,..,\mathbf a_n\right]}_{\mathbf A} ~\underbrace{\mathcal K\left(\mathbf A\right)^{\dag} ~\mathbf k_{\mathbf A}(\mathbf x)}_{\boldsymbol{\alpha}(\mathbf x)}
\end{eqnarray}
We note that the resulting network is hence essentially an auto-encoder. Specifically, the inner-products between the feature vectors of $\mathbf x$ and the anchor point $\mathbf a_i$ denoted by $\mathbf \alpha(\mathbf x) $ can be viewed as the latent features or compact code. As described previously, the coefficients $\boldsymbol \alpha = \mathcal K(\mathbf A)^{\dag} \mathbf{k_A}(\mathbf x)$ captures the geometry of the surface, while the top layer $\mathbf A$ is the decoder that recover the signal from its latent vectors. 

We note that the surface recovery algorithms in Section \ref{sampling} follow a null-space approach, where we identify the null-space of the feature space or equivalently the annihilation functions from the samples of the surface. Specifically, the sum of squares of the null-space functions in Section \ref{sos} provides a measure of the error in projecting the feature vector to the null-space of the feature matrix. 
\begin{eqnarray}\label{nullspaceprior}
\gamma(\mathbf{x}) &=& \sum_{i=1}^{|\Gamma\ominus\Lambda|}|\mu_i(\mathbf{x})|^2 =\sum_{i=1}^{|\Gamma\ominus\Lambda|}|\mathbf n_i^{T} \Phi_{\Gamma}(\mathbf{x})|^2\\
&=&  \|\mathbf N~ \Phi_{\Gamma}(\mathbf x)\|^2
\end{eqnarray}
where $\mathbf n_{i}$ are the null-space vectors. The projection energy is zero if  the point $\mathbf x$ is on $\mathcal S$ and is high when it is far from it. 

By contrast, the auto-encoder approach can be viewed as a signal subspace approach, where we project the samples to the basis vectors specified by the feature vectors of the anchors $\Phi_{\Gamma}(\mathbf a_{i})$. Specifically, we use the non-linearity specified by $\eqref{approx}$ and trained the network parameters ($\mathbf A$ as well as the weights of the inner-products) using stochastic gradient descent. The training data corresponds to randomly drawn points on the surface. To ensure that the network learns a projection, we trained the network as a denoising auto-encoder; the inputs correspond to samples on the surface corrupted with Gaussian noise, while the labels are the true samples. Once the training is complete, we plot the approximation error 
\begin{equation}
E(\mathbf x) = \|\mathbf x - \mathbf F \mathcal K(\mathbf A)^{\dag} \mathbf k_{\mathbf A}(\mathbf x)\|^2 = \|\underbrace{\left(\mathbf I - \mathbf F \mathcal K(\mathbf A)^{\dag} \mathbf k_{\mathbf A}\right)}_{\mathcal R}(\mathbf x)\|^2
\end{equation}
as a function of the input point in Fig. \ref{signalspace}.

We trained the network using the exemplar curve shown in Fig. \ref{prop_over}. We randomly choose 1000 points on the curve as the training data and 250 features are chosen in the middle layer. The bandwidth of the Dirichlet kernel is chosen to be $15$. The trained network is then used to learn the curve. The learned results are shown in Fig. \ref{signalspace}. From which one can see that the proposed learning framework performs well. We note that the projection error is close to zero on the surface, while it is high if it is away from the surface. Note that this closely mimics the plot in Fig. \ref{prop_over}. Once trained, the surface can be estimated in low-dimensional settings as the zero set of the projection error as shown in Fig. \ref{signalspace}.(b), which closely approximates the true curve in (c). We note that $\mathcal R$ can be viewed as a residual denoising auto-encoder. Once trained, this network can be used as a prior in inverse problems as in \cite{aggarwal2018modl}, where we have used the null-space network in Section \ref{nn}. We have also used the null-space prior \eqref{nullspaceprior} in our prior work \cite{poddar2019manifold}, where the null-space basis was learned as described in Section \ref{noisy}.

\begin{figure}[!h]
	\centering
	\subfloat[Learned curve]{\includegraphics[width=0.24\textwidth]{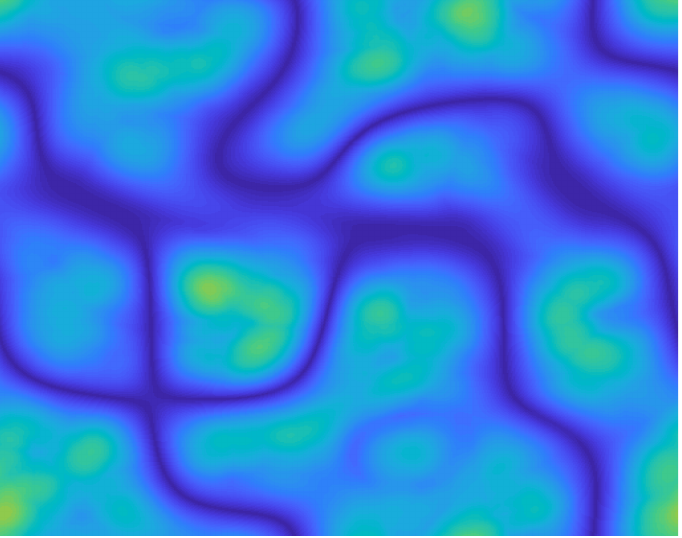}}\hspace{3em}
	\subfloat[Contour line]{\includegraphics[width=0.24\textwidth]{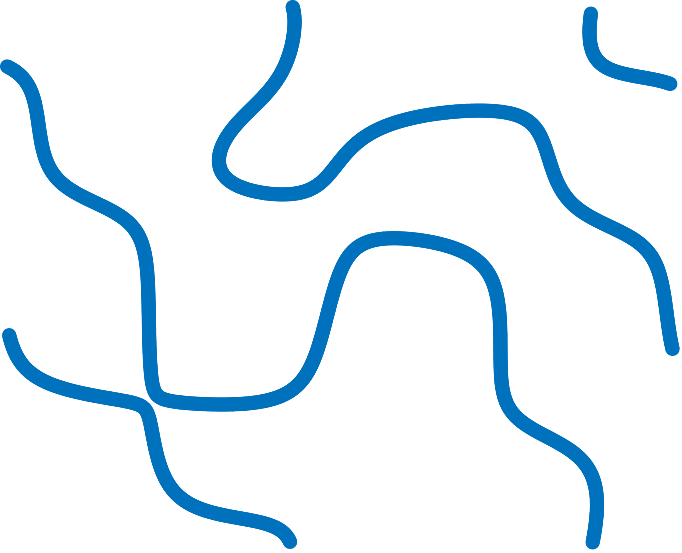}}\hspace{3em}
	\subfloat[Original curve]{\includegraphics[width=0.24\textwidth]{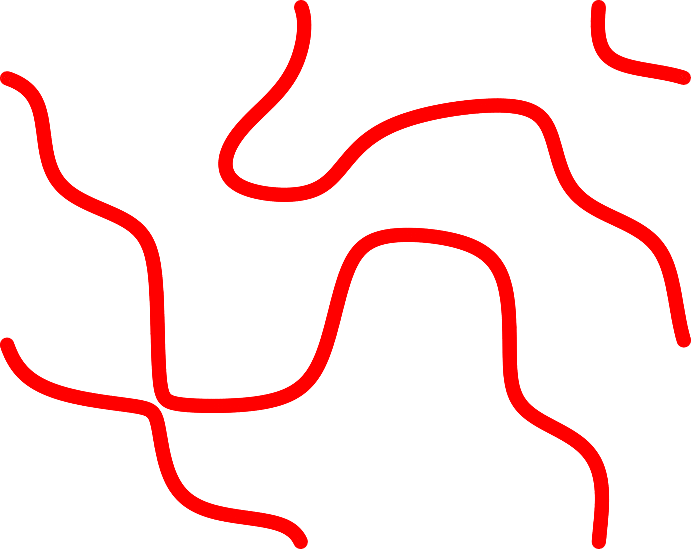}}
	\caption{Illustration of the surface learning network using the curve in Fig. \ref{prop_over}. (a) and (b) are the learned results. We compared the learned curve (blue curve) with the original curve (red curve) in (c). From which we see that the two curves are almost the same, indicating that the learned network performs well.}
	\label{signalspace}
\end{figure}

\section{Illustration in denoising}

We now illustrate the preliminary utility of the proposed network in image denoising. Specifically, we consider the learning of a function $f:\mathbb R^{p^2}\rightarrow \mathbb R$, which predicts the denoised center pixel of a patch from the noisy $p\times p$ patch. The function $f$ in $p^2$ dimensional space is associated with a large number of free parameters; learning of these unknowns are challenging due to the curse of dimensionality. Then the result in the previous section offers a work-around, which suggests that the function can be expressed as the linear combination of the features of ``anchor-patches'', weighted by $\mathbf{p}$.

We propose to learn the anchor patches $\mathbf a_i$ and the function values $f(\mathbf a_i)$ from exemplar data using stochastic gradient descent to minimize \eqref{training}. Note that the learned representation is valid for any patch, and hence the proposed scheme is essentially a convolutional neural network. The difference of our structure in \eqref{nneq} with the commonly used convolutional neural networks (CNN) structure is the activation function $\gamma$. We replaced the ReLU non-linearity in a network with the proposed function $\gamma$ in a single layer network. For the two-layer network, we replaced the ReLU non-linearity with $\gamma$ and $Q$ as indicated in \eqref{blfactor}.

We first tested the performance of the network on the MNIST dataset \cite{lecun2010mnist}. In the experiments, we choose the patch size to be $7\times 7$ and $d=7$ in \eqref{dirich}. We also trained a ReLU network with the same parameters for comparison. Besides, we compared the proposed scheme against non-local means (NLM) and dictionary learning (DL) \cite{elad2006image}. All algorithms, except for NLM were trained using the MNIST training set provided in TensorFlow. For the proposed network and the ReLU network, they are trained using 300 epoches and for the dictionary learning method, 500 iterations are used to learn the dictionaries. The comparison of the testing results is shown in Figure \ref{denoisecomp}. The comparison of the PSNR is reported in the caption.  The results show that the neural network based approaches offer improved performance compared to dictionary learning and non-local methods. Our results also show that the proposed networks provide comparable, if not slightly better performance, compared to the ReLU networks. The results also show the slight improvement in performance offered by the proposed two-layer networks over single layer networks. 

\begin{figure}[h!]
	\centering
	\includegraphics[width=0.7\textwidth]{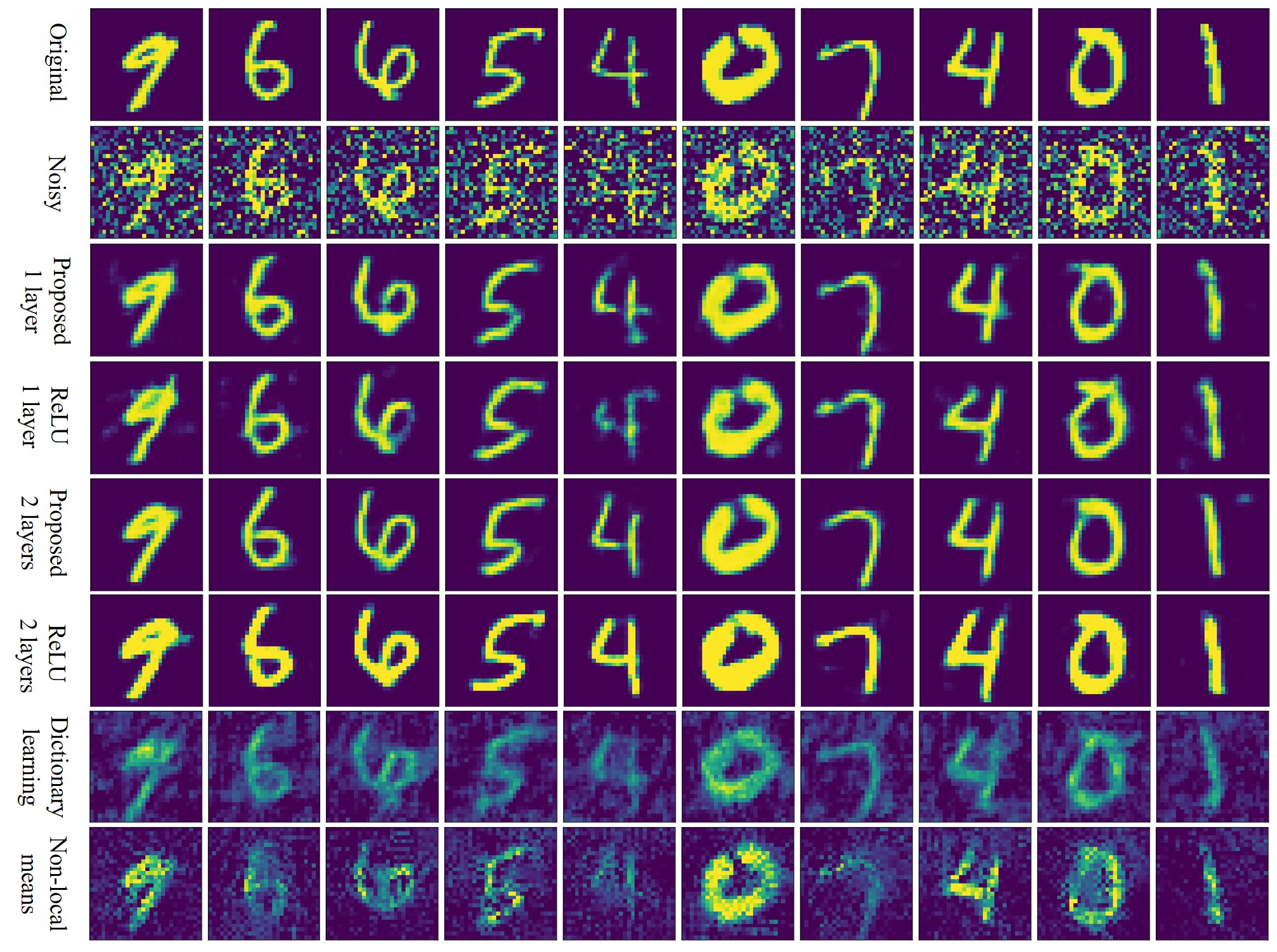}
	\caption{Comparison of our learned denoiser using the proposed activation function and the ReLU activation function. The testing results show that the denoising performance using the proposed activation function is comparable to the performance using ReLU. The eight rows  in the figure correspond to the original images, the noisy images, the denoised images using the proposed one-layer network, the denoised images using one layer ReLU network, the denoised images using the proposed two-layer network, the denoised images using two-layer ReLU network, the denoised images using dictionary learning and  the denoised images using non-local means. The averaged PSNR of the denoised images using the proposed one-layer network, one layer ReLU network, proposed two-layer network, two-layer ReLU network, dictionary learning and non-local means are $19.68$ dB, $20.03$ dB, $20.86$ dB, $17.48$ dB, $14.76$ dB and $14.28$ dB respectively. From the quantitative results, we can see that our proposed one-layer network performs comparable to the one-layer ReLU network. For the proposed two-layer network, the performance is getting better from both quantitative and visual points of view. For the two-layer ReLU network, visually the performance is better than that of the one-layer ReLU network. But the PSNR is getting worse. The main reason that causes the low PSNR for the two-layer ReLU network is the change of the pixel values on each hand-written digit.}
	\label{denoisecomp}
\end{figure}

The size of the image in the MNIST dataset is small. To better demonstrate the performance of the proposed network, we also applied the proposed scheme to the denoising of natural images. The algorithm was trained on the images of Hill, Cameraman, Couple, Bridge, Barbara and Boat at three different noise settings. We assume the noise is Gaussian white noise in the natural images setting. We compared the proposed scheme against dictionary learning (DL), non-local means (NLM) and transform learning (TL) \cite{ravishankar2013learning}. In the experiments for natural images, the patch size is chosen as $9\times 9$ and $d=7$ in \eqref{dirich}. For the proposed network and the ReLU network, they are trained using 300, 400, 450 epoches corresponding to the noise level $\sigma = 10, 20, 100$, and for the dictionary learning method, 500 iterations are used to learn the dictionaries. We then tested the denoising performance on two natural images: Man and Lighthouse.

The quantitative results (PSNR) of the algorithm are shown in Table \ref{snr}, while the results on Man and Lighthouse with noise of standard deviation $\sigma=20$ are shown in Fig. \ref{denoisecomp2} and Fig. \ref{denoisecomp3}. In Table \ref{snr}, Fig. \ref{denoisecomp2} and Fig. \ref{denoisecomp3}, ``ReLU1'' and ``ReLU2'' represent one-layer ReLU network and two-layer ReLU network, while ``Proposed1'' and ``Proposed2'' stand for the proposed one-layer network and proposed two-layer network. The results show that the performance of the neural network schemes is superior to classical methods and the proposed networks provide comparable or slightly better performance than the ReLU networks.

\begin{table}[!h]
	\begin{center}
		\begin{tabular}{|c || c | c | c | c | c | c | c | c | }
			\hline
			Img. & $\sigma$ & DL & NLM & TL & ReLU1 & ReLU2 & Proposed1 & Proposed2 \\
			\hline \hline
			& 10 & 26.63 & 26.64 & 27.41 & 30.29 & {31.11} & 30.99 & \bf{31.19} \\
			Man & 20 & 26.11 &26.35 & 27.02 & 27.47 & 27.33 & 27.25 & \bf{27.63} \\
			& 100 & 19.69 & 20.95 & 21.65 & 21.85 & \bf{22.11} & 21.91 & {22.06} \\
			\hline
			& 10 & 27.08 & 29.08 & 28.71 & 28.88 & 29.27 & 30.05 & \bf{30.28} \\
			Lighthouse & 20 & 25.51 &25.21 & 25.92 & 26.25 & 26.33 & 26.69 & \bf{26.74} \\
			& 100 & 19.14 & 20.14 & 20.15 & 20.21 & 20.46 & 20.35 & \bf{20.47} \\
			\hline
		\end{tabular}
		\caption{The PSNR (dB) of the denoised results for the two testing natural images with different noise level.}
		\label{snr}
	\end{center}
\end{table}
\begin{figure}[!h]
	\centering
	\subfloat[Original]{\includegraphics[width=0.15\textwidth]{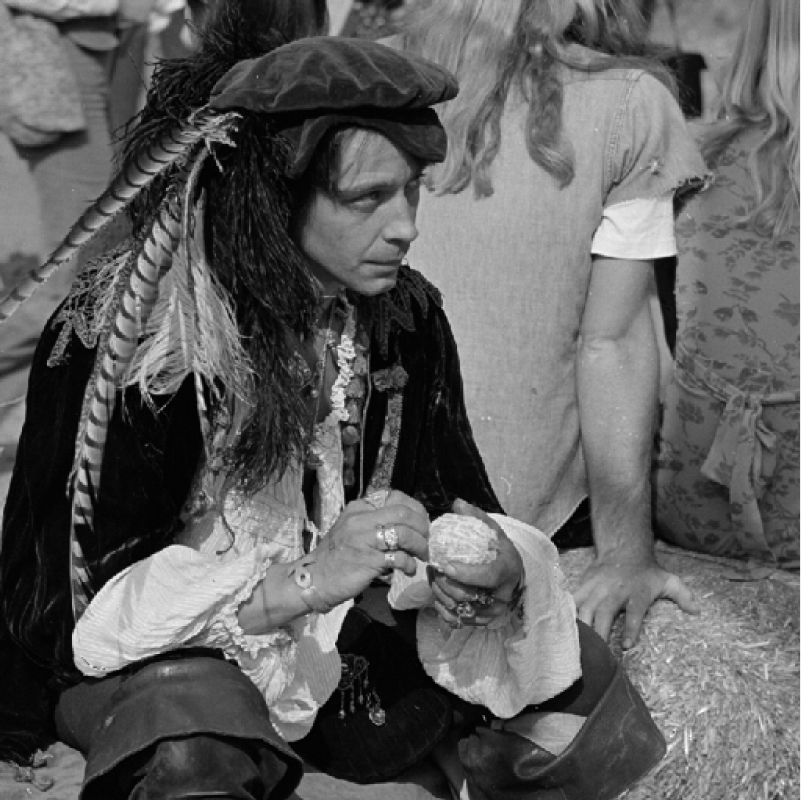}}
	\subfloat[Noisy 22.11 dB][Noisy \\ 22.11 dB]{\includegraphics[width=0.15\textwidth]{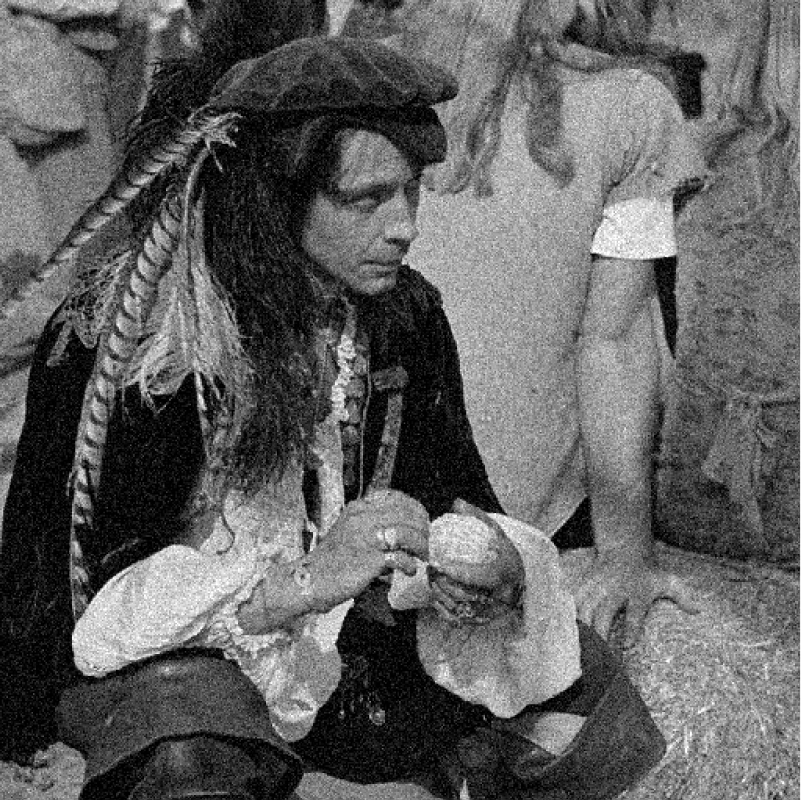}}
	\subfloat[DL 26.11 dB][DL \\ 26.11 dB]{\includegraphics[width=0.15\textwidth]{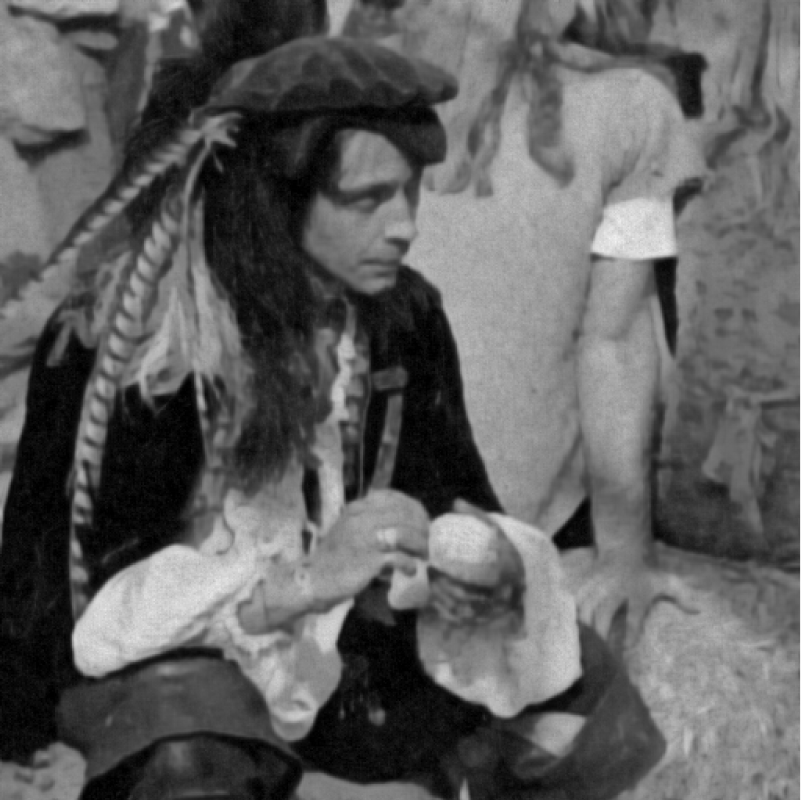}}
	\subfloat[NLM 26.35 dB][NLM \\ 26.35 dB]{\includegraphics[width=0.15\textwidth]{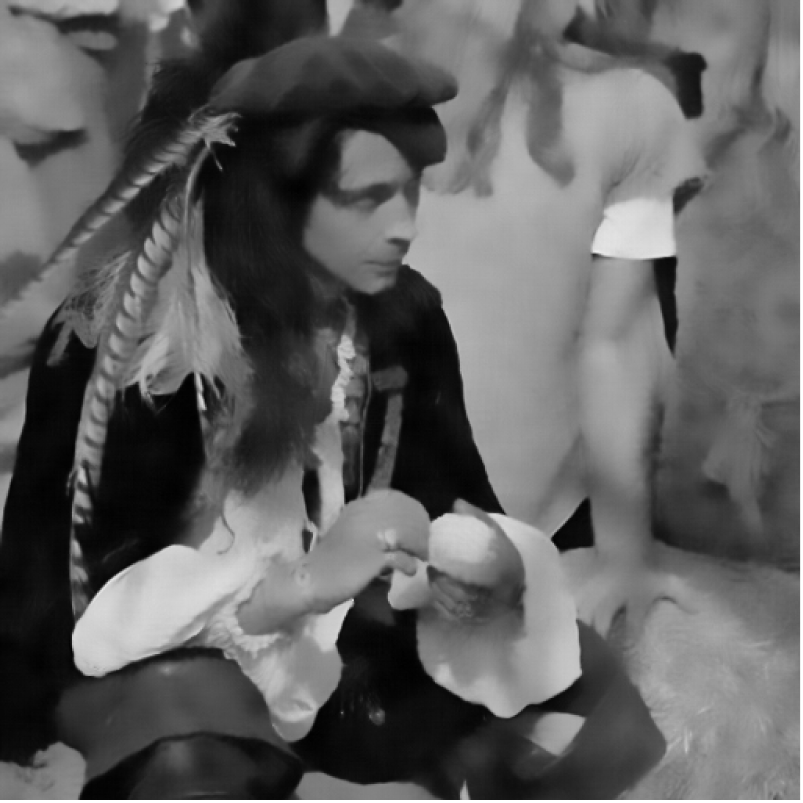}}
	\subfloat[TL 27.02 dB][TL \\ 27.02 dB]{\includegraphics[width=0.15\textwidth]{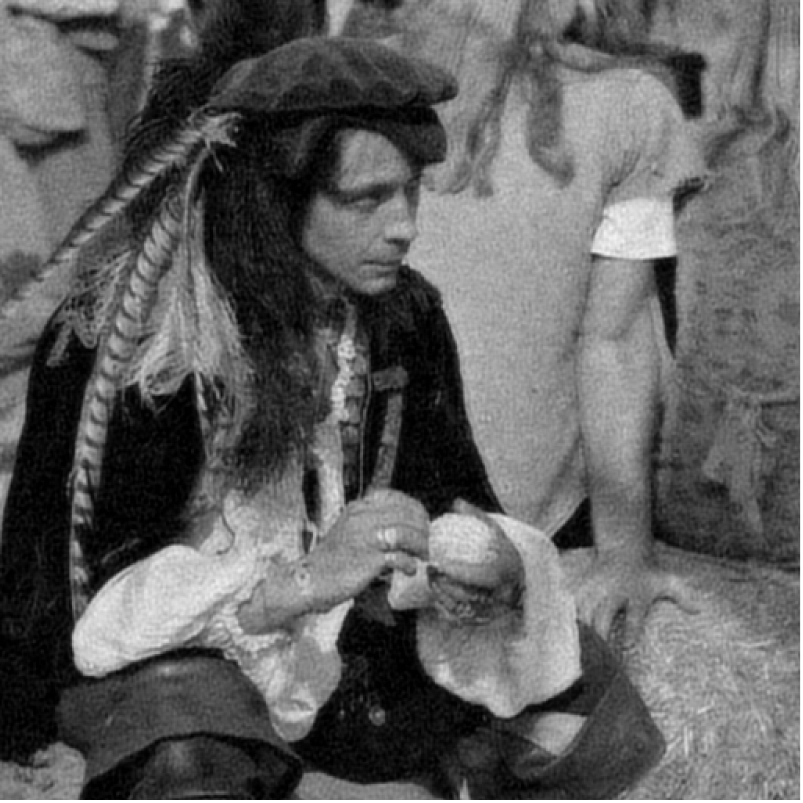}}
	\subfloat[ReLU2 27.33 dB][ReLU2 \\ 27.33 dB]{\includegraphics[width=0.15\textwidth]{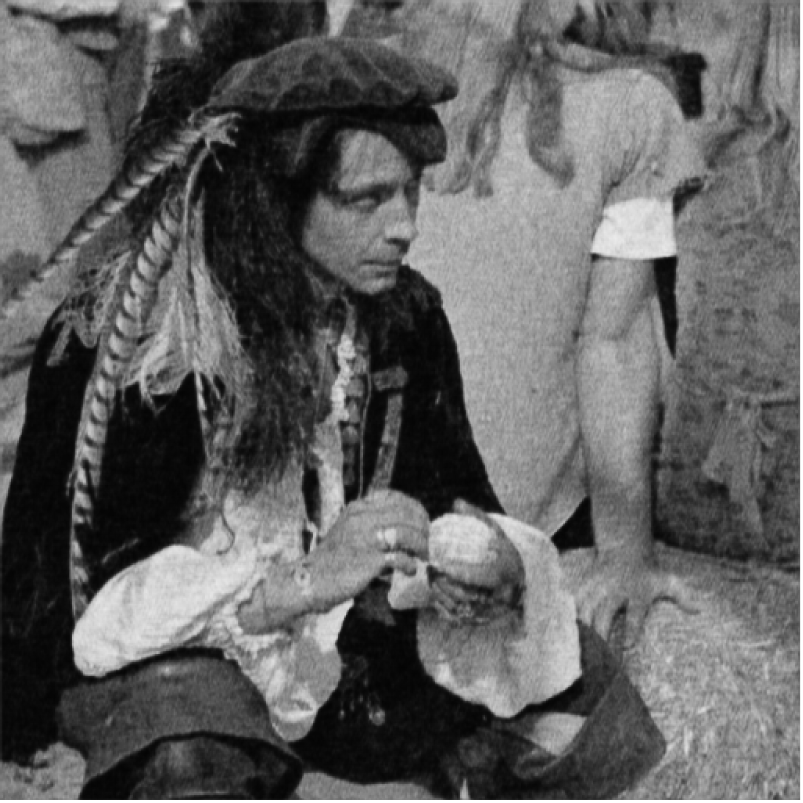}}
	\subfloat[Proposed2 27.63 dB][Proposed2 \\ 27.63 dB]{\includegraphics[width=0.15\textwidth]{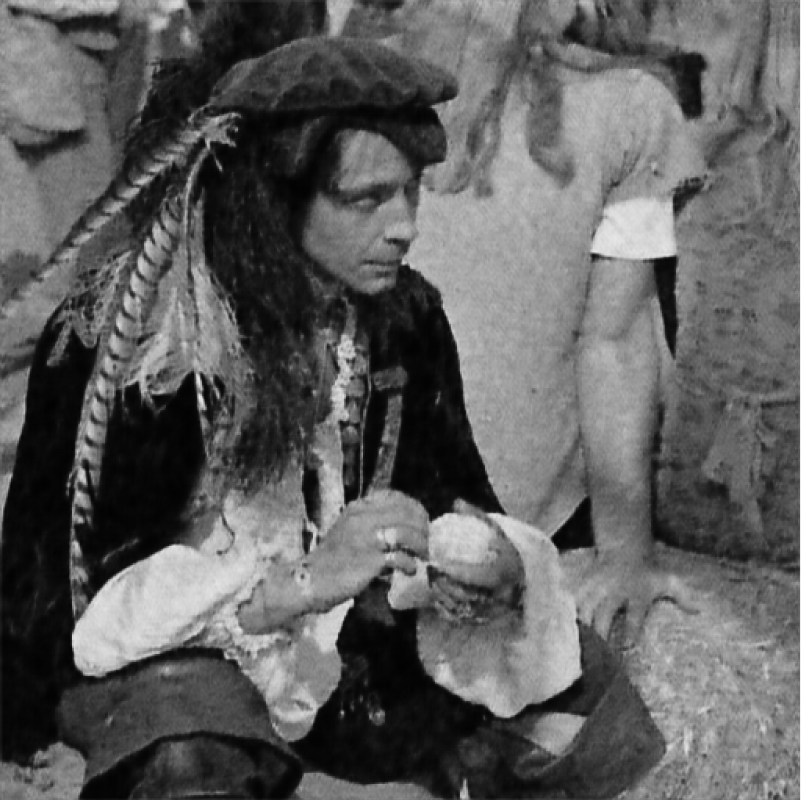}}\\
	\subfloat[Original]{\includegraphics[width=0.15\textwidth,trim=50 160 210 40,clip]{man_orig}}
	\subfloat[Noisy]{\includegraphics[width=0.15\textwidth,trim=50 160 210 40,clip]{man_20_noisy}}
	\subfloat[DL]{\includegraphics[width=0.15\textwidth,trim=50 160 210 40,clip]{man_20_dl}}
	\subfloat[NLM]{\includegraphics[width=0.15\textwidth,trim=50 160 210 40,clip]{man_20_nl}}
	\subfloat[TL]{\includegraphics[width=0.15\textwidth,trim=50 160 210 40,clip]{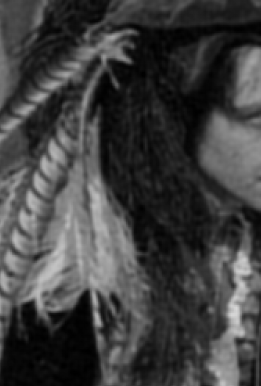}}
	\subfloat[ReLU2]{\includegraphics[width=0.15\textwidth,trim=50 160 210 40,clip]{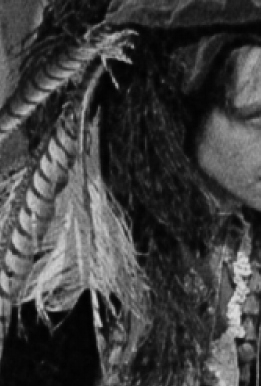}}
	\subfloat[Proposed2]{\includegraphics[width=0.15\textwidth,trim=50 160 210 40,clip]{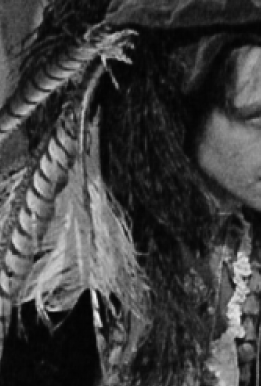}}\\
	\caption{Comparison of the proposed denoising algorithms on the image ``Man'' with $\sigma=20$.}
	\label{denoisecomp2}
\end{figure}
\begin{figure}[!h]
	\centering
	\subfloat[Original]{\includegraphics[width=0.15\textwidth]{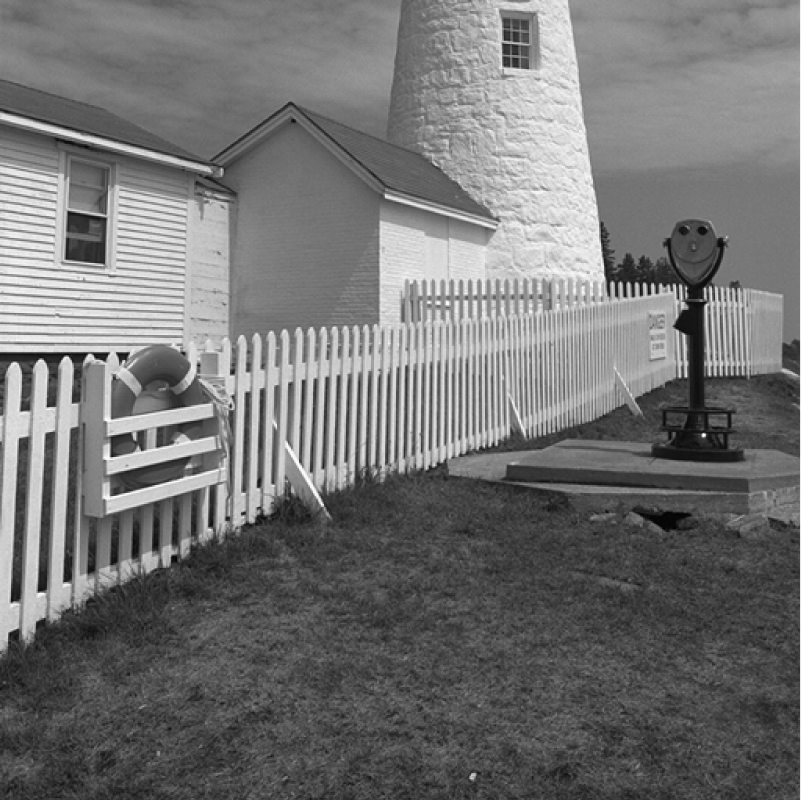}}
	\subfloat[Noisy 22.09 dB][Noisy \\ 22.12 dB]{\includegraphics[width=0.15\textwidth]{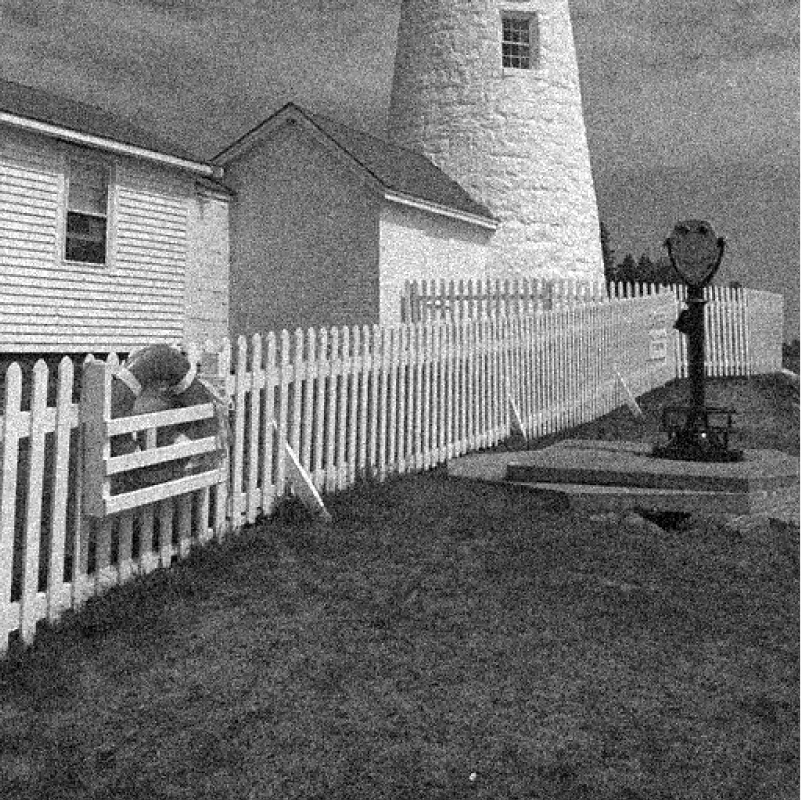}}
	\subfloat[DL 25.51 dB][DL \\ 25.51 dB]{\includegraphics[width=0.15\textwidth]{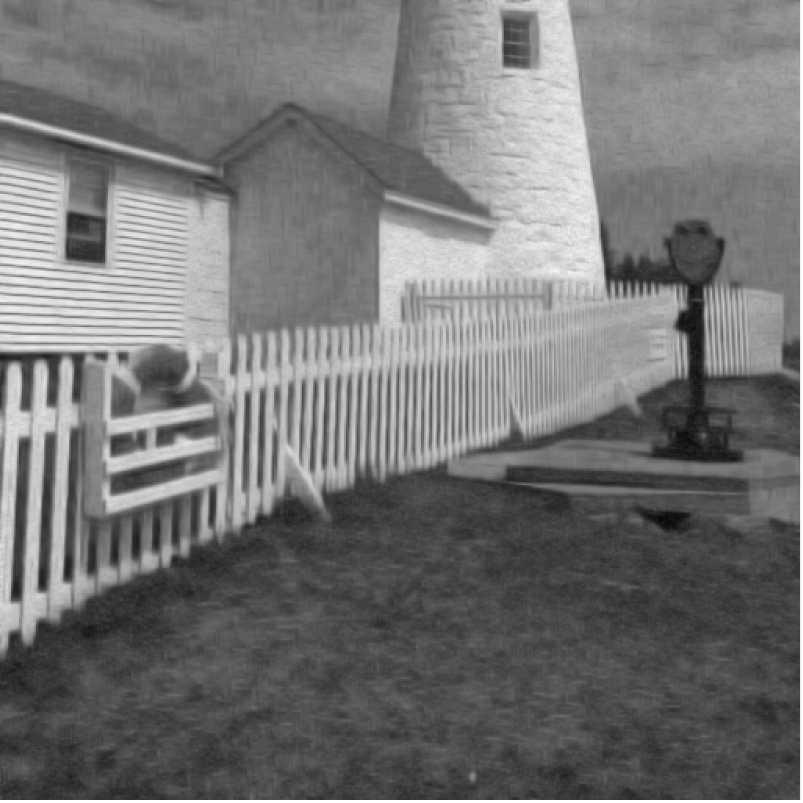}}
	\subfloat[NLM 25.21 dB][NLM \\ 25.21 dB]{\includegraphics[width=0.15\textwidth]{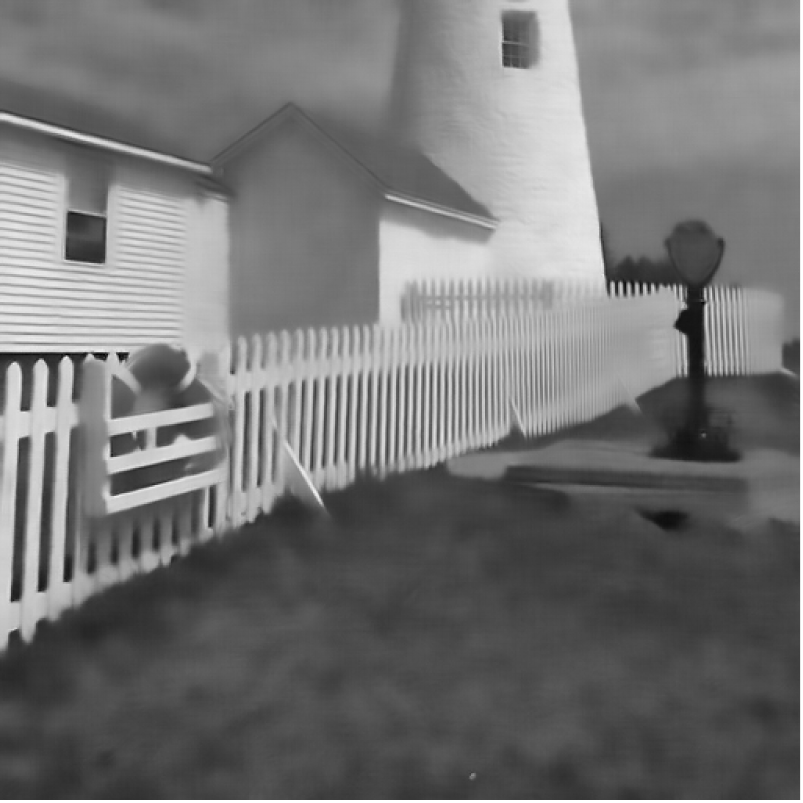}}
	\subfloat[TL 25.92 dB][TL \\ 25.92 dB]{\includegraphics[width=0.15\textwidth]{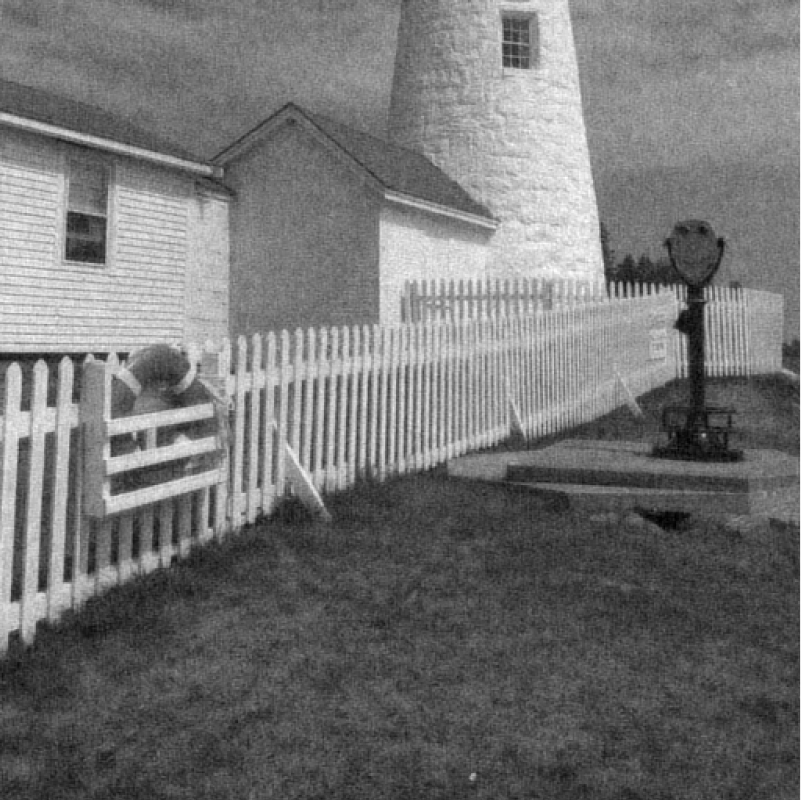}}
	\subfloat[ReLU2 26.33 dB][ReLU2 \\ 26.33 dB]{\includegraphics[width=0.15\textwidth]{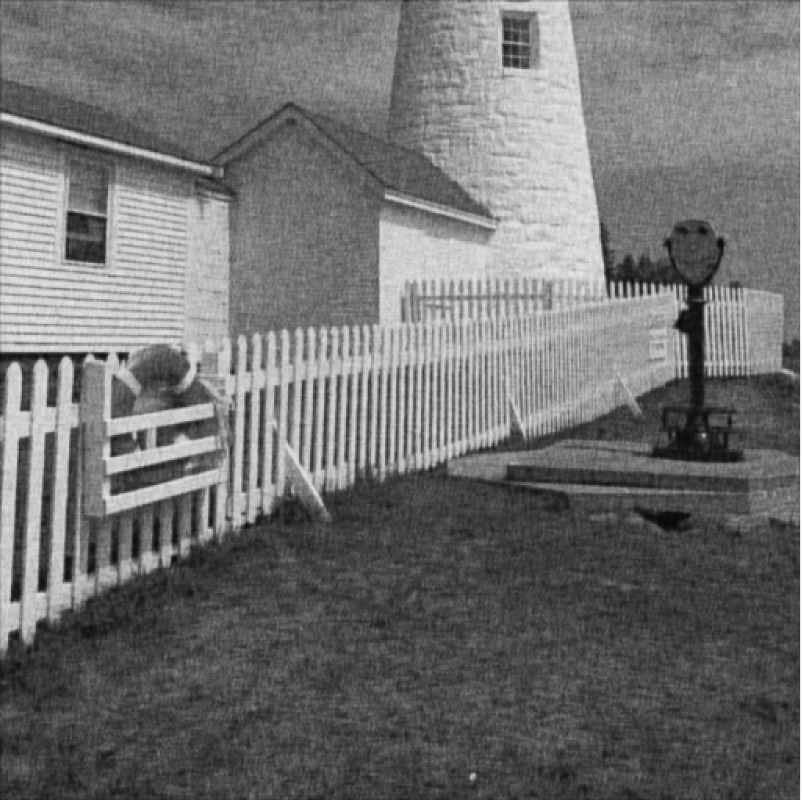}}
	\subfloat[Proposed2 26.74 dB][Proposed2 \\ 26.74 dB]{\includegraphics[width=0.15\textwidth]{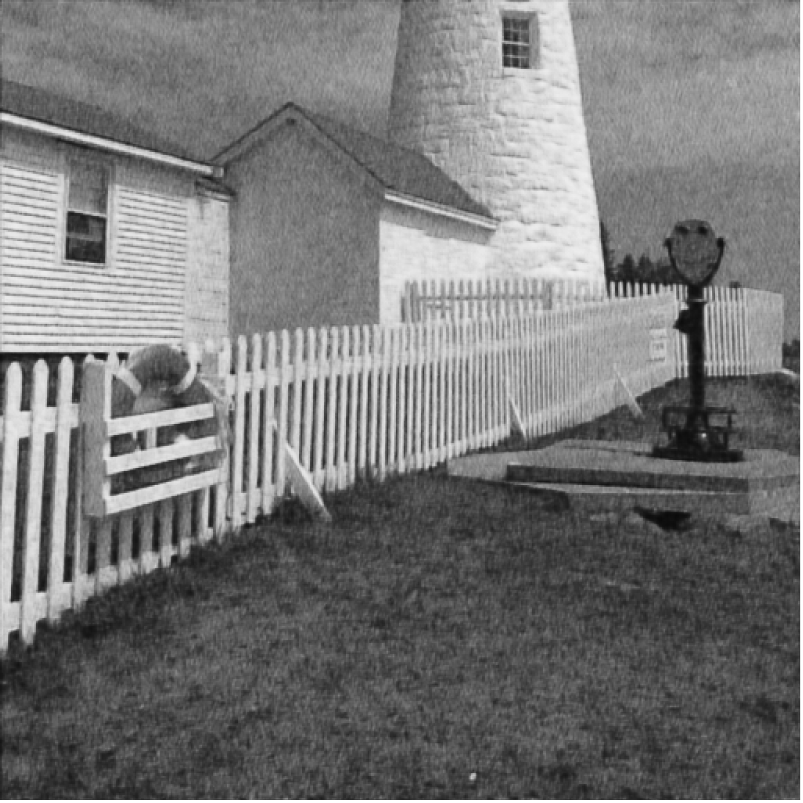}}\\
	\subfloat[Original]{\includegraphics[width=0.15\textwidth,trim=150 80 150 180,clip]{lighthouse_20_orig}}
	\subfloat[Noisy]{\includegraphics[width=0.15\textwidth,trim=150 80 150 180,clip]{lighthouse_20_noisy}}
	\subfloat[DL]{\includegraphics[width=0.15\textwidth,trim=150 80 150 180,clip]{lighthouse_20_dl}}
	\subfloat[NLM]{\includegraphics[width=0.15\textwidth,trim=150 80 150 180,clip]{lighthouse_20_nl}}
	\subfloat[TL]{\includegraphics[width=0.15\textwidth,trim=150 80 150 180,clip]{lighthouse_20_trans}}
	\subfloat[ReLU 2]{\includegraphics[width=0.15\textwidth,trim=150 80 150 180,clip]{lighthouse_20_r2}}
	\subfloat[Proposed 2]{\includegraphics[width=0.15\textwidth,trim=150 80 150 180,clip]{lighthouse_20_p2}}
	\caption{Comparison of the proposed denoising algorithms on the image ``Lighthouse'' with $\sigma=20$.}
	\label{denoisecomp3}
\end{figure}

\section{Conclusion}

In this work, we considered a data model, where the signals are localized to a surface that is the zero level set of a band-limited function $\psi$. The bandwidth of the function can be seen as a complexity measure of the surface. We show that the non-linear features of the samples, obtained by an exponential lifting, satisfy an annihilation relation. Using the annihilation relation, we developed theoretical sampling guarantees for the unique recovery of the surface. Our main contribution here is to prove that with probability 1, the surface can be uniquely recovered using a collection of samples, whose number is equal to the degrees of freedom of the representation. When the true bandwidth of the surface is unknown, which is usually the case, we introduced a method using the SoS polynomial to specify the surface. We also introduced the way to get back the samples when the original samples are corrupted by noise.

We then use this model to efficiently represent arbitrary band-limited functions $f$ living on the surface. We show that the exponential features of the points on the surface live in a low-dimensional subspace. This subspace structure is used to represent the $f$ efficiently using very few parameters. We note that the computational structure of the function evaluation mimics a single-layer neural network. We applied the proposed computational structure to the context of image denoising.

\section{Appendix}

\subsection{Proof of Proposition \ref{uniqueminimal}}\label{minimalsection}

As we mentioned in Section \ref{pp1section}, if we have a (hyper-)surface $\mathcal{S}$ which is given by the zero level set of a trigonometric polynomial, then there will be a minimal polynomial which defines $\mathcal{S}$ (Proposition \ref{uniqueminimal}). To prove this result, we need the following famous result.

\begin{lem}[Hilbert's Nullstellensatz \cite{atiyah1994introduction}]
Let $\mathbb K$ be an algebraically closed field (for example $\mathbb{C}$). Suppose $I \subset \mathbb K[x_{1},\cdots, x_n]$ is an ideal of polynomials, and $\mathcal{Z}(I)$ denotes the set of common zeros of all the polynomials in $I$. Let $\mathcal{I}(\mathcal{Z}(I))$ represents the ideal of polynomials in $\mathbb{K}[x_1,\cdots, x_n]$ vanishing on $\mathcal{Z}(I)$. Then, we have
\[\mathcal{I}(\mathcal{Z}(I)) = \sqrt{I},\]
where $\sqrt I$ denotes the radical of $I$, specified by the set 
\begin{equation}\label{radical}
\sqrt I = :\{p | p^{n} \in I, ~\mbox{for some}~  n\in\mathbb{Z}^+\}
\end{equation}
\end{lem}

\noindent {\it{Remark 1.}}\quad We say a set $I \subset K[x_1,\cdots, x_n]$ is an ideal, if $I$ is closed under the addition operation (e.g. addition``$+$''), satisfies the associative property, has a unit element $0$, and a valid inverse for every element in $I$. For the operation multiplication (e.g. ``$\cdot$''), we have $r\cdot p\in I$ and $p\cdot r\in I$ for any $r\in K[x_1,\cdots, x_n]$

\noindent{\it{Remark 2.}}\quad An important property of the radical of the ideal $I$ is that $I\subset \sqrt{I}$. Note that setting $n=1$ in \eqref{radical} will yield $I$.

\noindent{\it{Remark 3.}}\quad The above lemma states that the set of all polynomials that vanish on the common zeros $\mathcal Z(I)$ of the polynomials in $I$ is given by $\sqrt I \supset I$. Specifically, if we are given another polynomial $\eta(\mathbf x)$ that also vanishes on the common zero set $\mathcal{Z}(I)$, then there must be positive integer $n$ such that $\eta^n(\mathbf x) \in I$.

We denote the ideal generated by a function $f$ by $(f) = \{\mu| \mu=f\gamma\}$, where $\gamma$ is an arbitrary polynomial. The identity in this ideal is the zero polynomial. In particular, $(f)$ is the family of all functions that have $f$ as a factor. We note that the set of common zeros of all the functions in $(f)$, denoted by $\mathcal Z[(f)]$  is the same as the zero set of $f$, denoted by $Z[f]$.


\begin{lem}\label{equalpoly}
Let $f,g$ be two polynomials in $\mathbb{C}[x_1,\cdots, x_n]$ with the same zero set. Then the two polynomials must have (up to scaling) the same factors.
\end{lem}

\begin{proof}
Suppose $Z[f]=Z[g]= Z$ is the zero set of $f$ and $g$. Since $Z[f]=\mathcal Z[(f)]$, we have $\mathcal{Z}[(f)]= \mathcal{Z}[(g)] = Z$. By the Hilbert's Nullstellensatz, we have
\[\mathcal{I}(\mathcal{Z}(f)) = \sqrt{(f)},\qquad \mathcal{I}(\mathcal{Z}(g)) = \sqrt{(g)}.\]
Since $Z(f) = Z(g)$, we then have $\mathcal{I}(\mathcal{Z}(f)) = \mathcal{I}(\mathcal{Z}(g))$ and hence $\sqrt{(f)} = \sqrt{(g)}$. As mentioned above, we have $I\subset\sqrt{I}$ for any ideal $I$. Therefore, we have $(f) \subset \sqrt{(f)}$ and $(g)\subset\sqrt{(g)}$. This implies that $f\in \sqrt{(f)}$ and $g\in \sqrt{(g)}$. Because we have $\sqrt{(f)} = \sqrt{(g)}$, we can obtain that $f\in\sqrt{(g)}$ and $g\in\sqrt{(f)}$. By which we have that there exist $m,n\in\mathbb{Z}$ and $p,q\in\mathbb{C}[x_1,\cdots, x_n]$ such that
\[f^n = p\cdot g,\qquad g^m = q\cdot f.\]
Therefore, we can obtain that the irreducible factors of $g$ are of $f$ as well and vice versa, which proves the desired conclusion.
\end{proof}

With this conclusion, we can now prove Proposition \ref{uniqueminimal}.

\begin{proof}[Proof of Proposition \ref{uniqueminimal}]
The proof of the existence and uniqueness about $\psi$ is same as the proof of Proposition A.3 in \cite{ongie} and thus we omit them here.

In this proof, we show that $BW(\psi)\subseteq BW(\psi_1)$. Note that the algebraic surface $X = \{p=\mathcal{P}[\psi]=0\}$ is the union of irreducible surfaces $X_j = \{p_{i_j} = 0\}\subset\mathbb{C}^n$. Define 
\[\nu(x_1,\cdots, x_n) = (e^{j2\pi x_1},\cdots, e^{j2\pi x_n}).\]
Let $\mathcal{S}_j = \nu^{-1}(X_j\cap \mathbb{T}^n)$. Then we have a decomposition of $\mathcal{S}$ as the union of surfaces $\mathcal{S}_j$. If $\psi_1$ is another trigonometric polynomial with $\mathcal{S}$ as the zero level set as well. Then $\psi_1$ vanishes on each $\mathcal{S}_j$. Let $q = \mathcal{P}[\psi_1]$. Then we have $q= 0$ on the infinite set $\nu(\mathcal{S}_j)$, by which we can infer that $q$ and $p$ will have the same zero set using Theorem \ref{intersect}. Then by Lemma \ref{equalpoly}, we have $p \mid q$, which implies that $BW(\psi)\subseteq BW(\psi_1)$.
\end{proof}

\subsection{Proof of results in Section \ref{sampling}}\label{proofrankprop}

The key property of surfaces that we exploit is that the dimension of the intersection of two band-limited surfaces of dimension $k$ is strictly lower than $k$, provided their level set functions do not have any common factors. Hence, if we randomly sample one of the surfaces, the probability that the samples fall on the intersection of the two surfaces is zero. This result enables us to come up with the sampling guarantees. We will now show the results about the intersections of the zero sets of two trigonometric surfaces.

\subsubsection{Intersection of surfaces}
We will first state a known result about the intersection of the zero sets of two polynomials (non-trigonometric) whose level set functions do not have a common factor. 

\begin{thm}[\cite{gunning2009analytic},pp.115, Theorem 14]\label{intersect}
Let $\mathcal S[\psi]$ and $\mathcal{S}[\eta]$ be two surfaces of dimension $n-1$ over a field $\mathbb{K}$, which are the zero sets of the polynomials $\psi:\mathbb K^n\rightarrow \mathbb K$ and $\eta:\mathbb K^n\rightarrow \mathbb K$, respectively. If $\psi$ and $\eta$ do not have a common factor, then
\[\dim\big(\mathcal S[\psi]\cap \mathcal S[\eta]\big) <n-1.\]
\end{thm}


The above result is a generalization of the two dimensional case ($\mathbb{C}^2$) in \cite{ongie}, where  B\'ezout's inequality was used to prove the result. 
Specifically, the result in \cite{ongie} suggests that the intersection of two curves consists of a set of isolated points, if their potential function does not have any common factor. Theorem \ref{intersect} generalizes the above result to $n>2$; it suggests that the intersection of two surfaces with dimension $k$ is another surface, whose dimension is strictly less than $k$. For instance, the intersection of two 3-D surfaces which are given by the zero level set of some polynomials, could yield 2D curves or isolated points. We now extend Theorem \ref{intersect} to trigonometric polynomials using the mapping $\nu$ specified by \eqref{map}.

\begin{lem}\label{dimlemma}
Let $\mathcal S[\psi]$ and $\mathcal{S}[\eta]$ within $[0,1]^n\subset\mathbb{R}^n$ be two surfaces of dimension $n-1$ over $\mathbb{R}$, which are the zero level sets of the trigonometric polynomials $\psi$ and $\eta$. Suppose $\psi$ and $\eta$ do not have a common factor, then
\[\dim(\mathcal{S}[\psi]\cap\mathcal{S}[\eta])< n-1.\]
\end{lem}

\begin{proof}
Let $\nu=(\nu_1,\cdots,\nu_n)$ be defined by \eqref{map}. We now would like to prove the result by way of contradiction. Suppose 
\[\dim(\mathcal{S}[\psi]\cap\mathcal{S}[\eta]) = \dim(\mathcal{S}[\psi]) = \dim(\mathcal{S}[\eta]) = n-1.\]
This implies that $\nu(\mathcal{S}[\psi]\cap\mathcal{S}[\eta])$ will have the same dimension of $\nu(\mathcal{S}[\psi])$ and $\nu(\mathcal{S}[\eta])$. However, this is impossible according to Theorem \ref{intersect}. Therefore, we have the desired result.
\end{proof}


Based on this lemma, we can directly have the following Corollary.

\begin{cor}\label{measure}
Suppose $\psi(\mathbf{x}), \eta(\mathbf{x}), \mathbf{x}\in[0,1]^n$ are two trigonometric polynomials as in Lemma \ref{dimlemma}. Consider the $n-1$ dimensional Lebesgue measure on $\mathcal{S}[\psi]$. Then this Lebesgue measure of the intersection of the zero level sets of the trigonometric polynomials is zero, i.e.,
\[m(\mathcal{S}[\psi]\cap\mathcal{S}[\eta])=0.\]
\end{cor}
The Lebesgue measure can be viewed as the area of the $n-1$ dimensional surface. For example, when $n=3$, $\mathcal S[\psi]$ and  $\mathcal S[\eta]$ are 2-D surfaces, while their intersection is a 1-D curve or a set of isolated points with zero area. 

\subsubsection{Proof of  Proposition \ref{rankprop}}
\label{rankpropproof}
\begin{proof}
We note that $N\geq |\Lambda|-1$ is a necessary condition for the matrix to have a rank of $|\Lambda|-1$. We now assume that the surface is sampled with $N\geq |\Lambda|-1$ random samples, chosen independently, denoted by  $\mathbf{x}_i; i=1,\cdots,N \in\mathcal{S}[\psi]$. Since $\mathbf c \leftrightarrow \psi$ is a valid non-trivial null-space vector for the feature matrix $\Phi_{\Lambda}(\mathbf X)$ formed from these samples, we have ${\rm rank}\left(\Phi_{\Lambda}(\mathbf X)\right) \leq  |\Lambda|-1$. The polynomial $\psi(\mathbf{x}) = \mathbf{c}^T\Phi_{\Lambda}(\mathbf{x})$ is the minimal irreducible polynomial that defines the surface.  

We now prove the desired result by contradiction. Assume that these exists another linearly independent null-space vector $\mathbf{d} \leftrightarrow \eta$ or equivalently the rank of $\Phi_{\Lambda}(\mathbf X)$ is strictly less than $|\Lambda|-1$. Since $\mathbf{c}$ and $\mathbf{d}$ are linearly independent and $\psi(\mathbf{x})$ is the minimal polynomial, we know that $\psi(\mathbf{x})$ and $\eta(\mathbf{x})$ will not share a common factor.  Also note that $\mathbf{x}_i\in\mathcal{S}[\psi]\cap\mathcal{S}[\eta]$. However, since $\psi(\mathbf{x})$ and $\eta(\mathbf{x})$ do not share a common factor, the probability of each sample to be at the intersection of the two polynomials $(\mathbf{x}_i\in\mathcal{S}[\psi]\cap\mathcal{S}[\eta])$ is zero by Corollary \ref{measure}. Therefore, with probability 1 that such $\mathbf{d}$ does not exist, meaning that with probability 1 that the feature matrix will be of rank $|\Lambda|-1$ when $N \ge |\Lambda|-1$.
\end{proof}

\subsubsection{Proof of Proposition \ref{rankprop2}}\label{proofrankprop2}

\begin{proof}We note that $N\geq |\Lambda|-1$ is a necessary condition for the matrix to have a rank of $|\Lambda|-1$. We now assume that the surface is sampled with $N$ random samples $\mathbf{x}_i; i=1,\cdots, N $ satisfying the conditions in Proposition \ref{rankprop2}. The minimal polynomial $\psi(\mathbf{x}) = \mathbf{c}^T\Phi_{\Lambda}(\mathbf{x})$ that defines the surface can be factorized as $\psi(\mathbf x) = \psi_{1}(\mathbf x)\cdot\psi_{2}(\mathbf x)\cdots\psi_{M}(\mathbf x)$.  

We will prove the result by contradiction. Assume that these exists another linearly independent null-space vector $\mathbf d \leftrightarrow \eta$, or equivalently the rank of $\Phi_{\Lambda}(\mathbf X)$ is less than $|\Lambda|-1$.  Since $\mathbf c$ and $\mathbf d$ are linearly independent, $\psi$ and $\eta$ should differ by at least one factor. Without loss of generality, let us assume that $\eta(\mathbf x) = \mu(\mathbf{x})\prod_{i=1}^{M-1} \psi_{i}(\mathbf{x})$, where $\mu$ is an arbitrary polynomial of bandwidth $\Lambda_{M}$. Besides, $\mu$ and $\psi_{M}$ does not share a factor. Using the result of Proposition \ref{rankprop}, we see that the probability of $\mu$ and an irreducible $\psi_{M}$ vanish at $|\Lambda_{i}|-1$ independently drawn random locations is zero. If multiple factors are shared, the same argument can be extended to each one of the factors independently. 
\end{proof}

\subsubsection{Proof of Proposition \ref{nonminimal1}}\label{proofrankprop3}
\begin{proof}
We note that $N\geq  |\Gamma|-|\Gamma\ominus\Lambda|$ is a necessary condition for the matrix to have the specified rank. We now assume that the surface is sampled with $N\geq  |\Gamma|-|\Gamma\ominus\Lambda|$ random samples, chosen independently. We note that $\mathbf c \leftrightarrow \psi$ specified by \eqref{tri1}, as well as the $|\Gamma\ominus\Lambda|$ translates of $\mathbf c$ within $\Gamma$, are valid linearly independent null-space vectors of $\Phi_{\Lambda}(\mathbf X)$. We thus have
\begin{equation}
{\rm rank}\left(\Phi_{\Lambda}(\mathbf X)\right) \leq  |\Gamma|-|\Gamma\ominus\Lambda|
\end{equation}

We will show that the rank condition can be satisfied with probability 1 by contradiction. Assume that these exists another linearly independent null-space vector $\mathbf{d} \leftrightarrow \eta$ or equivalently the rank of $\Phi_{\Lambda}(\mathbf X)$ is less than $ |\Gamma|-|\Gamma\ominus\Lambda|$. Since $\mathbf{d}$ are linearly independent with $\mathbf{c}$ and its translates within $\Gamma$, we cannot express $\mathbf{d}$ as the linear combinations of the the other null-space vectors. Specifically, we have 
\begin{eqnarray}
\eta(\mathbf x) &\neq& \sum_{\mathbf k \in \Gamma \ominus \Lambda} \alpha_{\mathbf k}~\psi(\mathbf x) \exp(j2\pi\mathbf k^{T} \mathbf x)\\
&=& \psi(\mathbf x) \underbrace{\sum_{\mathbf k \in \Gamma \ominus \Lambda} \alpha_{\mathbf k}\exp(j2\pi\mathbf k^{T} \mathbf x)}_{\gamma(\mathbf x)} = \psi(\mathbf x)\gamma(\mathbf x).
\end{eqnarray}
Here $\alpha_{\mathbf k}$ is an arbitrary coefficients and hence $\gamma$ is an arbitrary polynomial. The linear independence property implies that $\eta(\mathbf x)$ cannot have $\psi(\mathbf x)$ as a factor. Since $\psi(\mathbf x)$ is the minimal polynomial, this also means that $\eta$ and $\psi$ does not have any common factor. 

Consider now the random sampling set $\mathbf{x}_i; i=1..  |\Gamma|-|\Gamma\ominus\Lambda|$. We have
\[\mathbf{c}^T\Phi_{\Lambda}(\mathbf{x}_i) = \mathbf{d}^T\Phi_{\Lambda}(\mathbf{x}_i) = 0,\,\,\, i = 1,\cdots,|\Lambda|-1.\]
This implies that $\mathbf{x}_i\in\mathcal{S}[\psi]\cap\mathcal{S}[\eta]$. However, since $\psi(\mathbf{x})$ and $\eta(\mathbf{x})$ do not share a common factor, the probability of each sample to be at the intersection of the two polynomials $(\mathbf{x}_i\in\mathcal{S}[\psi]\cap\mathcal{S}[\eta])$ is zero by Corollary \ref{measure}. Therefore, we have ${\rm rank}\left(\Phi_{\Lambda}(\mathbf X)\right) =  |\Gamma|-|\Gamma\ominus\Lambda|$ with probability one. 

\end{proof}

\subsubsection{Proof of Proposition \ref{nonminimal2}}\label{proofrankprop4}

\begin{proof}We note that $N\geq  |\Gamma|-|\Gamma\ominus\Lambda|$ is a necessary condition for the matrix to have the specified rank. We now assume that the surface is sampled with $N$ random samples satisfying the sampling conditions in Proposition \ref{nonminimal2}. The minimal polynomial $\psi(\mathbf{x}) = \mathbf{c}^T\Phi_{\Lambda}(\mathbf{x})$ that defines the surface can be factorized as $\psi(\mathbf x) = \psi_{1}(\mathbf x)\cdot\psi_{2}(\mathbf x)\cdots\psi_{M}(\mathbf x)$.  

Assume that there exists another linearly independent null-space vector $\mathbf{d} \leftrightarrow \eta$ or equivalently the rank of $\Phi_{\Lambda}(\mathbf X)$ is less than $ |\Gamma|-|\Gamma\ominus\Lambda|$.  Similar to the above arguments, if $\eta$ and $\psi$ does not have any common factors, the rank condition is satisfied with probability 1. Similar to Section \ref{proofrankprop2},  linear independence implies that $\eta(\mathbf x)$ cannot be a factor of $\psi$; there is at least one factor $\psi_{i}$ that is distinct. Based on Proposition \ref{nonminimal1}, these factors cannot vanish on more than $|\Gamma_i|-|\Gamma_i\ominus\Lambda_{i}|$ common samples.
\end{proof}

\section*{Acknowledgments}

The authors would like to thank Dr. Greg Ongie for the valuable discussion about the proofs of the results in this paper. The first author would also like to thank Prof. Theodore Shifrin for his help during the discussion on the intersection of two surfaces and Mr. Biao Ma for the discussion on Lemma \ref{dimlemma} and Corollary \ref{measure}. 

\bibliographystyle{siam}
\bibliography{refs}

\end{document}